\documentclass[opre,nonblindrev]{informs3} 

\OneAndAHalfSpacedXI 

\usepackage{endnotes}
\let\footnote=\endnote
\let\enotesize=\normalsize
\def\notesname{Endnotes}%
\def\makeenmark{$^{\theenmark}$}
\def\enoteformat{\rightskip0pt\leftskip0pt\parindent=1.75em
  \leavevmode\llap{\theenmark.\enskip}}

\usepackage{natbib}
 \bibpunct[, ]{(}{)}{,}{a}{}{,}%
 \def\bibfont{\small}%
\TheoremsNumberedThrough     
\ECRepeatTheorems

\usepackage{pdflscape}
\usepackage{amsfonts}
\usepackage{setspace}
{\end{list}}

\def\VaR{{\rm VaR}}
\def\CVaR{{\rm CVaR}}

\def\Var{{\rm Var}}

\def\bA{{\bf A}}
\def\ba{{\bf a}}
\def\bB{{\bf B}}
\def\bb{{\bf b}}
\def\bC{{\bf C}}
\def\bc{{\bf c}}

\def\bF{{\bf F}}
\def\boldf{{\bf f}}

\def\bi{{\bf i}}
\def\bI{{\bf I}}
\def\bJ{{\bf J}}

\def\bK{{\bf K}}

\def\bM{{\bf M}}

\def\bP{{\bf P}}

\def\bR{{\bf R}}

\def\bS{{\bf S}}

\def\bT{{\bf T}}

\def\bU{{\bf U}}

\def\bV{{\bf V}}
\def\bv{{\bf v}}

\def\bX{{\bf X}}

\def\bY{{\bf Y}}

\def\bZ{{\bf Z}}
\def\bz{{\bf z}}
\def\betheta{\boldsymbol\theta}

\def\bealpha{\boldsymbol\alpha}
\def\balpha{\mbox{\boldmath $\alpha$}}

\def\bPhi{{\bf \Phi}}

\def\bdelta{\mbox{\boldmath $\delta$}}

\def\bGamma{{\bf \Gamma}}

\def\bLambda{{\bf \Lambda}}

\def\bmu{\mbox{\boldmath $\mu$}}

\def\bSigma{{\bf \Sigma}}

\def\bTheta{{\bf \Theta}}
\def\btheta{\mbox{\boldmath $\theta$}}

\def\bxi{\mbox{\boldmath $\xi$}}

\def\b#1{{\mathbf{#1}}}
\def\bzero{{\mathbf 0}}

\def\sfF{\mathsf{F}}

\def\blot{\quad {$\vcenter{\vbox{\hrule height.4pt
             \hbox{\vrule width.4pt height.9ex \kern.9ex \vrule
width.4pt}
             \hrule height.4pt}}$}}

\usepackage{amssymb,latexsym}
\usepackage{amsmath}
\usepackage{graphics}
\usepackage{amsmath,amsfonts,amssymb,amsbsy,dsfont}
\usepackage{graphicx}
\usepackage{epstopdf}
\usepackage{multirow}
\usepackage{pstricks}
\usepackage{enumerate}
\usepackage{float}
\usepackage{bbm}
\usepackage{url}
\usepackage{color}
\usepackage{natbib}
\usepackage{booktabs}
\usepackage{algorithm}
\usepackage{algorithmicx}
\usepackage{algpseudocode}

\usepackage[normalem]{ulem}

\usepackage{comment}

\newcommand{\algorithmiclastcon}{\textbf{Lastcon:}}
\newcommand{\lastcon}{\item[\algorithmiclastcon]}

\renewcommand{\algorithmiclastcon}{\textbf{Output:}}

\def\b1{{\mathbf{1}}}
\def\bzero{{\mathbf{0}}}

\EquationsNumberedThrough    

\renewcommand{\theARTICLETOP}{}
\begin{document}
\graphicspath{{figures/}}


\RUNAUTHOR{He et al.}

\RUNTITLE{Adaptive Importance Sampling}

\TITLE{Adaptive Importance Sampling for Efficient Stochastic Root Finding and Quantile Estimation}

\ARTICLEAUTHORS{%
\AUTHOR{Shengyi He}
\AFF{Department of Industrial Engineering \& Operations Research, Columbia University, New York, NY 10027, USA, \EMAIL{sh3972@columbia.edu}} 
\AUTHOR{Guangxin Jiang}
\AFF{School of Management, Harbin Institute of Technology, Harbin, Heilongjiang 150001, China,\\ \EMAIL{gxjiang@hit.edu.cn}}
\AUTHOR{Henry Lam}
\AFF{Department of Industrial Engineering \& Operations Research, Columbia University, New York, NY 10027, USA, \EMAIL{khl2114@columbia.edu}} 
\AUTHOR{Michael C. Fu}
\AFF{The Robert H. Smith School of Business, Institute for Systems Research, University of Maryland, College Park, MD 20742, USA, \EMAIL{mfu@umd.edu}} 
} 

\ABSTRACT{%
In solving simulation-based stochastic root-finding or optimization problems that involve rare events, such as in extreme quantile estimation, running crude Monte Carlo can be prohibitively inefficient. To address this issue, importance sampling can be employed to drive down the sampling error to a desirable level. However, selecting a good importance sampler requires knowledge of the solution to the problem at hand, which is the goal to begin with and thus forms a circular challenge. We investigate the use of adaptive importance sampling to untie this circularity. Our procedure sequentially updates the importance sampler to reach the optimal sampler and the optimal solution simultaneously, and can be embedded in both sample average approximation and stochastic approximation-type algorithms. 
Our theoretical analysis establishes strong consistency and asymptotic normality of the resulting estimators. We also demonstrate, via a minimax perspective, the key role of using adaptivity in controlling asymptotic errors.
Finally, we illustrate the effectiveness of our approach via numerical experiments.
}%



\maketitle

%
\section{Introduction}
A stochastic root-finding problem refers to the search of a solution $\btheta^*\in\mathbb{R}^d$ to an equation $\boldf(\btheta^*) =\bzero$, where the function $\boldf(\btheta)$ lacks analytical tractability and can only be accessed via noisy simulation. This problem is intimately related to stochastic optimization, where $\boldf$ is then  the gradient and we solve the first-order optimality conditions. Such problems are fundamental in many fields, including operations research and data science. Examples include quantile estimation where $\boldf$ involves the probability distribution function (\citealt{Wetherill1963}), continuous-space simulation optimization (\citealt{Fu2015book}), commonly used machine learning algorithms where model parameters are trained via empirical risk minimization (\citealt{BoCuNo2018}), and other applications such as the characterizations of convex risk measures as the roots of decreasing functions (\citealt{DunkelWeber2010}).

In this paper, we are interested in situations where the root-finding problem involves extremal or rare-event considerations. A primary example is extreme quantile estimation, in which the target probability level can be very close to 0 or 1 (e.g., $10^{-6}$). In this case, crude Monte Carlo can be prohibitively inefficient, as it takes roughly a sample size reciprocal to the target probability level to obtain meaningful statistical information. This phenomenon occurs generally in other examples involving rare events: As long as the solution depends crucially on samples in a region that is infrequently hit, the effectiveness of crude Monte Carlo could be substantially hampered.

To address estimation challenges related to rare events, importance sampling (IS) is commonly used (e.g., \citealt{Asmussen2007} Chapters 5 and 6; \citealt{PG} Chapter 4; \citealt{rubinstein2016simulation} Chapter 5). IS is a variance reduction technique that draws samples using a distribution distinct from the original (the importance sampler) that hits the rare-event region more often. At the same time, the estimator maintains unbiasedness via a multiplication of the sample output with the so-called likelihood ratio. If the IS distribution is carefully chosen, so that the hitting frequency and the likelihood ratio magnitude are properly controlled, then the estimation efficiency can be significantly boosted. Choosing and analyzing good IS schemes have been a focus in many studies (see, e.g., the surveys \citealt{bucklew2013introduction,juneja2006rare,blanchet2012state}).

Although there exists a rich literature on IS, there are relatively few theoretical results on applying IS to resolve the efficiency issues for crude Monte Carlo in extreme quantile estimation and other root-finding or stochastic optimization problems involving rare events. Most of the literature in IS focuses on the estimation of a target probability or expectation-type risk quantities, and the question is whether the same techniques can be easily adapted for root-finding. To this end, IS is known to be sensitive to the input choice: That is, let us suppose we already have found a good parametric IS distribution class, in the sense that there is a parameter value in the class that achieves high efficiency in the resulting sampler. If this parameter is wrongly chosen, the efficiency could be bad -- in some cases even worse than crude Monte Carlo. In other words, choosing good parameter values is crucial to the success of IS. On top of this, this value typically depends highly on the problem specification (e.g., \citealt{GylIg1989,sadowsky1991large,l2010asymptotic}).


To illustrate the above, suppose we want to estimate $P(Y>\gamma)$ for some high exceedance level $\gamma$ and model output $Y$. Consider an IS distribution class, say $\{Q_\alpha\}$ for $Y$ that is parameterized over $\alpha$. The efficiency of the resulting IS for estimating $P(Y>\gamma)$ could be highly sensitive to the choice of $\alpha$, which in turn depends on $\gamma$. That is, we can think of a good or an optimal $\alpha$ to be $I(\gamma)$ for a specific function $I(\cdot)$. We argue that this can cause significant challenges in root-finding, as the dependence of parameter choice on the target probability's specification will result in a circular choice in an inverse problem like root-finding. In this example, suppose we would like to use the IS distribution $Q_\alpha$ to improve the efficiency in estimating the extreme quantile of $Y$. Then in principle we would like to select $\alpha=I(q)$, where $q$ is the quantile. This, however, is clearly not obtainable, as it requires knowledge on the quantile $q$, which is what we want to determine in the first place.

Our main goal in this paper is to provide a mechanism to untie the above circularity. Specifically, given an efficient IS scheme designed to estimate a target probability or expectation, we offer a mechanism to convert this IS into one that is also  efficient for the inverse problem of root-finding (see Figure \ref{fig:AdaIS}). Our conversion mechanism does not require knowledge on how the initial IS works, i.e., the initial IS algorithm can be a black box, and the only thing we know is that it works well for the initial estimation problem.

\begin{figure}[H]
\centering
\includegraphics[width=16cm]{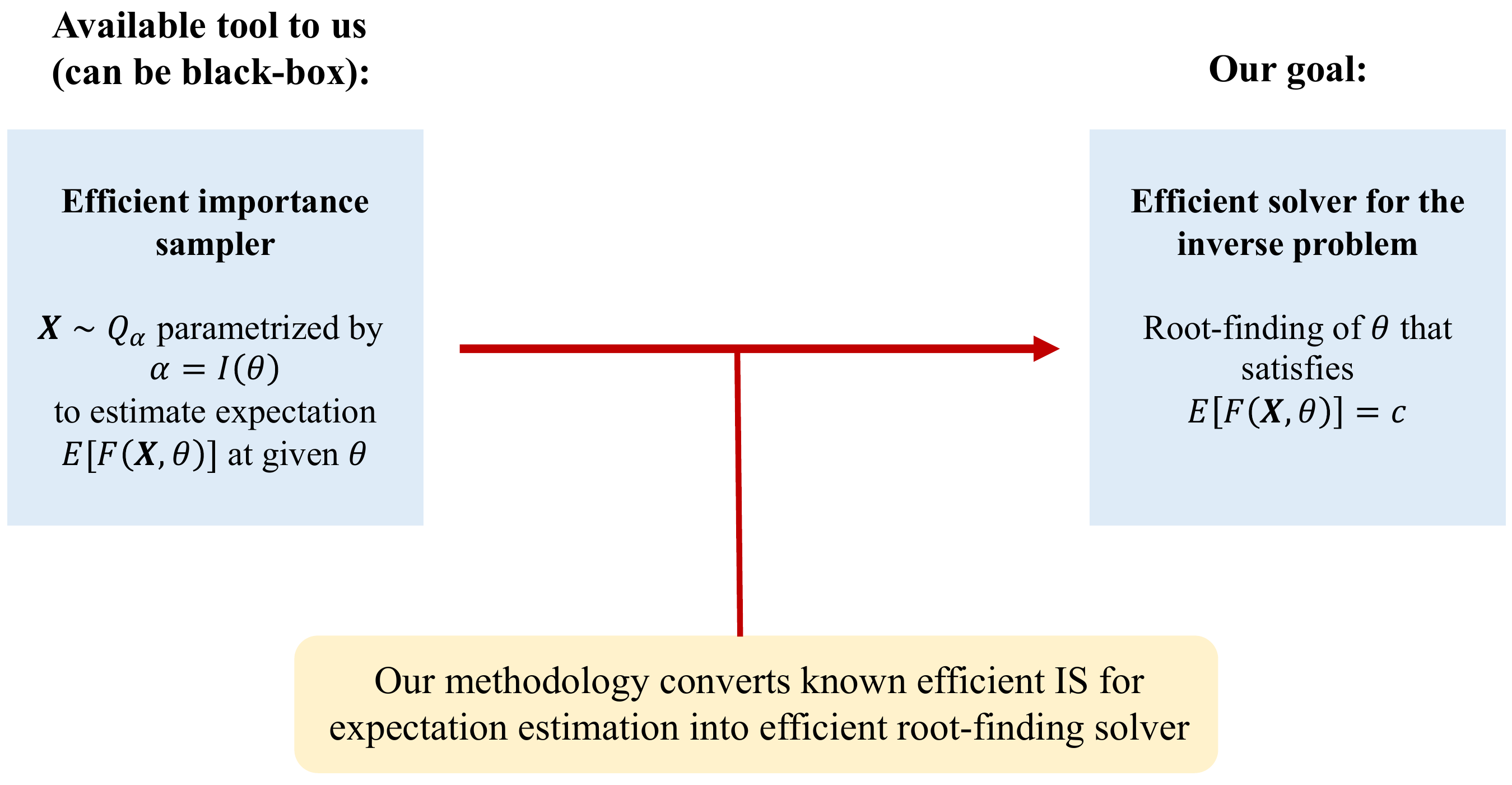}
\caption{Our goal and main usage of our methodology}
\label{fig:AdaIS}
\end{figure}

We propose an adaptive sampling approach that aims to  iteratively reach the root and update the IS parameter simultaneously. That is, at each iterative step we use the best myopic parameter value pretending that the target estimation problem is specified by the current root estimate. Using this IS, we generate new samples and update the root, which is then used to update the IS parameter again in a continuing manner. 

We support the necessity of the iterative approach above from a minimax perspective. If one does not allow iteration, then, since we do not know the root and hence the proper IS parameter, our root-finding procedure can have a large error in the worst case, as the parameter value that is used can misalign badly with the true root. On the other hand, we demonstrate that our iterative approach may achieve a substantial reduction of estimation error, regardless of the location of the root, by attaining a lower bound on the worst-case error dictated by a weak duality of the minimax error. Put in another way, this means that our approach exhibits the same asymptotic error \emph{as if we know the root in advance}, and thus is the best possible within the considered class of IS.

We discuss how to embed our adaptive IS into the two main numerical procedures for root-finding and stochastic optimization. First is sample average approximation (SAA), which replaces an unknown expectation with the empirical counterpart and applies deterministic solvers to locate the root (\citealt{Sha2003}). In quantile estimation, this corresponds to using the empirical quantile. The second main approach is stochastic approximation (SA), which can be viewed as a stochastic analog of the quasi-Newton method in deterministic optimization (\citealt{QRS2007}, \citealt{kushner2003stochastic}). SA iteratively updates the solution estimate by taking incremental steps to get closer to the true solution. In quantile estimation, this means adjusting the current quantile estimate by adding or subtracting a step depending on whether the new sample falls above or below the quantile estimate. Our adaptive IS can be embedded easily in both SAA and SA. In SAA, our adaptive approach leads to an iterative run of SAA programs, each with additional new samples. While this could be computationally intense for some problems, this approach is well suited for extreme quantile estimation, as each SAA corresponds simply to finding an empirical quantile. On the other hand, since SA is already an iterative approach, our IS can be naturally embedded at every iteration and does not cause extra computational complexity.

We investigate the consistency and central limit convergence of our adaptive IS embedded in both SAA and SA. These convergence theorems characterize the asymptotic behavior to support the superiority of our approach in performing at the same level as if we know the solution in advance, which is the best possible when using the considered IS class. Our technical developments require developing asymptotic normality results for a variant of SAA constructed from dependent data via martingale differences, and combining with functional complexity measures such as bracketing numbers, which are suitably adapted for our setting. Finally, we conduct experiments to validate our theoretical claims by comparing with other benchmarks. 

The rest of this paper is as follows. Section \ref{sec:lit review} first reviews related literature. Section \ref{sec:problem_set} formulates the stochastic root-finding problems and presents the challenges in applying standard IS via a worst-case analysis on the asymptotic variances.
Section \ref{sec: adaptiveIS} presents the main procedures in embedding our adaptive IS in SAA and SA, and provides  theoretical guarantees in single-dimensional settings.
Section \ref{sec:quantile} specializes to quantile estimation. Section \ref{sec:multidim} generalizes our framework to multi-dimensional problems. 
Section \ref{sec:example} applies our approach on numerical examples, including analyses to verify required assumptions. Section \ref{sec:concl} concludes our paper.
Proofs of all the theoretical results are provided in the Appendix.

\section{Related Work}\label{sec:lit review}
Our work is related to several others on quantile estimation that, similar to our approach, choose IS distributions adaptively in simulation rounds. 
\cite{morio2012extreme} proposes a nonparametric approach using a Gaussian kernel density to build the IS. 
\cite{egloff2010quantile} updates the IS parameter using an SA procedure similar to our proposed IS embedding, leading to a consistent quantile estimator where the parameter converges to the variance minimizer, but they do not analyze asymptotic variance. \cite{pan2020adaptive} considers adaptive IS for quantile estimation using a two-layer model where the inner layer is a black box, and are able to establish consistency and demonstrate variance reduction empirically by comparing with crude Monte Carlo.

The closest work to ours is \cite{bardou2009computing}, which considers estimation of Value-at-Risk (VaR) and conditional Value-at-Risk (CVaR), and also views the optimal IS parameter as a solution to a stochastic root-finding problem.
An iterative (SA) algorithm is proposed, for which they establish a central limit theorem (CLT) showing that this approach will lead to the smallest asymptotic variance among the chosen IS class. 
Our work considers a general framework to translate efficient IS from expectation estimation to stochastic root-finding, and generalizes \cite{bardou2009computing} in two main directions: 
(i) We can handle general parametrizations, whereas their work considers only location translations and exponential twists. (ii) We embed our adaptive IS in both SA and SAA, whereas they consider only the SA setting. In particular, our analysis of SAA is substantially more involved and requires developing new tools via empirical process theory.

We also briefly mention some variance reduction in quantile estimation using other adaptive methods. \cite{cannamela2008controlled} uses a reduced model to design the variance reduction scheme.
\cite{hu2008bootstrap} studies the use of adaptive IS in bootstrap quantile estimation. 
Other variance reduction approaches for quantile estimation that do not use adaptive methods include 
classical IS (\citealt{Glynn1996,SunHong2010}), control variates (\citealt{HsuNel1990, HesNel1998}), Latin hypercube sampling (\citealt{JinFuXiong2003, DongNaka2017}),
stratified sampling (\citealt{GlHeSh2000, ChuNa2012}), and splitting (\citealt{GuNeMa2011}). 

Work on using adaptive IS for expected value estimation is abundant, and we mention just a few here. 
\cite{au1999new} uses the Metropolis algorithm 
combined with kernel method to get an approximation to the optimal IS distribution. 
\cite{fu2002optimal} and  \cite{egloff2005optimal} find an optimal importance sampler using SA and then use this sampler to estimate the expectation. \cite{ryu2014adaptive} updates the IS sampler simultaneously with the expectation estimation. \cite{cornuet2012adaptive} adaptively chooses the IS density by fitting the moments and reweights all of the past samples in each step. \cite{KBCP99} and \cite{ABJ06} use adaptive IS along the time horizon to estimate the first passage of Markov chains and provide convergence guarantees. A comprehensive review of adaptive IS for expectation estimation can be found in \cite{bugallo2017adaptive}. 

Since our adaptive method is built on IS samplers for expectation estimation where rare events are involved, our work naturally relates to rare event simulation. \cite{bucklew2013introduction}, \cite{juneja2006rare} and \cite{blanchet2012state} provide surveys on the literature. Common approaches to design good IS for rare-event estimation utilize large deviations (\citealt{budhiraja2019analysis}) by scrutinizing the exponential twist in the rate function (\citealt{nicola2001techniques,dupuis2012importance,collamore2002importance,BlLiNa2019,BL14}), which also leads to the dominating point method (\citealt{sadowsky1990large,dieker2005asymptotically,owen2019importance}), subsolution approaches (\citealt{dupuis2009importance,blanchet2012lyapunov}), and mixture-based schemes that are especially useful in heavy-tailed problems (\citealt{blanchet2008efficient,blanchet2008state,chen2019efficient,blanchet2012efficient,murthy2015state,hult2012importance}). A more closely related approach to our study is the large body of work on the cross-entropy method (\citealt{rubinstein2004CE,rubinstein2016simulation, DEBKROMANRUB05}), which originally was used to design IS for estimating rare-event probabilities and later was applied also to root-finding and optimization. This approach involves iteratively updating the IS parameters by minimizing the Kullback-Leibler divergence between the considered IS class and a zero-variance estimator. This latter step typically results in an SAA problem constructed from samples drawn from the most recent IS. While our adaptive IS and cross entropy have similar characteristics in using sequential IS updates and SAA formulations, the settings and the formulated SAA are different. The SAA in cross entropy is an empirical counterpart of the Kullback-Leibler divergence minimization. In contrast, our SAA, or SA, arises from the objective function in our target root-finding problem, and we assume a priori availability of a good IS for the corresponding expectation estimation problem. As a result, the guarantees we achieve are also different from the cross-entropy method. Finally, we mention extensions and variants to cross entropy such as model reference adaptive search \citep{HuFuMa07} for the optimization setting and methods using, e.g., Markov chain Monte Carlo (\citealt{botev2013markov,grace2014automated,chan2012improved,botev2020sampling}).

We briefly review methods for stochastic root-finding and optimization problems, where SAA and SA are commonly used.
The Robbins-Monro SA (RM-SA) algorithm (\citealt{RM1951}) is one of the most widely used stochastic root-finding and optimization methods, which can be viewed as a stochastic counterpart to the quasi-Newton iteration in solving deterministic root-finding problems. 
To increase robustness and alleviate the well-known sensitivity of RM-SA to its stepsize sequence, iterate averaging, first proposed by \cite{polyak1992}, is often used. 
Other recent variants for improving the practical performance of the SA method include robust SA (\citealt{NemJudLanSha2009}),  accelerated SA (\citealt{GhaLan2013}), and the secant-tangent averaged SA (\citealt{ChQuFu14,ChauFu2017}). General finite-time bounds for SA can be found in, e.g., \cite{broadie2011general} and \cite{srikant2019finite}. More details of SA can be found in \cite{Fu2015}, \cite{kushner2003stochastic}, \cite{borkar2009stochastic}, and the references therein.


The SAA method (\citealt{shapiro2014lectures,KPH2015}), also known as the Monte Carlo sampling method (\citealt{Sha2003}, \citealt{ShaNem2005}, \citealt{homem2014monte}), is another commonly used technique for stochastic root-finding and optimization. 
Theoretical properties like consistency and asymptotic normality of SAA can be found in \cite{Rob1996} and \cite{kleywegt2002sample}. The lines of work \cite{mak1999monte}, \cite{bayraksan2006assessing}, \cite{bayraksan2011sequential}, \cite{freimer2012impact}, \cite{lam2017empirical} and \cite{lam2018bounding} study the statistical properties and estimation of optimality gaps in SAA. 
To improve computational efficiency, the retrospective approximation method is proposed to generate a sequence of SAA problems with progressively increasing sample size and then solve these problems with decreasing error tolerances (\citealt{PaSch2009,Pasu2010}). 

Recently, the probabilistic bisection algorithm (PBA), first proposed by \cite{Hor1963}, has been applied to stochastic root-finding problems. 
\cite{WaePeHe2013} derives a PBA where the expected absolute residuals converge to zero at a geometric rate, and \cite{PeHeWae2019} proposes an extended PBA that has a convergence rate arbitrarily close to, but slower than, the rate of SA.
\cite{RoLu2020} extends PBA to unknown sampling distributions and location-dependent settings. 
More methods for stochastic root-finding problems can be found in \cite{PaKi2011}, \cite{Waeber2013}, and the references therein.

\section{Problem Setting and Motivation}\label{sec:problem_set}
We consider a stochastic root-finding problem in the following standard form. Let $\boldf(\cdot)$ be a vector-valued function $\mathbb R^d\to\mathbb R^d$. Let $\bF(\mathbf X,\cdot)$ be an unbiased observation for $\boldf(\cdot)$ generated from simulation, where $\bF(\cdot,\cdot):\mathbb R^r\times\mathbb R^d\to\mathbb R^d$ and $\bX=(X_1,X_2,\ldots,X_r)\in\mathbb R^r$ is a random vector with probability distribution $P$. We are interested in finding the unique root $\btheta^*\in \mathbb R^{d}$ to the equation 
\begin{equation}\label{eq:basicRFeq}
\boldf(\btheta)\triangleq\mathds{E}_{\bX\sim P}\left[\bF(\bX,\btheta)\right]=\bc,
\end{equation}
where $\mathds{E}_{\bX\sim P}[\cdot]$ denotes the expectation under $P$. Our premise is that the function $\boldf(\btheta)$ lacks analytical tractability and can only be accessed via the unbiased simulation output $\bF(\bX,\btheta)$. Important examples that can be cast in the above general formulation include:
\begin{example}[Quantile estimation]
We are interested in finding the $p$th quantile of a random output, say $h(\bX)$, that can be simulated. Denote $F_{h}$ as the cumulative distribution function of $h(\bX)$. This problem is equivalent to finding the root of
\begin{equation*}
F_{h}(\theta)= P(h(\bX)\leq \theta)=\mathds{E}_{\bX\sim P}\left[\mathbf{1}{\{h(\bX)\leq \theta\}}\right] = p,
\end{equation*}
which is \eqref{eq:basicRFeq} with $\bF(\bX,\btheta) = \mathbf{1}{\{h(\bX)\leq \theta\}}$ and $\bc=p$.
\end{example}

\begin{example}[Stochastic optimization]
We are interested in an optimization problem
\begin{equation*}
\min_{\betheta} \mathds{E}_{\bX\sim P}[H(\bX,\btheta)],
\end{equation*}
where 
we have a stochastic gradient estimator $\mathbf F(\bX,\btheta)$ with $\mathds{E}_{\bX\sim P}[\mathbf F(\bX,\btheta)] = \nabla\mathds{E}_{\bX\sim P}[H(\bX,\btheta)]$ that can be simulated.
Then the first-order optimality condition becomes finding the root to
$$
\nabla\mathds{E}_{\bX\sim P}[H(\bX,\btheta)] =
\mathds{E}_{\bX\sim P}[\mathbf F(\bX,\btheta)]=\bzero,
$$
which is \eqref{eq:basicRFeq} with $\bc=\bzero$. Under appropriate regularity conditions, 
such a $\mathbf F(\bX,\btheta)$
can be found via techniques such as infinitesimal perturbation analysis (\citealt{heidelberger1988convergence,ho1983infinitesimal,PG1991,l1990unified}), the likelihood ratio or the score function method (\citealt{glynn1990likelihood,rubinstein1986score,reiman1989sensitivity}), measure-valued differentiation (\citealt{heidergott2010gradient}), or other variants 
(e.g., \citealt{FuHu97,peng2018new}). Alternately, when a ``direct" unbiased gradient estimator is unavailable, we can also approximate $\nabla\mathds{E}_{\bX\sim P}[H(\bX,\btheta)]$ via finite-difference schemes using $H(\bX,\cdot)$ to define $\mathbf F(\bX,\btheta)$. 



\end{example}

To solve \eqref{eq:basicRFeq}, the two main approaches in the literature are sample average approximation (SAA) and stochastic approximation (SA). They work as follows. In SAA, we first generate $n$ simulation samples $\{\bX_1,\bX_2,\ldots,\bX_n\}$ from $P$, and approximate \eqref{eq:basicRFeq}  by replacing the expectation with its empirical counterpart, namely
\begin{equation*}
\frac{1}{n}\sum_{i=1}^n \bF(\bX_i,\btheta) =\bc.
\end{equation*}
Then we apply deterministic root-finding procedures, e.g., the Newton-Raphson method, to obtain the estimated root $\hat \btheta_n$. 
Under regularity conditions (\citealt{shapiro2014lectures}, \citealt{ShaNem2005}), this estimator satisfies the asymptotic normality
\begin{equation}\label{SAA_asymptotic_standard}
    \sqrt{n}\left(\hat{\betheta}_n-\betheta^*\right)\Rightarrow \mathcal{N}\left(\bzero, [\bJ(\btheta^*)]^{-\top}\Var_{\bX\sim P}\left( \bF(\bX,\btheta^*)\right)  [\bJ(\btheta^*)]^{-1}\right),
\end{equation}
where ``$\Rightarrow$'' means convergence in distribution, $\bJ(\btheta)\triangleq D\boldf(\betheta)/D\betheta $ is the Jacobian matrix of $\boldf(\betheta)$, and ``$\top$'' denotes the transpose (``$-\top$'' denotes the inverse of the transpose).

In SA, the Robbins-Monro SA (RM-SA) procedure estimates the root by generating a sequence of iterates $\{\hat \btheta_{k}\}$ via the recursion
\begin{equation}\label{eq:SA-RM}
\hat \btheta_{k+1}=\hat \btheta_{k}-\gamma_{k}\bK\left( \bF\left(\bX_k, \hat \btheta_{k}\right)-\bc\right),~k=1,2,\ldots,
\end{equation}
where $\gamma_k$ is an appropriate stepsize such that 
\begin{equation*}\label{eq:RM_stepsize}
\sum_{i=1}^\infty \gamma_k = \infty~~\text{and}~~\sum_{i=1}^\infty \gamma^2_k < \infty,
\end{equation*}
and $\bK$ is an appropriate matrix. 
In practice, we usually set $\gamma_k = \gamma/k$, where $\gamma$ is a prescribed constant. With this choice of stepsize, under regularity conditions (\citealt{fabian1968}) we have asymptotic normality  (which is also an implication of our Theorem \ref{prop: SA_CLT} later) as follows.
Let $\bP$ be an orthogonal matrix such that 
\[
\gamma\bP^{\top}\bK\bJ(\btheta^{*})\bP=\bLambda
\]
is diagonal, then
\[
\sqrt{n}\left(\hat{\btheta}_{n}-\btheta^{*}\right)\Rightarrow\mathcal{N}\left(\bzero,\bP\bM\bP^{\top}\right),
\]
where $[\bM]_{ij}=\gamma^2[\bP^{\top}\bK\Var_{\bX\sim P}\left( \bF(\bX,\btheta^*)\right)\bK^{\top} \bP]_{ij}([\bLambda]_{ii}+[\bLambda]_{jj}-1)^{-1}$.

A variant of the RM-SA algorithm \eqref{eq:SA-RM} is the Polyak-Ruppert averaging SA (PR-SA; \citealt{polyak1992}), motivated by the desire to reduce sensitivity to stepsize in RM-SA. 
This approach takes bigger steps in \eqref{eq:SA-RM} 
(e.g., $\gamma_k = \gamma/k^{\alpha}$ for $1/2<\alpha<1$) and, when the iteration stops, averages the historical iterates to obtain
\begin{equation*}
\bar \btheta_n = \frac{1}{n}\sum_{k=1}^n \hat \btheta_k.
\end{equation*}
Under proper regularity conditions, SAA and PR-SA achieve the same asymptotic variance, which is optimal among RM-SA when the stepsize constant $\gamma$ and $\bK$ is chosen such that $\gamma \bK= [\bJ (\betheta^*)]^{-1}$ (\citealt{polyak1992,Asmussen2007}).
That is, we have 
\begin{equation}\label{SA_average_asymptotic_standard}
    \sqrt{n}\left(\bar{\betheta}_n-\betheta^*\right)\Rightarrow \mathcal{N}\left(\bzero, [\bJ(\btheta^*)]^{-\top}\Var_{\bX\sim P}\left( \bF(\bX,\btheta^*)\right)  [\bJ(\btheta^*)]^{-1}\right).
\end{equation}
The methodology that we propose in this paper applies to all three algorithms depicted above. 

\subsection{Challenges in Incorporating Importance Samplers}\label{sec:challenge IS}

In this paper, we consider situations where the root-finding problem \eqref{eq:basicRFeq} involves a rare event. A prime example is extreme quantile estimation, where the root $\theta^*$ is the quantile corresponding to a very high (or low) probability level. To facilitate presentation of our main ideas, we consider a one-dimensional output, i.e., $\bF(\cdot,\cdot)\in\mathbb R$ (also denoted as unbold form $F(\cdot,\cdot)$), in this section, and use the extreme quantile (Example 1) as our running example.

In the one-dimensional case, the asymptotic normality \eqref{SAA_asymptotic_standard} is simplified to 
$$\sqrt{n}\left(\hat{\theta}_n-\theta^*\right) \Rightarrow \mathcal{N}\left(0,\frac{\Var_{\bX\sim P}(F(\bX,\theta^*)) }{(f^{\prime}(\theta^*))^2}\right).$$
In quantile estimation, this is $\sqrt{n}(\hat{\theta}_n-\theta^*) \Rightarrow N\left(0,{p(1-p)}/{\phi(\theta^*)^2}\right)$, where $\phi(\theta)$ is the density of $h(\bX)$. Typically, when $p$ is very close to 0 or 1, $\phi(\theta)$ is correspondingly tiny and the asymptotic variance ${p(1-p)}/{\phi(\theta^*)^2}$ becomes very large. Note that the asymptotic normality dictates that the sample size $n$ required to achieve a target accuracy level is proportional to the asymptotic variance, so that a large variance implies a large required sample size. In fact, in many situations where the rare-event probability is governed by a large deviations theory (\citealt{bucklew2013introduction}), this required sample size is exponentially large in the ``rarity parameter". Such an issue motivates one to consider variance reduction, in particular IS.

The basic idea of IS is to change the probability measure from which random variables are generated, which results in more frequent hits on the important regions. To maintain unbiasedness, the outputs are weighted by the so-called likelihood ratios, which are the Radon-Nikodym derivatives between the IS and the original measures. IS achieves variance reduction by using a generating distribution that has well-controlled likelihood ratios in the target hit region.  
Consider the estimation of $f(\theta)$. The crude Monte Carlo estimator for $f(\theta)$ is given by $\sum_{i=1}^n F(\bX_i,\theta)/n$, where $\bX_1,\dots,\bX_n$ are i.i.d. samples drawn from the original distribution $P$. 
In contrast, IS generates samples $\bX_1,\dots,\bX_n$ from an IS distribution $P_{\bealpha}$, and estimates $f(\theta)$ by $\sum_{i=1}^n F(\bX_i,\theta)\ell(\bX_i,\bealpha)/n$, where $\ell(\bX,\bealpha) \triangleq \phi(\bX)/\phi_{\bealpha}(\bX)$, with $\phi$ and $\phi_{\bealpha}$ the densities of $\bX$ under $P$ and $P_{\bealpha}$, respectively, is the likelihood ratio. This latter estimator is unbiased as we can write $f(\theta) = \mathds{E}_{\bX\sim P_{\bealpha}}\left[ F(\bX,\theta) \ell(\bX,\bealpha)\right]$. Here, a good choice of the IS sampler $P_{\bealpha}$ should exhibit a small variance of $F(\bX_i,\theta)\ell(\bX_i,\bealpha)$.

However, for root-finding problems, the design of an IS sampler is more challenging. To see this, suppose we want to use some sampler $P_{\bealpha}$, and we apply SAA, RM-SA and PR-SA. More specifically, with IS, SAA would generate $\bX_1,\dots,\bX_n$ i.i.d. from $P_{\bealpha}$ and estimate the objective function $f(\theta)$ using
$$\frac{1}{n}\sum_{i=1}^n F(\bX_i,\theta)\ell(\bX_i,\bealpha)$$ 
from which we output the root, denoted $\hat{\theta}_n$. Under regularity conditions, the behavior of this root estimation is governed by the following CLT:
$$\sqrt{n}\left(\hat{\theta}_n-\theta^*\right) \Rightarrow \mathcal{N}\left(0,\frac{\Var_{\bX\sim P_{\bealpha}}\left(F(\bX,\theta^*)\ell(\bX,\bealpha)\right)}{(f^{\prime}(\theta^*))^2}\right).$$

For RM-SA and PR-SA, the recursion is replaced by 
$$\hat{\theta}_{k+1} = \hat{\theta}_{k}-\gamma_k F(\bX_k,\hat{\theta}_k)\ell(\bX_k,\bealpha),$$
where $\bX_k\sim P_{\bealpha}$ is independent of the past samples. Under regularity conditions, the errors of the root estimation using these methods are also governed by CLTs. For PR-SA, we know that 
$$\sqrt{n}\left(\bar{\theta}_n-\theta^*\right) \Rightarrow \mathcal{N}\left(0,\frac{\Var_{\bX\sim P_{\bealpha}}\left(F(\bX,\theta^*)\ell(\bX,\bealpha)\right)}{(f^{\prime}(\theta^*))^2}\right).$$
For RM-SA, it is known that if the stepsize is chosen as $\gamma_k = \gamma/k$, then 
$$\sqrt{n}\left(\hat{\theta}_n-\theta^*\right) \Rightarrow \mathcal{N}\left(0,\frac{\gamma^2\Var_{\bX\sim P_{\bealpha}}\left(F(\bX,\theta^*)\ell(\bX,\bealpha)\right)}{2\gamma f^{\prime}(\theta^*) -1}\right).$$
From the above asymptotic normalities, we observe that the asymptotic variance would depend on the choice of $P_{\bealpha}$ through the variance $\Var_{\bX\sim P_{\bealpha}}\left(F(\bX,\theta^*)\ell(\bX,\bealpha)\right) $. A good sampler $P_{\bealpha}$ needs to make this variance small. However, this expression for the variance involves $\theta^*$, which is unknown a priori as it is exactly what we want to solve. This leads to a circular challenge: {\em On one hand we need a good sampler to efficiently estimate the root; on the other hand, we need the root to decide the efficient sampler.} This phenomenon is fundamental in the design of IS for stochastic root-finding problems.

\subsection{A Worst-Case Perspective of Asymptotic Variances}\label{sec:minimax}

To more concretely illustrate the challenge and what our proposed adaptive IS can achieve, we use a worst-case perspective and look at the minimax error. To set up this discussion, suppose that the candidate importance samplers form a parametric family $\left\{ P_{\bealpha},\bealpha\in\Lambda\right\}$. Let $\Theta$ be the family of all possible roots, and we consider the family of root-finding problems in which we find the root of $\mathds{E}[F(\theta,\bX)]-c$ where $c$ varies in $\{f(\theta):\theta\in\Theta\} $. 
Note that, since we do not know the true root in advance, even if we know $c$, we cannot tell what is the $\theta$ that makes $f(\theta)=c$ beforehand.


Let us consider SAA and PR-SA first. If we use any fixed importance sampler $P_{\bealpha}$, then when $c=f(\theta_1)$, the estimator has the following asymptotic normality (suppose that $f^{\prime}(\theta)>0$ for every $\theta\in\Theta$)
\[
\sqrt{n}\left(\hat{\theta}_{n}-\theta_{1}\right)\Rightarrow \mathcal{N}\left(0,\frac{\Var_{\bX\sim P_{\bealpha}}(F(\mathbf{X},\theta_{1})\ell(\mathbf{X},\bealpha))}{(f^{\prime}(\theta_{1}))^{2}}\right).
\]
For any fixed sampler $P_{\bealpha}$ determined without knowledge of $\theta_{1}$, we have that the worst-case variance when $c$ varies in $\{f(\theta):\theta\in\Theta\} $ is
\[
\max_{\theta_{1}\in\Theta}\frac{\Var_{\bX\sim P_{\bealpha}}(F(\mathbf{X},\theta_{1})\ell(\mathbf{X},\bealpha))}{(f^{\prime}(\theta_{1}))^{2}}
\]
and the minimum possible worst-case asymptotic variance for any fixed IS is given by
\[
\min_{\alpha\in\Lambda}\max_{\theta_{1}\in\Theta}\frac{\Var_{\bX\sim P_{\bealpha}}(F(\mathbf{X},\theta_{1})\ell(\mathbf{X},\mathbf{\bealpha}))}{(f^{\prime}(\theta_{1}))^{2}}.
\]

On the other hand, suppose we can make use of the root when designing the importance sampler (i.e., we can choose $\bealpha$ based on $\theta_1$). Then we would choose $\bealpha=\argmin_{\bealpha}\Var_{\bX\sim P_{\bealpha}}(F(\bX,\theta_1)\ell (\bX,\bealpha)) $, and be able to achieve the following worst-case asymptotic variance 
$$ \max_{\theta_1\in\Theta}\min_{\bealpha\in\Lambda}  \frac{\Var_{\bX\sim P_{\bealpha}}(F(\bX,\theta_1)\ell (\bX,\bealpha))}{(f^{\prime}(\theta_1))^2}.$$
The relation between the two asymptotic variances can be seen by a weak duality: 
\begin{theorem}\label{thm: weak_duality} Suppose that $f^{\prime}(\theta)\neq 0$ for any $\theta\in\Theta$. Then
$$\max_{\theta_1\in\Theta}\min_{\bealpha\in\Lambda}  \frac{\Var_{\bX\sim P_{\bealpha}}(F(\bX,\theta_1)\ell (\bX,\bealpha))}{(f^{\prime}(\theta_1))^2}\leq  \min_{\bealpha\in\Lambda}\max_{\theta_1\in\Theta}  \frac{\Var_{\bX\sim P_{\bealpha}}(F(\bX,\theta_1)\ell (\bX,\bealpha))}{(f^{\prime}(\theta_1))^2}.$$
\end{theorem}

This theorem tells us that, if we use a fixed IS scheme, then we will suffer from a loss in efficiency due to a lack of a priori knowledge on the root. The gap between the two sides in Theorem \ref{thm: weak_duality} could be huge especially if only very limited knowledge of $\theta_1$ is available beforehand, i.e., $\Theta$ is a large set. The key of our proposed approach is to achieve the lower end of the inequality, by using suitable adaptive schemes. 
In other words, we achieve an asymptotic variance {\em as if we know the root}. 

We make two remarks. First, we consider the paradigm where we are already given an efficient IS for estimating the expectation (i.e., $f(\theta)$), which can be a ``black-box" in which we do not need to know any algorithmic details. The expectation estimation problem has been a long-standing focus of variance reduction, and our approach builds on the availability of these good IS schemes from the large literature. Second, we point out that in most practically interesting cases, it may not be easy or worthwhile to obtain an accurate minimizer for $\Var_{\bX\sim P_{\bealpha}}(F(\bX,\theta_1)\ell (\bX,\bealpha))$, even in expectation estimation problems when $\theta_1$ is known. This has led to different efficiency notions such as weak or logarithmic efficiency in the rare-event literature (\citealt{l2010asymptotic}). While Theorem \ref{thm: weak_duality} does not directly capture the comparisons using these specialized notions, our main assertion of achieving an asymptotic variance as if we know the root, which is not attainable using a fixed IS, still holds. Theorem \ref{thm: weak_duality} makes our main assertion clear by considering the basic setting where the minimization can be accurately solved.

We can do a similar analysis for the RM-SA algorithm, the only difference being that the asymptotic variance of the RM-SA algorithm is sensitive to the stepsize, and so we compare the behaviors of different IS schemes based on the same choice of stepsize, say $\gamma_k=\gamma/k$. For this choice, we have the following weak duality  
$$\max_{\theta_1\in\Theta}\min_{\bealpha\in\Lambda}  \frac{\gamma^2\Var_{\bX\sim P_{\bealpha}}(F(\bX,\theta_1)\ell (\bX,\balpha))}{2\gamma f^{\prime}(\theta_1)-1}\leq  \min_{\bealpha\in\Lambda}\max_{\theta_1\in\Theta}  \frac{\gamma^2\Var_{\bX\sim P_{\bealpha}}(F(\bX,\theta_1)\ell (\bX,\balpha))}{2\gamma f^{\prime}(\theta_1)-1},$$
where the right-hand side (RHS) is the best possible for all fixed IS schemes and the left-hand side (LHS) is the best we can do as if we know the root, which can be achieved using our adaptive IS. 









\section{Adaptive Importance Sampling}\label{sec: adaptiveIS}

In this section, we first present the main procedures of our adaptive IS to embed in SAA and SA. We then present asymptotic results and connect to our minimax discussion in Section \ref{sec:minimax}. As in Sections \ref{sec:challenge IS} and \ref{sec:minimax}, here we focus on the single-dimensional case with $d=1$, i.e., $\theta\in\mathbb R$, to highlight the main ideas of our developments. We generalize to multi-dimensional settings in Section \ref{sec:multidim} and the Appendix.

\subsection{Main Procedures}
We want to estimate the root of $f(\theta)=c$, where $f(\theta)=\mathds{E}_{\bX\sim P}[F(\bX,\theta)] $. We suppose there is an available good IS in the class $P_{\bealpha}$ for estimating the expectation $f(\theta)$ for a given $\theta$, i.e., once $\theta$ is given to us, we can choose $\bealpha = I(\theta)$ with a low variance  $\Var_{\bX\sim P_{\bealpha}}(F(\bX,\theta)\ell (\bX,\bealpha))$. Our procedures utilize this choice $I(\cdot)$.

Our procedure is adaptive and consists of iterations to update both the IS parameter and the root estimate simultaneously. More precisely, at each iteration, given the most updated root estimate, we create a new IS sampler with parameter $\bealpha=I(\theta)$ that aims to estimate the expectation with $\theta$ being precisely the most updated root estimate, and use it to draw new samples. From these new samples, we update our root estimate, via either SAA or SA, and repeat the iteration. 

Algorithm \ref{alg:SAA_rootfinding_blackbox} shows our adaptive IS embedded in SAA. At iteration $n=1,2,\ldots$, we are first given the most updated root estimate $\hat\theta_n$, and we set a new IS parameterized by $\bealpha_n=I(\hat\theta_n)$ (with some technical adjustment that we discuss momentarily). We use this IS to generate sample $\bX_n$. Then, we solve a new SAA problem constructed from all the observed samples $\bX_1,\dots,\bX_n$ with a proper importance weighting, given by \eqref{eq:SAA_AIS}, to obtain $\hat\theta_{n+1}$ and repeat the process.

For technicality reasons to ensure correct convergence, in Algorithm \ref{alg:SAA_rootfinding_blackbox} we add a truncation to the IS parameter set so that it cannot diverge too fast. This means we construct a series of deterministic sets $A_n$ such that $A_1\subset A_2 \subset\dots$ and their union contains the optimal IS parameter $\bealpha^* := I(\theta^*)\in\cup_{i=1}^\infty A_i$. Our implemented choice of IS would have parameter $\Pi_{A_{k+1}}[\bealpha_{k+1}]$, where $\Pi_A$ means a projection to set $A$. For the choice of the truncation sets $A_n$, if we have some prior knowledge that the optimal IS belongs to some compact set $A$, then we can simply let $A_n=A$ for each $n$. In particular, for quantile estimation problems, if we have crude knowledge on the possible deterministic range for the true quantile, then the possible value of $I(q)$ would belong to a computable compact set that implies the truncation set. If we do not have this prior knowledge, then we can let $A_i$ increase to the whole space: $\cup_{i=1}^\infty A_i = \mathbb{R}$. We also need that $A_i$ does not grow too fast (so that Assumption \ref{assu: truncation_alpha} in our sequel is satisfied).

\begin{algorithm}[htp]
\caption{SAA with adaptive importance sampling for stochastic root finding}
\label{alg:SAA_rootfinding_blackbox}
\begin{algorithmic}[1]
\Ensure Original sampling distribution $P$; initial IS parameter $\balpha_1$; initial iteration index $n=1$; truncation sets $A_1\subset A_2 \subset \dots $; black-box IS function $I$.
\While{stopping criteria not met}
\State Generate sample $\bX_n\sim P_{\bealpha_n}$; 
\State Update root estimate $\hat \theta_n$ by solving the equation 
\begin{equation}\label{eq:SAA_AIS}
\frac{1}{n}\sum_{i=1}^n F(\bX_i,\theta)\ell(\bX_i,\balpha_i)=c;
\end{equation}
\State Update IS parameter $\balpha_{n+1} = \Pi_{A_{n+1}}[I(\hat \theta_n)]$;
\State Set $n=n+1$;
\EndWhile
\lastcon{Root estimate $\hat \theta_n$.
}
\end{algorithmic}
\end{algorithm}

Algorithm \ref{alg: SA_rootfinding_blackbox} shows our adaptive IS embedded in SA, where we replace the SAA problem in each iteration with an SA move, with the corresponding importance weight in \eqref{eq: update_SA_QE_proj1}. At the end of the procedure, we either output the final root iterate $\hat\theta_n$ (RM-SA) or the average $\bar\theta_n=\sum_{i=1}^n\hat\theta_i/n$ (PR-SA). Note that in Algorithm \ref{alg: SA_rootfinding_blackbox} the updating step in each iteration only utilizes the current sample $\bX_n$, as opposed to using all past samples as in SAA.

\begin{algorithm}[htb]
\caption{SA with adaptive importance sampling for stochastic root finding}
\label{alg: SA_rootfinding_blackbox}
\begin{algorithmic}[1]
\Ensure Original sampling distribution $P$; initial IS parameter $\bealpha_1$; initial root $\hat \theta_0$; stepsize constant $\gamma$; prior information set $A$; initial iteration index $n=1$; black-box IS function $I$.
\While{stopping criteria not met}
\State Generate sample $\bX_n\sim P_{\bealpha_n}$, and calculate $F(\bX_n,\hat \theta_{n-1})$ and $\ell(\bX_n,\balpha_n)$;
\State Set $\gamma_n=\gamma/n^{\alpha}$ (usually $\alpha=1$ for RM-SA; $1/2<\alpha<1$ for PR-SA);
\State Update root estimate
\begin{equation}\label{eq: update_SA_QE_proj1}
\hat{\theta}_{n}=\Pi_A\left[\hat{\theta}_{n-1}-\gamma_{n}\left(F(\bX_n,\hat \theta_{n-1})\ell(\bX_{n},\balpha_{n})-c\right)\right];
\end{equation}
\State Update IS parameter $\balpha_{n+1} = I(\hat \theta_n)$;
\State Set $n=n+1$;
\EndWhile
\lastcon{Root estimate $\hat \theta_n$ for RM-SA, or $\bar \theta_n=\sum_{i=1}^n \hat \theta_i/n$ for PR-SA.
}
\end{algorithmic}
\end{algorithm}

\subsection{Theoretical Results}\label{sec: theoretical_1dim}

We present our main theoretical results on the consistency and asymptotic normality of our SAA and SA with embedded adaptive IS. 

\subsubsection{SAA with Adaptive IS.}

We first consider SAA, i.e., Algorithm \ref{alg:SAA_rootfinding_blackbox}. We need several assumptions and intermediate lemmas.
The first assumption is about the growth rate of the IS parameter $\bealpha_n$. 
\begin{assumption}\label{assu: truncation_alpha}
For each $\theta \in \Theta$,  $\mathds{E}_{\bX\sim P_{\bealpha_n}}\left[\left(F(\bX,\theta)\ell(\bX,\balpha_{n})\right)^{2}\right]=O(n^{1-\epsilon})$ holds for some $\epsilon>0$.
\end{assumption}

In our algorithm, this is guaranteed by introducing the truncation sets $A_n$ (see the discussion before Algorithm \ref{alg:SAA_rootfinding_blackbox}).
With this assumption, we have pointwise convergence of the estimated objective function in the following lemma, the proof of which is provided in Appendix \ref{appx:prooflemma1}. Note that $\bX_i$ refers to a sample drawn from the adaptive algorithm, i.e., $\bX_i\sim P_{\bealpha_i}$ and given $\bealpha_i$, the distribution of $\bX_i$ is independent of $\bX_1,\dots,\bX_{i-1}$.
\begin{lemma}\label{lem: average_consistent}
Under Assumption \ref{assu: truncation_alpha}, for each $\theta\in\Theta$,
we have
\begin{equation*}
\frac{\sum_{i=1}^{n}F(\bX_{i},\theta)\ell(\bX_{i},\balpha_{i})}{n}\rightarrow \mathds{E}_{\bX\sim P}\left[F(\bX,\theta)\right]~\text{a.s.}\
\end{equation*}
\end{lemma}

Based on this pointwise convergence, we can show a uniform convergence of ${\sum_{i=1}^{n}F(\bX_{i},\theta)\ell(\bX_{i},\balpha_{i})}/{n}$ over all possible values of $\theta$. Following the bracketing number approach (see Section 2.4 of \citealt{vanWell1996}), we make the following assumption.
\begin{assumption}\label{assu: bracket}
Define $\sfF=\{f(\bX,\balpha):=F(\bX,\theta)\ell(\bX,\balpha), \theta\in \Theta\}$.
For each $\epsilon>0$, there exists a finite set $K_{\epsilon}$ whose
elements are pairs of functions such that:\\
(1) For each $f\in\sfF$, there exists $(f_{L},f_{R})\in K_{\epsilon}$ such that $f_{L}\leq f\leq f_{R}$;\\
(2) For each pair of $(f_{L},f_{R})\in K_{\epsilon}$, the limits 
\begin{equation*}
\lim_{n\rightarrow\infty}\frac{1}{n}\sum_{i=1}^{n}f_{L}(\bX_{i},\balpha_{i})~~\text{and}~~\lim_{n\rightarrow\infty}\frac{1}{n}\sum_{i=1}^{n}f_{R}(\bX_{i},\balpha_{i})
\end{equation*}
exist, and
\begin{equation*}\label{bracket_size}
\lim_{n\rightarrow\infty}\frac{1}{n}\sum_{i=1}^{n}\left[f_{R}(\bX_{i},\balpha_{i})-f_{L}(\bX_{i},\balpha_{i})\right]\leq\epsilon.
\end{equation*}
\end{assumption}

To verify Assumption \ref{assu: bracket}, we can also follow similar arguments to bound bracketing numbers as in, e.g., Section 2.7 of \cite{vanWell1996}. We have intentionally stated our assumption in a general way without using any condition like Lipschitz continuity or smoothness of function $F$. That is because, for quantile estimation problems which we consider as an important example,  $F(\bX,\theta)=\mathbf{1}\{h(\bX)\leq\theta\}$ is not even continuous in $\theta$. On the other hand, if we have some smoothness conditions for $F$, it could help verify this assumption. For example, if the function class is Lipschitz continuous in $\theta$, in the sense that 
\[
| F(\bX,\theta_1) - F(\bX,\theta_2)|\leq d(\theta_1,\theta_2)h(\bX)
\]
for some metric $d$ on the index set, and $\lim_{n\rightarrow\infty} \sum_{i=1}^n h(\bX_i)\ell(\bX_i,\balpha_i)/n<\infty$ (which can be shown in a similar way as Lemma \ref{lem: average_consistent}), then similar to Theorem 2.7.11 of \cite{vanWell1996}, a sufficient condition for Assumption \ref{assu: bracket} is that the $\epsilon$-covering number on set $\Theta$, $N(\epsilon,\Theta,d)$, is finite. Here the covering number is the minimum number of $\epsilon$-balls under metric $d$ needed to cover $\Theta$, where an $\epsilon$-ball centered at $\theta_0$  under metric $d$ means $\{\theta\in\Theta:d(\theta,\theta_0)<\epsilon\}$.
More concretely, if $d$ is the Euclidean distance, then the compactness of $\Theta$ would be sufficient for the aforementioned condition. But the compactness of $\Theta$ is not necessary because it could be the case that $d(\theta_1,\theta_2)$ shrink to 0 when $\theta_1$ and $\theta_2$ are large.

The following lemma presents the uniform convergence of the estimated objective function. Its proof uses a generalized bracketing number to handle  the sum of martingale difference array and is in Appendix \ref{appx:prooflemma2}. 

\begin{lemma}\label{Lem: unif_converg_aver_obj}
Under Assumptions \ref{assu: truncation_alpha} and \ref{assu: bracket},
we have that, as $n\rightarrow\infty$, 
\begin{equation*}
\sup_{\theta\in \Theta}\left|\frac{\sum_{i=1}^{n}F(\bX_{i},\theta)\ell(\bX_{i},\balpha_{i})}{n}-\mathds{E}_{\bX\sim P}\left[F(\bX,\theta)\right]\right|\rightarrow0~\text{a.s.}
\end{equation*}
\end{lemma}

Based on the uniform convergence of the sample-averaged objective function, similar in spirit to Theorem 5.7 of \cite{vaart_1998}, we make the following assumption to ensure the root to $f(\theta)-c$ is well separated.
\begin{assumption} \label{assu: Obj_regular}
The objective function $f(\theta)$ is differentiable at $\theta$=$\theta^{*}$ with continuously invertible derivative, and $\theta^*$ is the unique root in the sense that for any $\epsilon>0$,
\begin{equation*}
\inf_{| \theta-\theta^{*}|\geq\epsilon}\left| f(\theta)- c\right| >0.
\end{equation*}
\end{assumption}

With this assumption, and using the uniform convergence derived in Lemma \ref{Lem: unif_converg_aver_obj}, we can show the strong consistency of the SAA estimator with adaptive IS. The proof is included in Appendix \ref{appx:proof_Thm2}.
\begin{theorem}[Consistency of SAA with embedded adaptive IS]\label{prop:SAA_consistency}
Under Assumptions \ref{assu: truncation_alpha} - \ref{assu: Obj_regular}, the root estimator generated by Algorithm \ref{alg:SAA_rootfinding_blackbox} is strongly consistent, i.e., $$\hat{\theta}_{n}\rightarrow\theta^{*}~a.s.$$
\end{theorem}

Next, we will establish a CLT for the estimator using the weak convergence of martingale processes. To develop this, we introduce additional notation and some preliminary technical tools. 
Let $\Theta_{\delta}:=\{\theta:\left\Vert\theta-\theta^*\right\Vert\leq\delta\}$, $f_{\theta}$ be the function defined by $f_{\theta}(\bX):= F(\bX,\theta)$, and $\sfF_{\delta}=\{f_{\theta},\theta\in\Theta_{\delta}\}$. 
Furthermore, for each measurable function $g$ on the probability space $(\Omega,\mathcal{F},P)$, let
\[
V_{n,i}(g)=\frac{g(\bX_{i})\ell(\bX_{i},\balpha_{i})-\mathds{E}_{\bX\sim P}[g(\bX)]}{\sqrt{n}}
~~\text{and}~~S_{n}(g)=\sum_{i=1}^n V_{n,i}(g).
\]
Notice that
when $\theta$ varies in $\Theta_\delta$, $V_{n,i}(f_{\theta})$ and
$S_{n}(f_{\theta})$ could be regarded as processes indexed by $\Theta_{\delta}$. We denote $\mathcal F_n=\sigma(\bX_1,\ldots,\bX_n)$ as the filtration. We write $\mathds{E}_{i-1}$ as the conditional expectation given $\mathcal{F}_{i-1}$, and similarly $\Var_{i-1}$ as the conditional variance given
$\mathcal{F}_{i-1}$. 
Notice that, for each $\theta\in\Theta_{\delta}$, 
\[
V_{n,i}\left(f_{\theta}\right)=\frac{F(\bX_{i},\theta)\ell(\bX_{i},\balpha_{i})-\mathds{E}_{\bX\sim P}[F(\bX,\theta)]}{\sqrt{n}}
\]
is a martingale difference array. For each fixed $\theta$, the standard martingale CLT can give the asymptotic normality for $S_n(f_{\theta})$, but for the asymptotic normality of the root, we need a uniform behavior of $S_n(f_{\theta})$ for $\theta$ in a neighborhood of $\theta^*$. To this end, enlightened by Donsker-type theorems and the analysis of the weak convergence of function-valued martingale difference arrays in \cite{nishiyama2000}, we construct the following assumption. 
\begin{assumption}\label{assu:(uniform-integrable-entropy)}
There exists $\Pi=\{\Pi(\epsilon)\}_{\epsilon\in(0,\Delta_{\Pi}]}$ such that each $\Pi(\epsilon)=\{\sfF(\epsilon;k):1\leq k \leq N_{\Pi}(\epsilon)\}$ is a cover of $\sfF_{\delta}$ (i.e., $\cup_{1\leq k \leq N_{\Pi}(\epsilon)}\sfF(\epsilon;k)=\sfF_{\delta}$) and $N_{\Pi}(\Delta_{\Pi})=1$. Here for each $1\leq k\leq N_{\Pi}(\epsilon)$, $\sfF(\epsilon;k)$ is an $\epsilon$-ball under $L_2$-distance $\rho(g,h) := (\mathds{E}_{\bX\sim P}[(g(\bX)-h(\bX))^2])^{1/2}$. Moreover,
\[
\sup_{\epsilon\in(0,\Delta_{\Pi}]\cap \mathbb{Q}} \max_{1\leq k\leq N_{\Pi}(\epsilon)}\frac{\sqrt{\sum_{j=1}^{n}\mathds{E}_{j-1}\left[\left\vert V_{n,j} (\sfF(\epsilon;k))\right\vert^2\right]}}{\epsilon}=O_p(1),
\]
where for a set $\sfF^{\prime}$, $V_{n,j}(\sfF^{\prime})$ is defined as the smallest $\mathcal{F}_i$-measurable function that is greater than $\sup_{f,g\in \sfF^{\prime}}\left\vert V_{n,j}(f)-V_{n,j}(g)\right\vert$. 
Furthermore,
\[
\int_0^{\Delta_{\Pi}}\sqrt{\log N_{\Pi}(\epsilon)} d\epsilon<\infty.
\]
\end{assumption}

Here, the first displayed condition requires that each set in the cover $\Pi(\epsilon)$ should be small enough and the second displayed condition requires that there cannot be too many sets in the cover. The following proposition provides a more transparent sufficient condition for Assumption \ref{assu:(uniform-integrable-entropy)}. With slight abuse of notations, we let $\rho(\theta_1,\theta_2):=\left(\mathds{E}_{\bX\sim P}[(F(\bX,\theta_1)-F(\bX,\theta_2))^2]\right)^{1/2}$ be a pseudo-metric on $\Theta_{\delta}$. Notice that in Assumption \ref{assu:(uniform-integrable-entropy)} we also use $\rho$ to denote the $L_2$-distance between functions under $P$, and we have that $\rho(\theta_1,\theta_2)=\rho(f_{\theta_1},f_{\theta_2})$. Let the {\em covering number} $N(\epsilon,\Theta_{\delta},\rho)$ be the minimum number of balls $\{\theta: \rho(\theta,\theta_1)<\epsilon\}$ of radius $\epsilon$ needed to cover $\Theta_{\delta}$. 
\begin{proposition}\label{prop: verify_entropy}
Suppose that we are given $\hat{\theta}_{n}\rightarrow\theta^{*}$ and $\bealpha_{n}\rightarrow\bealpha^{*}$, both a.s. Then the following condition is sufficient for Assumption \ref{assu:(uniform-integrable-entropy)}: There exists a $\delta>0$ such that 
\[
(i)~~~~~~~~~~~~~~~~~~~~~~~~~~~~~~~~~~~~~~~~
\int_{0}^{1}\sqrt{\log N\left(\epsilon,\Theta_{\delta},\rho\right)}d\epsilon<\infty,
~~~~~~~~~~~~~~~~~~~~~~~~~~~~~~~
\]
and (ii) there exists constant $\delta_1>0$ and real-valued function $L(\bX,\bealpha)$ 
such that for any $\theta_1,\theta_2\in\Theta_{\delta}$,
\[
\left(F(\bX,\theta_1)-F(\bX,\theta_2)\right)^2\ell(\bX,\bealpha)\leq L(\bX,\bealpha)(\rho(\theta_1,\theta_2))^2
\mbox{~~with~}  \sup_{\left\Vert\bealpha-\bealpha^*\right\Vert\leq \delta_1}\mathds{E}_{\bX\sim P}[L(\bX,\bealpha)]<\infty. 
\]
\end{proposition}

Here, the first part is a uniform entropy condition that is commonly assumed for Donsker-type theorems. The second can be regarded as a Lipschitz condition for function $F$ when $\theta$ is close to $\theta^*$ and $\bealpha$ is close to $\bealpha^*$.  
Similar to the conditions for the martingale CLT (see e.g. Theorem 8.2.8 of \citealt{Durret2019}), we introduce our last assumption.
\begin{assumption}(Lindeberg's condition)\label{assu:(Lindeberg-Feller)}
There exists a $\delta>0$ such that for every $\epsilon>0$, 
\[
\sum_{i=1}^{n}\mathds{E}_{i-1}\left[\left(V_{n,i}(E)\right)^{2}\mathbf{1}\left\{V_{n,i}(E)>\epsilon\right\}\right]\stackrel{P}{\longrightarrow}0,
\] 
where $V_{n,i}(E)$ is the adapted envelope for $V_{n,i}(f),f\in\sfF_{\delta}$, i.e., $V_{n,i}(E)$ is the smallest $\mathcal{F}_i$-measurable random variable such that  $\sup_{f\in\sfF_{\delta}}|V_{n,i}(f)|\leq V_{n,i}(E)~a.s$. 
\end{assumption}

To verify this assumption, we can bound $V_{n,i}(E)$ when $\bealpha_i$ and $\hat{\theta}_i$ are close to $\bealpha^*$ and $\theta^*$ respectively. We summarize this observation in the following proposition. The proof is in Appendix \ref{appx:proofprop13}.
\begin{proposition}\label{prop: LF_sufficient}
Suppose that we are given $\hat{\theta}_{n}\rightarrow\theta^{*}$ and $\bealpha_{n}\rightarrow\bealpha^{*}$, both a.s. Suppose that there exist $\delta,\delta_{1}>0$ and $V(\bX)$ such that $V(\bX)\geq (F(\bX,\theta))^2\ell(\bX,\bealpha)$ for all $\theta\in\Theta_{\delta},\left\Vert \bealpha-\bealpha^{*}\right\Vert \leq\delta_{1}$. Also suppose $\mathds{E}_{\bX\sim P} [V(\bX)]<\infty$. Then Assumption \ref{assu:(Lindeberg-Feller)} holds.
\end{proposition}

With these assumptions, we can first prove asymptotic equicontinuity for the sum of martingale difference arrays at $f_{\theta^*}$. That is, for any large $n$ and $\theta\in\Theta_{\delta}$, as long as $\rho(\theta,\theta^*)$ is small, $S_n(f_{\theta})$ and $S_n(f_{\theta^*})$ should be close enough. Similar to Donsker-type theorems, this would guarantee a uniform behavior in the martingale CLT $S_n(f_\theta)\Rightarrow Z(\theta)$ for all $\theta\in\Theta_{\delta}$, where $Z(\theta)$ is a normal random variable with mean zero and variance determined by $\theta$. Another main ingredient of our proof is that, since $\hat{\theta}_n\rightarrow\theta^*$ a.s., we have that $\bealpha_n\rightarrow\bealpha$ a.s., thus the variance of $S_n(\hat{\theta}_n)$ would converge to $\Var_{\bX\sim P_{\bealpha^*}}(F(\bX,\theta^*)\ell(\bX,\bealpha^*))$. This tells us the variance of $Z(\theta^*)$. Then, based on this, we derive the asymptotic normality of the SAA with adaptive IS in the following theorem. The full proof is in Appendix \ref{appx:proof_of_A2}.

\begin{theorem}[Asymptotic normality of SAA with embedded adaptive IS]\label{thm: SAA_CLT}
Under Assumptions \ref{assu: truncation_alpha} - \ref{assu:(Lindeberg-Feller)}, suppose that the function $\mathds{E}_{\bX\sim P_{\bealpha}}\left[\left(\bF(\bX,\theta)\ell(\bX,\balpha)\right)^{2}\right]$ is continuous in $\bealpha$, and $\rho\left(\theta,\theta^{*}\right)\rightarrow0$ as $\theta\rightarrow\theta^{*}$.
Suppose further that the black-box function $I$ is continuous.
Then we have asymptotic normality of $\hat{\theta}_{n}$ generated from Algorithm \ref{alg:SAA_rootfinding_blackbox}, given by
\begin{equation*}
\sqrt{n}\left(\hat{\theta}_n-\theta^*\right)\Rightarrow \mathcal{N}\left(0,\frac{\Var_{\bX\sim P_{\bealpha^*}}\left(F(\bX,\theta^*)\ell(\bX,\bealpha^*)\right)}{(f'(\theta^*))^2}\right),
\end{equation*}
where $\balpha^*=I(\theta^*)$.
\end{theorem}

This theorem tells us that using our adaptive IS, the asymptotic variance would be the same as if we use a fixed sampler $\bealpha^*$, or in other words as if we know the root in advance. In particular, when  $I(\theta) = \argmin_{\bealpha} \Var_{\bX\sim P_{\bealpha}}(F(\bX,\theta)\ell(\bX,\bealpha))$, we would achieve the asymptotic normality
\begin{equation*}
\sqrt{n}\left(\hat{\theta}_n-\theta^*\right)\Rightarrow \mathcal{N}\left(0,\min_{\bealpha}\frac{\Var_{\bX\sim P_{\bealpha}}\left(F(\bX,\theta^*)\ell(\bX,\bealpha)\right)}{(f'(\theta^*))^2}\right).
\end{equation*}
So the worst-case variance would be the LHS of the weak duality in Theorem \ref{thm: weak_duality}.

\subsubsection{SA with Adaptive IS.}\label{sec:SIS_SA}
Now we turn to SA, i.e., Algorithm \ref{alg: SA_rootfinding_blackbox}. For consistency, we adopt the projected ordinary differential equation (ODE) approach in \cite{kushner2003stochastic}. Since $\bealpha_{n+1}\in\mathcal{F}_{n}$, we have that 
\begin{equation*}\label{eq: SA_unbias}
\mathds{E}_{\bX_{n+1}\sim P_{\bealpha_{n+1}}}\left[F(\bX_{n+1},\hat{\theta}_{n})\ell(\bX_{n+1},\balpha_{n+1})|\mathcal{F}_{n}\right]={f}\left(\hat{\theta}_{n}\right),
\end{equation*}
so iteration \eqref{eq: update_SA_QE_proj1} has a martingale difference noise. 
To proceed, we state some assumptions.
The first is the boundedness of the conditional variance of the gradient estimator. Let $V_{n}=F(\bX_{n+1},\hat{\theta}_{n})\ell(\bX_{n+1},\balpha_{n+1})-f(\hat{\theta}_{n})$ be the error of the estimated objective function. We assume the following.
\begin{assumption}\label{assu:. sup_L2}
There exists a constant $C>0$ such that 
 $\mathds{E}_{n}[V_n^2]<C$.
\end{assumption}

In our adaptive method, $\bealpha_{n+1}$ can be chosen to make this conditional variance small, so that this assumption is typically readily verifiable. 
Our next assumption is a one-dimensional version of the constraint set condition (A4.3.2) of \cite{kushner2003stochastic} and the uniqueness of the root. 
\begin{assumption}\label{assu: unique_solution}
$A=[a,b]$ for some $-\infty<a<b<\infty$, $\theta^{*}$ belongs to the interior of $A$, $f(a)<c<f(b)$ and $\theta^*$ is the unique root of $f(\theta)=c$.
\end{assumption}

By verifying the conditions of Theorem 5.2.3 in \cite{kushner2003stochastic} (See Appendix \ref{sec: proof_consistency_SA}), we have the consistency of $\hat{\theta}_{n}$ and $\bar \theta_n$.
\begin{theorem}[Consistency of SA with embedded adaptive IS]\label{prop:RM_consistency}
Under Assumptions \ref{assu:. sup_L2}-\ref{assu: unique_solution}, both the RM-SA estimator
$\hat{\theta}_{n}$ and the PR-SA estimator $\bar\theta_n$ defined in Algorithm \ref{alg: SA_rootfinding_blackbox} converge
to $\theta^{*}$ a.s. 
\end{theorem}

Next we present the asymptotic normality of the root estimator, considering first the RM-SA estimator $\hat \theta_n$.
Following \cite{fabian1968}, we introduce the following assumption on the uniform integrability of the squared noise. 
\begin{assumption}
\label{assu: Polyak_noise}
\[
\sup_{n}\mathds{E}_{n-1}[|V_{n}|^{2}\mathbf{1}\{|V_{n}|>R\}]\stackrel{P}{\longrightarrow}0~\text{ as \ensuremath{R\rightarrow\infty}.}
\]
\end{assumption}
Similar to Assumption \ref{assu:(Lindeberg-Feller)}, we also have the following sufficient condition to verify Assumption \ref{assu: Polyak_noise}. The proof is in Appendix \ref{appx:proofprop13}.
\begin{proposition}\label{prop: PN_sufficient}
Suppose that we are given $\hat{\theta}_{n}\rightarrow\theta^{*}$ and $\bealpha_{n}\rightarrow\bealpha^{*}$, both a.s. Suppose that there exist $\delta,\delta_{1}>0$ and $V(\bX)$ such that $V(\bX)\geq (F(\bX,\theta))^2\ell(\bX,\bealpha)$ for all $\theta\in\Theta_{\delta},\left\Vert \bealpha-\bealpha^{*}\right\Vert \leq\delta_{1}$. Also suppose $\mathds{E}_{\bX\sim P} [V(\bX)]<\infty$. Then Assumption \ref{assu: Polyak_noise} holds. 
\end{proposition}

With these, we have the following CLT for the RM-SA estimator.
\begin{theorem}[Asymptotic normality of RM-SA with embedded adaptive IS] \label{prop: SA_CLT}
Suppose that $f(\cdot)$ is a twice differentiable function with $f^{\prime}(\theta^*)\neq 0 $. Also suppose that the function $\mathds{E}_{\bX\sim P_{\bealpha}}\left[\left(\bF(\bX,\theta)\ell(\bX,\balpha)\right)^{2}\right]$ is continuous in $\bealpha$ and the black-box function $I$ is continuous.
Under Assumptions \ref{assu:. sup_L2}-\ref{assu: Polyak_noise}, the RM-SA estimator $\hat{\theta}_n$ in Algorithm  \ref{alg: SA_rootfinding_blackbox} is asymptotically normal, viz.,
\begin{equation*}
\sqrt{n}\left(\hat{\theta}_n-\theta^*\right)\Rightarrow \mathcal{N}\left(0,\frac{\gamma^2\Var_{\bX\sim P_{\bealpha^*}}\left(F(\bX,\theta^*)\ell(\bX,\balpha^*)\right)}{2\gamma f^{\prime}(\theta^*) -1}\right),
\end{equation*}
where $\balpha^*=I(\theta^*)$. 
\end{theorem}
The proof of Theorem \ref{prop: SA_CLT} requires reformulating the recursion in our Algorithm \ref{alg: SA_rootfinding_blackbox} in an alternative form used in \cite{fabian1968}. See Appendix \ref{appx:proofTheorem3} for the details.

Lastly, we consider the PR-SA estimator $\bar \theta_n$. 
Following \cite{polyak1992}, we have the following asymptotic result, whose proof is in Appendix \ref{appx:proofprop6}. 
\begin{theorem}[Asymptotic normality of PR-SA with embedded adaptive IS]\label{prop: SA_average_CLT}
Under the same assumptions in Theorem \ref{prop: SA_CLT}, the PR-SA estimator $\bar \theta_n$ in Algorithm \ref{alg: SA_rootfinding_blackbox} is asymptotically normal, viz., 
\begin{equation}\label{eq:asym_norm}
\sqrt{n}\left(\bar{\theta}_{n}-\theta^{*}\right)\Rightarrow \mathcal{N}\left(0,\frac{\Var_{\bX\sim P_{\bealpha^*}}\left(F(\bX,\theta^*)\ell(\bX,\bealpha^*)\right)}{f'(\theta^*)^2}\right),
\end{equation}
where $\bar{\theta}_{n}=\sum_{i=1}^n\hat{\theta}_{i}/n$ and $\balpha^*=I(\theta^*)$. 
\end{theorem}

Similar to the argument after Theorem \ref{thm: SAA_CLT}, when  $I(\theta) = \argmin_{\bealpha} \Var_{\bX\sim P_{\bealpha}}(F(\bX,\theta)\ell(\bX,\bealpha))$, the asymptotic normality \eqref{eq:asym_norm} is equivalent to
\begin{equation*}
\sqrt{n}\left(\bar{\theta}_{n}-\theta^{*}\right)\Rightarrow \mathcal{N}\left(0,\min_{\bealpha}\frac{\Var_{\bX\sim P_{\bealpha}}\left(F(\bX,\theta^*)\ell(\bX,\bealpha)\right)}{f'(\theta^*)^2}\right),
\end{equation*}
so the worst-case variance using PR-SA with our adaptive IS is the LHS in the weak duality result in Theorem \ref{thm: weak_duality}.

\section{Quantile Estimation with Adaptive Importance Sampling}\label{sec:quantile}
In this section, we apply our adaptive IS to the quantile estimation problem. 
Let $h:\mathbb{R}^d\to\mathbb{R}$ be a performance function of a stochastic system, and recall that $F_h(x)=P\{h(\bX)\leq x\}$ is the cumulative distribution function of $h(\bX)$. Then the $p$-quantile of $h(\bX)$ is
\begin{equation}
q^*=\inf\{x: P\{h(\bX)\leq x\}\geq p\} =\inf\{x: F_h(x)\geq p\}.\label{quantile original}
\end{equation}
Henceforth, we assume that $h(\bX)$ is a continuous random variable, so the $p$-quantile of $h(\bX)$ is the root of the equation
\begin{equation}\label{eq:quantile_root}
F_h(q)=\mathds{E}_{\bX\sim P}[\mathbf{1}\{h(\bX)\leq q\}] = p.
\end{equation}
Note that in this section, $q$ is used in place of $\theta$ as the variable in the root-finding equation.

As discussed in the introduction, when $p$ is close
to 0 or 1, standard Monte Carlo could perform poorly, where the standard approach means using the empirical quantile, i.e., generating i.i.d. samples $\bX_{1},\bX_{2},\cdots,\bX_{n}\sim P$, and solving the empirical root-finding problem \eqref{eq:quantile_root} by plugging in the empirical distribution for $F_h$. To address the extreme value estimation problem, we apply the adaptive IS approach to the quantile estimation setting, for which we also provide milder and easier-to-verify conditions for the asymptotic analysis.

\subsection{Empirical Quantile}
The empirical quantile depicted above is the analog of the SAA solution in quantile estimation. 
Similar to the previous section, the IS distribution is parameterized by $\balpha$, on which we know a black-box function $I(q)$ that gives a good IS parameter $\balpha$ for estimating $P(h(\bX)\leq q)$. Algorithm \ref{alg:SAA_quantile} presents our adaptive IS embedded in the empirical quantile, where the truncation set $A_n$ is the same as in Algorithm \ref{alg:SAA_rootfinding_blackbox} that is used to guarantee strong consistency. 
\begin{algorithm}[htb]
\caption{Empirical quantile with adaptive importance sampling}
\label{alg:SAA_quantile}
\begin{algorithmic}[1]
\Ensure Original sampling distribution $P$; initial IS parameter $\balpha_1$; initial iteration index $n=1$; truncation sets $A_1\subset A_2 \subset \dots $; black-box IS function $I$. 
\While{stopping criteria not met}
\State Generate sample $\bX_n\sim P_{\bealpha_n}$ ;
\State Update quantile estimate
\[
\hat q_n=\inf\left\{\frac{1}{n}\sum_{i=1}^n \mathbf{1}\{h(\bX_i)\leq q\}\ell(\bX_i,\balpha_i)\geq p\right\};
\]
\State Update IS parameter $\balpha_{n+1}=\Pi_{A_{n+1}}\left[I(\hat{q}_n)\right]$;
\State Set $n=n+1$;
\EndWhile
\lastcon{Quantile estimate $\hat q_n$.
}
\end{algorithmic}
\end{algorithm}

Because of the special monotone structure of the objective function, the assumptions required for the quantile estimation asymptotics are considerably simpler than those used in the general case presented previously. Our first assumption is analogous to Assumption \ref{assu: truncation_alpha}.
\begin{assumption} \label{assu: SAA_QE_consistency}
There exist $\delta,\epsilon>0$ such that  $\mathds{E}\left[\left(\mathbf{1}\{h(\bX_n)\leq q^*+\delta\}\ell(\bX_{n},\balpha_{n})\right)^{2}\right]=O(n^{1-\epsilon})$.
\end{assumption}

The next assumption corresponds to the condition in Proposition \ref{prop: LF_sufficient} (or \ref{prop: PN_sufficient}) specialized for the quantile estimation problem. Observe that $F(\bX,\theta)$ in Proposition \ref{prop: LF_sufficient} corresponds to $\mathbf{1}\{h(\bX)\leq q\}$ in the quantile estimation case, which is monotone w.r.t. $q$. 
\begin{assumption}
\label{assu: QE_Feller}There exist $\delta_1,\delta_2>0$ such that there exists $V(\bX)\geq \sup_{\left\Vert\bealpha-\bealpha^*\right\Vert\leq\delta_2}\mathbf{1}\{h(\bX)\leq q^*+\delta_1\}\ell(\bX,\bealpha)$ with $\mathds{E}_{\bX\sim P}[V(\bX)]< \infty$.  
\end{assumption}

The next assumptions about the smoothness of the variance and objective function are Assumption \ref{assu: Obj_regular} and the condition depicted in Theorem \ref{thm: SAA_CLT} specialized to the quantile estimation case. 
\begin{assumption}\label{assu: continuous_variance}
For $(q,\bealpha)$ in a neighborhood of $(q^*,\bealpha^*)$, the function 
$$\Var_{\bX\sim P_{\bealpha^*}}\left(\mathbf{1}\{h(\bX)\leq q\}\ell(\bX,\bealpha)\right)$$ is continuous in $(q,\bealpha)$.
\end{assumption}

\begin{assumption}
\label{assu: SAA_QE_obj}The distribution function $F_h(x)$ is differentiable
at $x=q^{*}$, and the density $f_h(x)$ is strictly positive at $x=q^{*}$.
\end{assumption}

With these, we have the consistency and asymptotic normality of the quantile estimator $\hat q_n$. 
\begin{theorem}{\bf (Consistency and asymptotic normality of empirical quantile with embedded adaptive IS)}.
\label{prop: SAA_QE_opt}Under Assumptions \ref{assu: SAA_QE_consistency}-\ref{assu: SAA_QE_obj}, suppose that the black-box function $I$ is continuous.
Then the quantile estimator $\hat q_n$ obtained from Algorithm \ref{alg:SAA_quantile} is strongly consistent, i.e., $\hat{q}_n\rightarrow q^*$ a.s., and asymptotically normal, i.e.,
\[
\sqrt{n}(\hat{q}_{n}-q^{*})\Rightarrow\mathcal{N}\left(0,\frac{\Var_{\bX\sim P_{\bealpha^*}}\left(\mathbf{1}\{h(\bX)\leq q^{*}\}\ell(\bX,\bealpha^*)\right)}{\left(f_h(q^{*})\right)^{2}}\right),
\]
where $\balpha^*=I(q^*)$.
\end{theorem}

In the proof of Theorem \ref{prop: SAA_QE_opt}, we first verify the conditions for invoking Theorem \ref{prop:SAA_consistency} to show strong consistency. Regarding asymptotic normality, instead of using Theorem \ref{thm: SAA_CLT} directly, we exploit the special structure of quantile estimation, and the required condition is slightly weaker (the counterpart of Assumption \ref{assu:(uniform-integrable-entropy)} is not required for Theorem \ref{prop: SAA_QE_opt}). While the idea follows generally from \cite{serfling2009approximation} Theorem A, Section 2.3.3, one notable difference is that to show the weak convergence of the empirical estimate of $F_h \left(q+t\sqrt{n^{-1}}\right)$, we need to use a triangular-array martingale CLT instead of the Berry-Essen bound in \cite{serfling2009approximation}. As a result, we also do not need to assume that the likelihood ratio has bounded third-order moment as required by Berry-Esseen.

\subsection{SA with Adaptive IS for Quantile Estimation}

Algorithm \ref{alg:SA_quantile} presents our procedure to embed IS in SA for quantile estimation where, as in Section \ref{sec: adaptiveIS}, we consider both the RM-SA quantile estimator $\hat q_n$ and the PR-SA quantile estimator $\bar q_n$.
\begin{algorithm}[htb]
\caption{SA with adaptive importance sampling for quantile estimation}
\label{alg:SA_quantile}
\begin{algorithmic}[1]
\Ensure Original sampling distribution $P$; initial IS parameter $\balpha_1$; initial quantile estimate $\hat q_0$; stepsize constant $\gamma$; prior information set $A$; initial iteration index $n=1$; black-box IS function $I$.
\While{stopping criteria not met}
\State Generate sample $\bX_n\sim P_{\bealpha_n}$, and calculate $h(\bX_n)$ and $\ell(\bX_n,\balpha_n)$;
\State Set $\gamma_n=\gamma/n^{\alpha}$ (usually $\alpha=1$ for RM-SA; $1/2<\alpha<1$ for PR-SA);
\State Update quantile estimate
\begin{equation}\label{eq: update_SA_QE_proj}
\hat{q}_{n}=\Pi_A\left[\hat{q}_{n-1}-\gamma_{n}\left(\mathbf{1}\{h(\bX_{n})\leq\hat{q}_{n-1})\}\ell(\bX_{n},\balpha_{n})-p\right)\right];
\end{equation}
\State Update IS parameter $\balpha_{n+1}=I(\hat q_n)$;
\State Set $n=n+1$;
\EndWhile
\lastcon{Quantile estimate $\hat q_n$ for RM-SA, or $\bar q_n=\sum_{i=1}^n \hat q_i/n$ for PR-SA.
}
\end{algorithmic}
\end{algorithm}

Let $v(q)=\mathds{E}_{\bX\sim P_{I(q)}}\left[\mathbf{1}\{h(\bX)\leq q\}\ell (\bX,I(q))^2\right]$.
The following is Assumption \ref{assu:. sup_L2} specialized to quantile estimation.
\begin{assumption} \label{assu: QE_Fabian_noise}
There exists a constant $C$ such that $v(\hat{q}_n)<C$. 
\end{assumption}

We get the following special case of Theorem \ref{prop: SA_CLT}.

\begin{theorem}\label{prop:Asynorm_quantile_RM}
{\bf (Consistency and asymptotic normality of RM-SA with embedded adaptive IS for quantile estimation).}
Under Assumptions \ref{assu: QE_Feller}, \ref{assu: continuous_variance} and \ref{assu: QE_Fabian_noise},
suppose that $A=[a,b]$ for $a<q^*<b$,  {$F_h$ is a differentiable function and $f_h(q^{*})>0$},
and $\gamma>{1}/{(2f(q))}$. Suppose further that the black-box function $I$ is continuous. Then $\hat{q}_n\rightarrow q^*$ a.s. and 
\[
\sqrt{n}(\hat{q}_{n}-q^{*})\Rightarrow\mathcal{N}\left(0,\frac{\gamma^{2}}{2\gamma f_h(q^*)-1}\Var_{\bX\sim P_{\bealpha^*}}\left(\mathbf{1}\{h(\bX)\leq q^{*})\}\ell(\bX,\balpha^*)\right)\right),
\]
where ${\balpha}^{*}=I(q^{*})$.
\end{theorem}

We also have asymptotic normality of our adaptive IS embedded in PR-SA for quantile estimation.
\begin{theorem}{\bf (Consistency and asymptotic normality of PR-SA with embedded adaptive IS for quantile estimation).}
\label{prop: QE_Polyak_General}
Suppose that $I$ is continuous. Under Assumptions \ref{assu: QE_Feller}, \ref{assu: continuous_variance} and \ref{assu: QE_Fabian_noise},
suppose that $A=[a,b]$ for $a<q^*<b$, {$F_h$ is a twice differentiable function and $f_h(q^{*})>0$}, then $\bar q_n\rightarrow q^*$ a.s. and
\[
\sqrt{n}(\bar q_n-q^*)\Rightarrow\mathcal{N}\left(0,\frac{\Var_{\bX\sim P_{\bealpha}}\left(\mathbf{1}\{h(\bX)\leq q^{*}\}\ell(\bX,\balpha^*)\right)}{\left(f_h(q^{*})\right)^{2}}\right),
\]
where $\bar q_n = \sum_{i=1}^n \hat q_i/n$ and ${\balpha}^{*}=I(q^{*})$.
\end{theorem}

Finally, we note that Algorithms \ref{alg:SAA_quantile} and \ref{alg:SA_quantile}, with IS outputs of the form $\mathbf{1}\{h(\bX)\leq q)\ell(\bX,\balpha)$, are designed for the case when $p$ is close to 0.
When $p$ is close to 1,
we should use outputs of the form $\mathbf{1}\{h(\bX)\geq q)\ell(\bX,\balpha)$, because the area $h(\bX)\geq q$ becomes more important. 
Correspondingly, we can replace the indicator function in Algorithms \ref{alg:SAA_quantile} and \ref{alg:SA_quantile} by $\mathbf{1}(h(\bX)\geq \hat q_n)$ and $p$ by $1-p$. All our theoretical results continue to hold, as we can view this case equivalently as simply adding a negative sign to $h(\bX)$.

\section{Multidimensional Setting}\label{sec:multidim}
In Sections \ref{sec: adaptiveIS} and \ref{sec:quantile}, we restricted our discussion to a one-dimensional root $\theta\in\mathbb{R}$. 
In this section, we generalize our developments to multidimensional settings. Most of these generalizations follow naturally, but one major new issue arises in comparing (asymptotic) performance, since the scalar measure of (asymptotic) variance is now replaced by a matrix. 
To make this comparison well-defined, we introduce a scalar-valued performance function $g$, and we measure errors in terms of $g(\hat{\betheta}_n)-g(\betheta^*)$, e.g., $g$ can be an approximation of the objective function in a considered optimization problem. 
With this, the best IS parameter would be an optimal solution to minimize the variance of the approximated objective function. 

We study the asymptotic variance using this performance function $g$. By the delta method, and recalling formulas \eqref{SAA_asymptotic_standard} and \eqref{SA_average_asymptotic_standard}, if $g$ is continuously differentiable, we have that for a fixed IS parameter $\bealpha$, 
\begin{equation*}\label{eq: SAA_asymptotic_fixed_IS_multidim}
    \sqrt{n}\left(g(\hat{\betheta}_n)-g(\betheta^*)\right)\Rightarrow \mathcal{N}\left(0, \nabla g (\betheta^*)^\top[\bJ(\btheta^*)]^{-\top}\Var_{X\sim P_{\bealpha} } \left( \bF(\bX,\btheta^*) \ell(\bX,\bealpha)\right)  [\bJ(\btheta^*)]^{-1}\nabla g(\betheta^*)\right),
\end{equation*}
where $\nabla g(\betheta)$ is the gradient of $g(\betheta)$, and recall that $\bJ(\betheta)= D\boldf(\betheta)/D\betheta$ is the Jacobian matrix of $\boldf(\betheta)$.

Since the variance now becomes one-dimensional, we can follow our developed one-dimensional approach and say that a good IS parameter $\balpha$ should minimize the quantity $$\nabla g (\betheta^*)^\top[\bJ(\btheta^*)]^{-\top}\Var_{X\sim P_{\bealpha} } \left( \bF(\bX,\btheta^*) \ell(\bX,\bealpha)\right)  [\bJ(\btheta^*)]^{-1}\nabla g(\betheta^*). $$  
Note that this quantity is exactly the variance when we use  $$\nabla g (\betheta^*)^\top[\bJ(\btheta^*)]^{-\top}\bF(\bX,\btheta^*)\ell(\bX,\bealpha)$$ as an IS estimator for the expectation $\mathds{E}_{\bX\sim P}\left[\nabla g (\betheta^*)^\top[\bJ(\btheta^*)]^{-\top}\bF(\bX,\btheta^*)\right]$. 
Similar to the one-dimensional case, this means if given $\betheta^*$ and $\bJ (\betheta^*)$, then what we want is simply a good sampler for this expectation estimation problem. Suppose we have already an available good IS for this problem, i.e., we know a function $I(\betheta, \bJ)$ that parameterizes a good IS to estimate $\mathds{E}_{\bX\sim P}\left[\nabla g (\betheta)^\top[\bJ]^{-\top}\bF(\bX,\btheta)\right]$ for given  $\betheta$ and $\bJ$. Then in Step 4 of Algorithms \ref{alg:SAA_rootfinding_blackbox} and \ref{alg: SA_rootfinding_blackbox}, we plug in $\hat{\betheta}_n$ and the estimate
\begin{equation*} 
\hat{\bJ}_n=\frac{D}{D\btheta}\left[\frac{1}{n}\sum_{i=1}^{n}\bF(\bX_{i},\hat{\btheta}_{n})\ell(\bX_{i},\balpha_i)\right]
\end{equation*} 
and use $\bealpha_{n+1}= I(\hat{\betheta}_n, \hat{\bJ}_n)$ to parameterize the IS in the next iteration. 
Algorithms \ref{alg:SAA_rootfinding_blackbox_multidim} and \ref{alg: SA_rootfinding_blackbox_multidim} in Appendix \ref{appx:algorithmsmd} provide the respective procedures of our adaptive IS embedded in SAA and SA for this multivariate setting.

We have consistency and asymptotic normality for adaptive IS in the multivariate setting as follows.
\begin{theorem}{\bf (Consistency and asymptotic normality of SAA with embedded adaptive IS (multivariate case))} (Full version in Theorem \ref{thm: SAA_CLT_multidim})
Under a multidimensional version of the assumptions in Theorem \ref{thm: SAA_CLT}, for the root estimate $\hat{\betheta}_n$ generated by the SAA algorithm with embedded adaptive IS, $\hat{\betheta}_n\rightarrow \betheta^*$ a.s. and
\[
\sqrt{n}\left(g(\hat{\btheta}_{n})-g(\btheta^{*})\right)\Rightarrow \mathcal{N}\left(\bzero,\nabla g(\btheta^{*})^{\top}\bV\nabla g(\btheta^{*})\right),
\]
where $\bV=\left(\bJ(\betheta^*)\right)^{-\top}\bSigma\left(\bJ(\betheta^*)\right)^{-1}$ and $\bSigma = \Var_{\bX\sim P_{\bealpha^{*}}}\left(\bF(\bX,\btheta^{*})\ell(\bX,\balpha^{*})\right)$.
\end{theorem}

\begin{theorem}{\bf (Consistency and asymptotic normality of RM-SA with embedded adaptive IS (multivariate case))}. (Full version in Theorem \ref{prop: SA_CLT_multidim})
Under a multidimensional version of the assumptions in Theorem \ref{prop: SA_CLT},
let $\bP$ be an orthogonal matrix such that 
\[
\gamma\bP^{\top}\bJ(\btheta^{*})\bP=\bLambda
\]
is diagonal. 
Then for the root estimate $\hat{\betheta}_n$ generated by the RM-SA algorithm with embedded adaptive IS, $\hat{\betheta}_n\rightarrow \betheta^*$ a.s. and
\[
\sqrt{n}\left(g(\hat{\btheta}_{n})-g(\btheta^{*})\right)\Rightarrow \mathcal{N}\left(\bzero,\nabla g(\btheta^{*})^{\top}\bP\bM\bP^{\top}\nabla g(\btheta^{*})\right),
\]
where $[\bM]_{ij}=\gamma^2[\bP^{\top}\bSigma\bP]_{ij}([\bLambda]_{ii}+[\bLambda]_{jj}-1)^{-1}$ and $\bSigma = \Var_{\bX\sim P_{\bealpha^{*}}}\left(\bF(\bX,\btheta^{*})\ell(\bX,\balpha^{*})\right)$.

\end{theorem}
\begin{theorem}{\bf (Consistency and asymptotic normality of PR-SA with embedded adaptive IS (multivariate case))}.
(Full version in Theorem \ref{prop: SA_average_CLT_multidim}). Under a multidimensional version of the assumptions in Theorem \ref{prop: SA_average_CLT}, then for the root estimate $\bar{\betheta}_n$ generated by PR-SA with embedded adaptive IS, we have $\bar{\betheta}_n\rightarrow\betheta^*$ a.s. and
\[
\sqrt{n}(g(\bar{\btheta}_{n})-g(\btheta^{*}))\Rightarrow \mathcal{N}\left(\bzero,\nabla g(\btheta^{*})^{\top}\bV\nabla g(\btheta^{*})\right),
\]
where $\bV=\left(\bJ(\betheta^*)\right)^{-\top}\bSigma\left(\bJ(\betheta^*)\right)^{-1}$ and $\bSigma = \Var_{\bX\sim P_{\bealpha^{*}}}\left(\bF(\bX,\btheta^{*})\ell(\bX,\balpha^{*})\right)$.
\end{theorem}

These results are natural generalizations of those in Section \ref{sec: adaptiveIS}. One new challenge in proving these results is to show the consistency of the Jacobian estimate $\hat{\bJ}_n\rightarrow\bJ(\betheta^*)$, and we also need an extra assumption (Assumption \ref{assu: GC_Jacobian}). 
The other required assumptions can be regarded as the component-wise generalizations of the assumptions in Section \ref{sec: adaptiveIS}. To avoid repetition, we defer the explicit formulations of these assumptions, theorems, and their proofs to Appendices \ref{appx:multi_assu_result} - \ref{appx:proof_prop_SAA_QE_opt}.

\section{Examples}\label{sec:example}
In this section, we consider two sets of examples. The first set comprises toy examples on extreme quantile estimation for a standard normal distribution, an exponential distribution, and a Pareto-tailed distribution. 
The second set considers estimation of VaR and CVaR for a financial portfolio.

\subsection{Toy Examples}
We consider quantile estimation for three toy examples: $Z\sim\mathcal{N}(0,1)$, 
$Z\sim\exp(\lambda)$, 
$P\{Z>x\} = x^{-\lambda}$. 
We describe theoretical analysis and numerical results for the first example here, with analogous analysis and results provided for the other two distributions in Appendix \ref{appx:add_example}.

\subsubsection{Theoretical Analysis.}
Suppose we want to use Monte Carlo samples of $Z\sim\mathcal{N}(0,1)$ to estimate its $p$-quantile, where $p$ is large. 
We use Algorithms \ref{alg:SAA_quantile} and \ref{alg:SA_quantile}. 
To specify the algorithms, we first give the choice of the IS class $P_{\alpha}$ (here we are using $\alpha$ instead of $\balpha$, as it will be seen that $\alpha$ is one-dimensional for this example), the black-box IS function $I(q)$, and the truncation scheme $A_n$. 
As is well-known in the IS literature (e.g., \citealt{bucklew2013introduction}), a natural choice of $P_{\alpha}$ is the set of IS samplers derived from exponential shifting, where in the normal case this would be a normal distribution with mean $\alpha$. Moreover, for a given $q$, a good IS estimator for $P(Z\geq q)$ is to set its mean at $q$, i.e., we choose $I(q)=q$. To complete the algorithm specification, we select the truncation sets $\{A_n\}$. One simple way to do this is to estimate a lower and upper bound for $q^*$ using some concentration inequality, and from a knowledge $q^*\in[q_{\min},q_{\max}]$ we can let $A_n = [q_{\min},q_{\max}]$. If we do not have this bound, then another way is to let $A_n$ grow to $\mathbb{R}$. 
The next proposition gives some conditions for which the required assumptions of our theoretical results would be satisfied. 
 
\begin{proposition}\label{prop: verification_normal}
(i) When $A_n=[-\sqrt{\log(an^{1-\epsilon)}},\sqrt{\log(an^{1-\epsilon)}}]$ or $A_n = [q_{\min},q_{\max}]$, the assumptions for Theorem \ref{prop: SAA_QE_opt} hold.
(ii) When $A=[q_{\min},q_{\max}]$, the assumptions for Theorems \ref{prop:Asynorm_quantile_RM} and \ref{prop: QE_Polyak_General} hold. 
\end{proposition}

From Theorems \ref{prop: SAA_QE_opt}, \ref{prop:Asynorm_quantile_RM} and \ref{prop: QE_Polyak_General}, both SAA and PR-SA exhibit the asymptotic variance
\begin{equation}\label{eq: normal_asymptotic_formula}
\frac{\Var_{Z\sim P_{q^{*}}}(\mathbf{1}\{Z\geq q^{*}\}\ell(Z,q^{*}))}{(\phi(q^{*}))^{2}},
\end{equation}
where $q^*$ is the true quantile, $\ell(Z,\alpha)={\exp(-{Z^2}/{2}})/
{{\exp(-{(Z-\alpha)^2}/{2}})} = \exp(-\alpha Z+\alpha^2/2)$ is the likelihood ratio, and $\phi(x)$ is the standard normal density.
RM-SA with stepsize $\gamma_{n}={\gamma}/{n}$ exhibits the asymptotic variance
\[
\frac{\gamma^{2}}{2\gamma \phi(q^{*})-1}\Var_{Z\sim P_{q^{*}}}(\mathbf{1}\{Z\geq q^{*}\}\ell(Z,q^{*})).
\]

We now analyze the variance reduction. The numerator of \eqref{eq: normal_asymptotic_formula} is bounded by 
\begin{align*}
\mathds{E}_{Z\sim \mathcal{N}(q^{*},1)}\left[\left(\mathbf{1}\{Z\geq q^{*}\}\ell(Z,q^{*})\right)^{2}\right] & =\int_{q^{*}}^{\infty}\frac{1}{\sqrt{2\pi}}e^{-\frac{(x-q^*)^{2}}{2}}e^{-2q^{*}x+q^{*2}}dx\\
 & =e^{q^{*2}}\int_{q^{*}}^{\infty}\frac{1}{\sqrt{2\pi}}e^{-\frac{(x+q^{*})^{2}}{2}}dx.
\end{align*}
When $q^{*}\geq0$, the RHS is bounded by $\exp({-q^{*2}})/2$, using a tail bound for the standard normal distribution $P(Z \geq t)\leq \exp({-t^2/2})/2$ for $t\geq 0 $.
So the asymptotic variance of $\sqrt{n}(\hat{q}_{n}-q^{*})$ is bounded
by ${\exp({-q^{*2}}})/{(2{\phi(q^{*})^{2}})}=\pi$ in SAA. 

On the other hand, if we use SAA without IS, then the variance of $\sqrt{n}(\hat{q}_n-q^*)$ is given by ${p(1-p)}/{\phi(q^*)^2}$. When $p$ is close to 1, we have 
\[
1-p = \int_{q^*}^{\infty}\frac{1}{\sqrt{2\pi}}e^{-\frac{x^{2}}{2}}dx,
\] 
so the variance 
\[
\frac{p(1-p)}{\phi(q^*)^2}\approx \frac{\int_{q^*}^{\infty}\frac{1}{\sqrt{2\pi}}e^{-\frac {x^{2}}{2}}dx}{\phi(q^*)^2}\geq \frac{\pi e^{-(q^*+1)^2/2}}{2 e^{-q^{*2}}} = \frac{\pi}{2}e^{q^{*2}/2-2q^*-1}
\]
would grow to infinity exponentially as $q^*$ goes to infinity, where the inequality used another tail bound for the standard normal distribution $P(Z \geq t)\geq \exp{\{-{(t+1)^2}/{2}\}}/2$. 
Comparing variances in the above cases, we see that our adaptive IS significantly reduces the asymptotic variance when $p$ is close to 1. 

\subsubsection{Numerical Experiments.}\label{sec:Norm_num_example}

We report numerical experiments to estimate the quantile of the standard normal distribution using our adaptive IS, using the following settings:
\begin{itemize}
\item For SAA, we set $\alpha_n =\Pi_{A_n} [\hat q_{n-1}]$, where $A_n = [-\sqrt{\log(an^{1-\epsilon})}, \sqrt{\log(an^{1-\epsilon})}]$ and $a=5,\epsilon=0.1$.
\item For RM-SA, we set the stepsize $\gamma_n = \gamma/n$ with $\gamma = 1/\phi(q^*)$ (the optimal choice of the stepsize parameter), where $\phi(x)$ is the standard normal density and $q^*$ is the true quantile, and the projection set $A=[0,5]$. 
\item For PR-SA, we set the stepsize $\gamma_n = \gamma/n^{0.9}$ with $\gamma = 1/\phi(q^*)$.
We average the estimates $\hat q_n$ beginning from the ($N_0+1$)th iteration, i.e., for $n=1,2,\ldots,N_0$, we set $\bar q_n = \hat q_n$; for $n=N_0+1,N_0+2,\ldots,N$, we set $\bar q_n = \sum_{n=N_0+1}^N\hat q_n/(N-N_0)$. We set $N_0=100$ and use the projection set $A=[0,5]$.
\end{itemize}

We also run the algorithms under the same setting but without using IS. We set $p=0.99,0.999,0.9999$, and vary the total number of simulation samples from $500$ to $500\times 2^8$ to estimate the quantiles. We repeat the procedure $200$ times to estimate the variance and mean squared error (MSE) of the estimated quantiles. Figures \ref{fig:normal_var}-\ref{fig:normal_mse} show the results. Tables \ref{tab:normal01}-\ref{tab:normal0001} further show their numerical details, and also the ratios of the variance of the without-IS estimator over the with-IS counterpart. 
\begin{figure}[H]
\vspace*{-10pt}
\begin{minipage}[t]{0.33\linewidth}
\centering
\includegraphics[width=5.5cm]{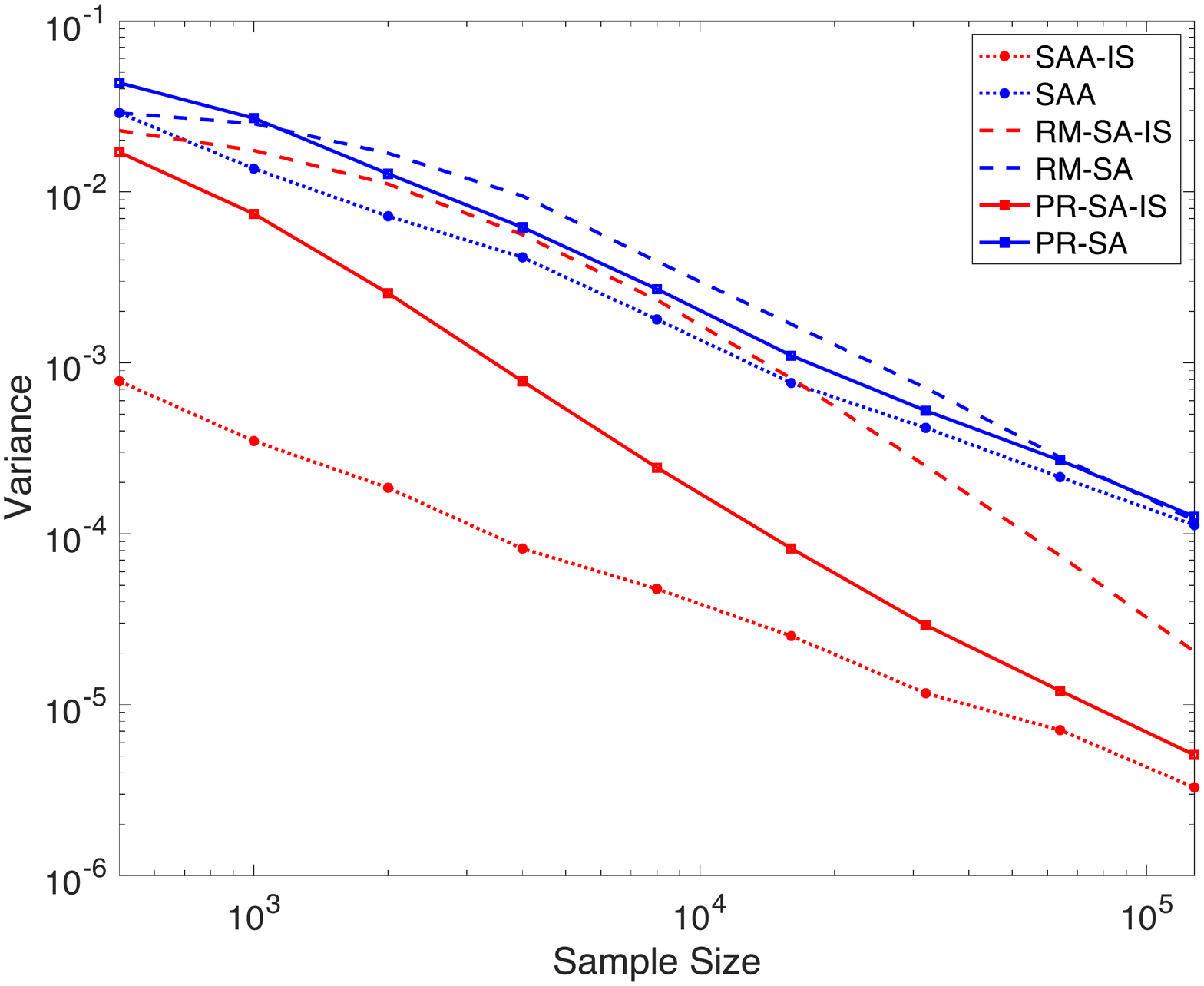}
\end{minipage}
\begin{minipage}[t]{0.33\linewidth}
\centering
\includegraphics[width=5.5cm]{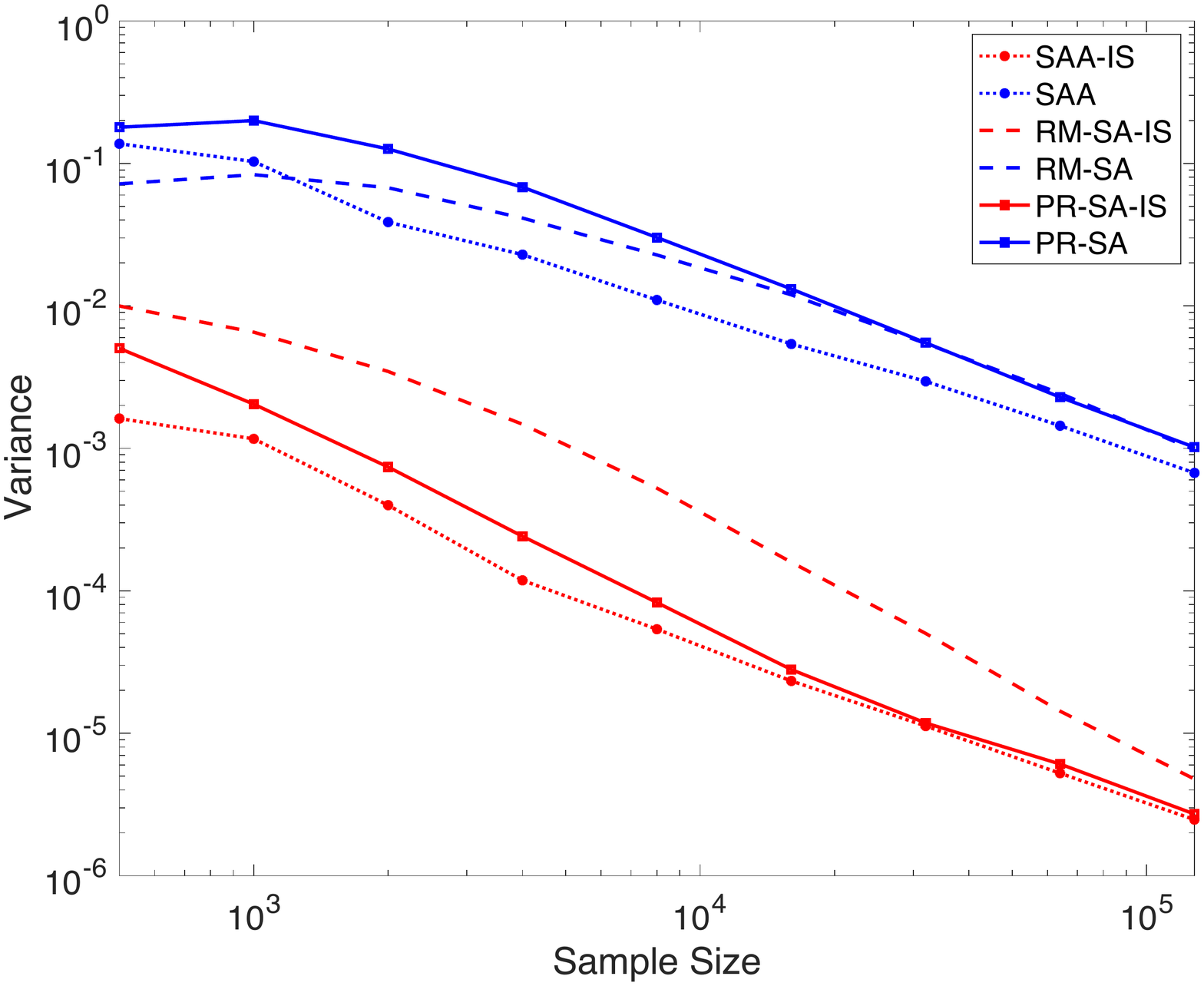}
\end{minipage}
\begin{minipage}[t]{0.32\linewidth}
\centering
\includegraphics[width=5.5cm]{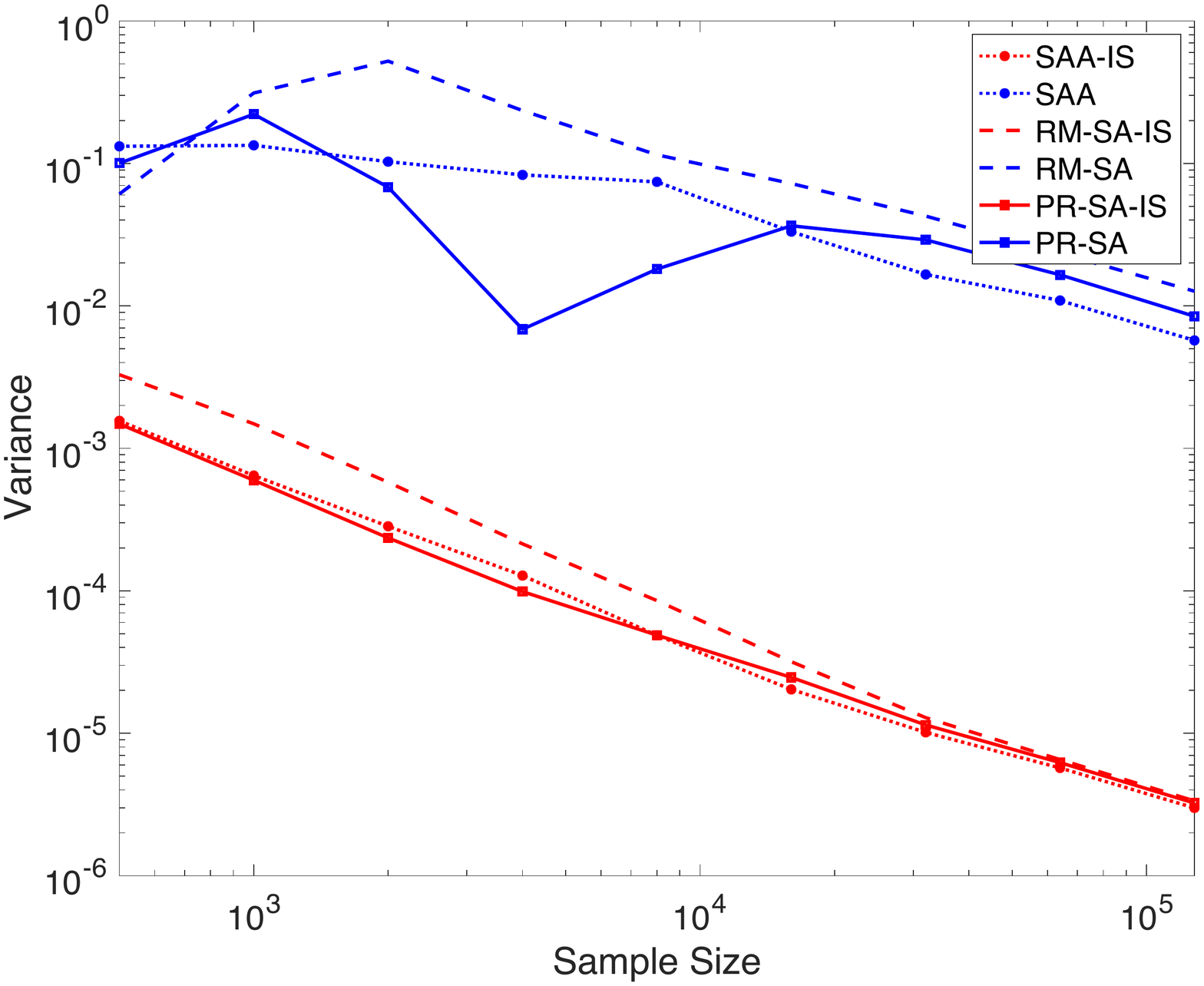}
\end{minipage}
\vspace*{-5pt}
\caption{Variance of SAA, RM-SA, PR-SA, with and without adaptive IS ($p=0.99$ for the left panel; $p=0.999$ for the middle panel; $p=0.9999$ for the right panel)} 
\label{fig:normal_var}
\end{figure}

\begin{figure}[H]
\vspace*{-10pt}
\begin{minipage}[t]{0.33\linewidth}
\centering
\includegraphics[width=5.5cm]{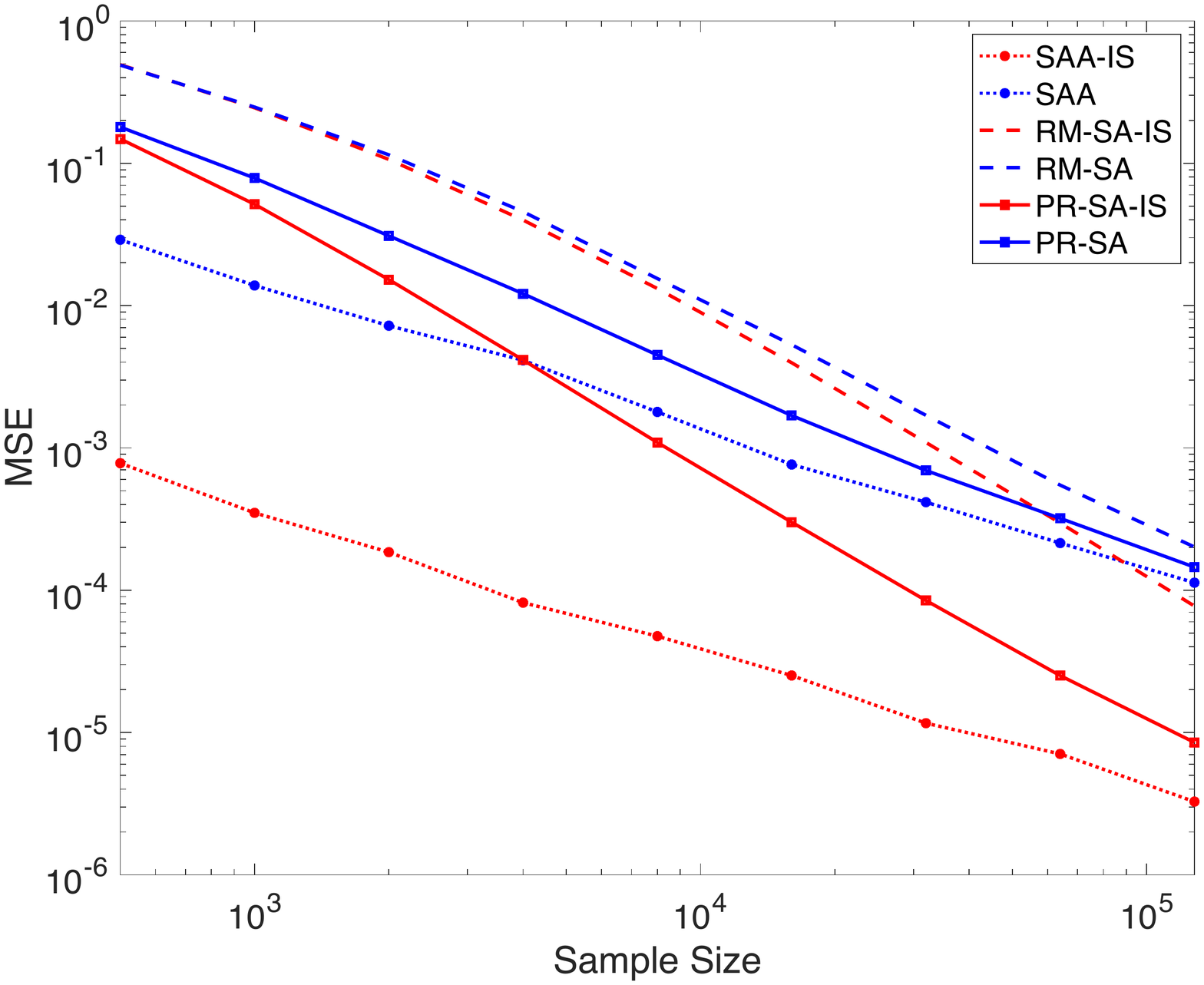}
\end{minipage}
\begin{minipage}[t]{0.33\linewidth}
\centering
\includegraphics[width=5.5cm]{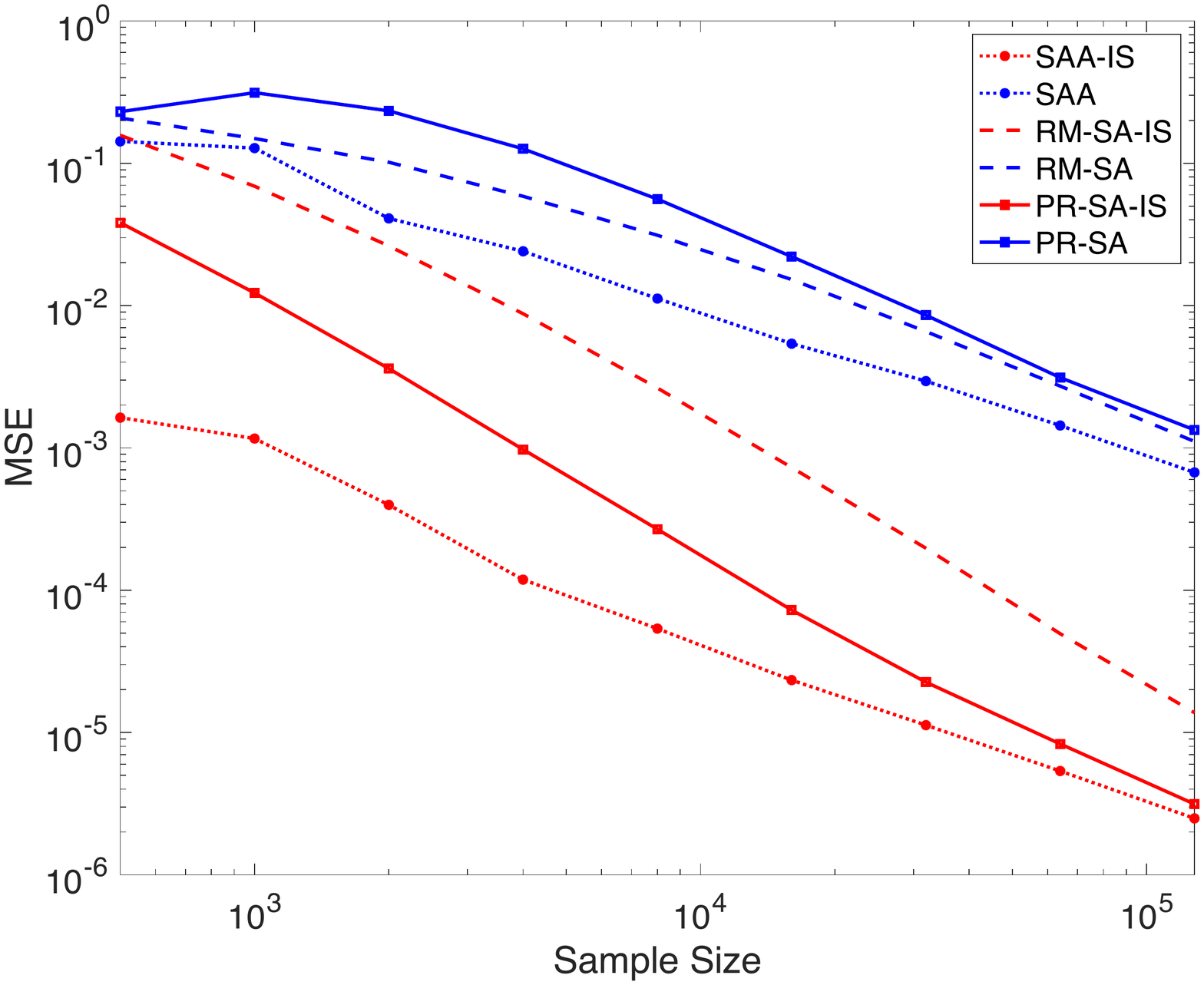}
\end{minipage}
\begin{minipage}[t]{0.32\linewidth}
\centering
\includegraphics[width=5.5cm]{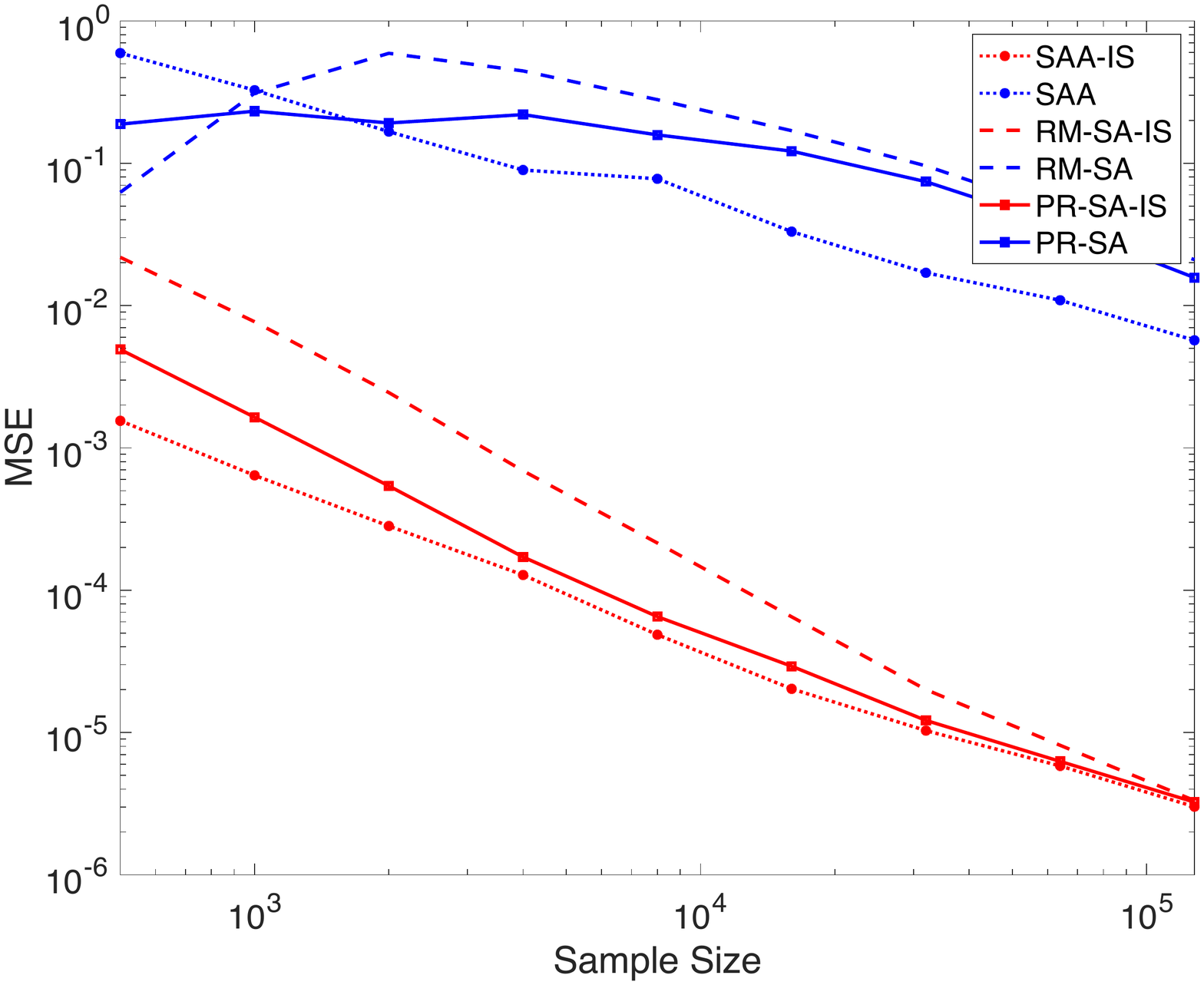}
\end{minipage}
\vspace*{-5pt}
\caption{MSE of SAA, RM-SA, PR-SA, with and without adaptive IS ($p=0.99$ for the left panel; $p=0.999$ for the middle panel; $p=0.9999$ for the right panel)}
\label{fig:normal_mse}
\end{figure}

\begin{table}[H]
\centering
\vspace*{-5pt}
\caption{Variance of SAA, RM-SA and PR-SA with and without adaptive IS ($p=0.99$)}
\label{tab:normal01}
\small
\begin{tabular}{c|l l l|l l l|l l l l|}
\toprule
 Sample Size     & SAA-IS   & SAA      & ratio  & RM-SA-IS    & RM-SA       & ratio  & PR-SA-IS   & PR-SA      & ratio \\
 \midrule
 500    & 7.80E-04 & 2.89E-02 & 37 & 2.28E-02 & 2.90E-02 & 1.3 & 1.70E-02 & 4.34E-02 & 2.5  \\
1000   & 3.49E-04 & 1.37E-02 & 39 & 1.75E-02 & 2.51E-02 & 1.4 & 7.44E-03 & 2.70E-02 & 3.6  \\
2000   & 1.86E-04 & 7.20E-03 & 39 & 1.11E-02 & 1.68E-02 & 1.5 & 2.55E-03 & 1.27E-02 & 5.0  \\
4000   & 8.19E-05 & 4.14E-03 & 51 & 5.62E-03 & 9.46E-03 & 1.7 & 7.80E-04 & 6.20E-03 & 7.9  \\
8000   & 4.76E-05 & 1.80E-03 & 38 & 2.34E-03 & 3.92E-03 & 1.7 & 2.44E-04 & 2.70E-03 & 11 \\
16000  & 2.52E-05 & 7.63E-04 & 30 & 8.11E-04 & 1.68E-03 & 2.1 & 8.20E-05 & 1.10E-03 & 13 \\
32000  & 1.17E-05 & 4.17E-04 & 36 & 2.49E-04 & 7.13E-04 & 2.9 & 2.92E-05 & 5.25E-04 & 18 \\
64000  & 7.10E-06 & 2.14E-04 & 30 & 7.45E-05 & 2.80E-04 & 3.8 & 1.21E-05 & 2.69E-04 & 22 \\
128000 & 3.28E-06 & 1.12E-04 & 34 & 2.05E-05 & 1.20E-04 & 5.9 & 5.08E-06 & 1.26E-04 & 25\\
\bottomrule
\end{tabular}
\vspace*{-5pt}
\end{table}

\begin{table}[H]
\centering
\vspace*{-5pt}
\caption{Variance of SAA, RM-SA and PR-SA with and without adaptive IS ($p=0.999$)}
\label{tab:normal001}
\small
\begin{tabular}{c|l l l|l l l|l l l l|}
\toprule
 Sample Size     & SAA-IS   & SAA      & ratio  & RM-SA-IS    & RM-SA       & ratio  & PR-SA-IS   & PR-SA      & ratio  \\  
\midrule
500    & 1.62E-03 & 1.37E-01 & 85  & 9.97E-03 & 7.18E-02 & 7.2   & 5.03E-03 & 1.79E-01 & 36  \\
1000   & 1.17E-03 & 1.03E-01 & 88  & 6.53E-03 & 8.33E-02 & 13  & 2.04E-03 & 2.00E-01 & 98  \\
2000   & 3.99E-04 & 3.87E-02 & 97  & 3.46E-03 & 6.74E-02 & 20  & 7.39E-04 & 1.26E-01 & 171 \\
4000   & 1.18E-04 & 2.29E-02 & 193 & 1.47E-03 & 4.13E-02 & 28  & 2.41E-04 & 6.81E-02 & 282 \\
8000   & 5.37E-05 & 1.10E-02 & 205 & 5.25E-04 & 2.28E-02 & 43  & 8.26E-05 & 3.02E-02 & 365 \\
16000  & 2.33E-05 & 5.40E-03 & 232 & 1.59E-04 & 1.19E-02 & 75  & 2.80E-05 & 1.31E-02 & 469 \\
32000  & 1.12E-05 & 2.95E-03 & 263 & 5.02E-05 & 5.44E-03 & 108 & 1.18E-05 & 5.50E-03 & 466 \\
64000  & 5.24E-06 & 1.44E-03 & 275 & 1.43E-05 & 2.43E-03 & 170 & 6.08E-06 & 2.28E-03 & 376 \\
128000 & 2.47E-06 & 6.71E-04 & 271 & 4.76E-06 & 9.75E-04 & 205 & 2.71E-06 & 1.02E-03 & 376  \\ 
 \bottomrule
\end{tabular}
\end{table}

\begin{table}[H]
\centering
\vspace*{-5pt}
\caption{Variance of SAA, RM-SA and PR-SA with and without adaptive IS ($p=0.9999$)}
\label{tab:normal0001}
\small
\begin{tabular}{c|l l l|l l l|l l l l|}
\toprule
 Sample Size     & SAA-IS   & SAA      & ratio  & RM-SA-IS    & RM-SA       & ratio  & PR-SA-IS   & PR-SA      & ratio  \\  
 \midrule
 500    & 1.56E-03 & 1.32E-01 & 85   & 3.29E-03 & 6.10E-02 & 19   & 1.48E-03 & 1.00E-01 & 68   \\
1000   & 6.44E-04 & 1.34E-01 & 208  & 1.49E-03 & 3.12E-01 & 210  & 5.97E-04 & 2.21E-01 & 371  \\
2000   & 2.83E-04 & 1.03E-01 & 363  & 5.78E-04 & 5.22E-01 & 903  & 2.35E-04 & 6.81E-02 & 290  \\
4000   & 1.28E-04 & 8.30E-02 & 649  & 2.13E-04 & 2.35E-01 & 1100 & 9.86E-05 & 6.85E-03 & 69   \\
8000   & 4.87E-05 & 7.42E-02 & 1524 & 8.53E-05 & 1.15E-01 & 1346 & 4.88E-05 & 1.81E-02 & 371  \\
16000  & 2.03E-05 & 3.32E-02 & 1635 & 3.18E-05 & 7.22E-02 & 2274 & 2.46E-05 & 3.65E-02 & 1480 \\
32000  & 1.01E-05 & 1.66E-02 & 1638 & 1.29E-05 & 4.27E-02 & 3316 & 1.14E-05 & 2.90E-02 & 2541 \\
64000  & 5.71E-06 & 1.09E-02 & 1909 & 6.53E-06 & 2.34E-02 & 3580 & 6.22E-06 & 1.65E-02 & 2654 \\
128000 & 2.99E-06 & 5.72E-03 & 1913 & 3.35E-06 & 1.27E-02 & 3786 & 3.23E-06 & 8.44E-03 & 2612 \\ 
 \bottomrule
\end{tabular}
\end{table}

We observe the following: 
(i) All three procedures (SAA, RM-SA and PR-SA) are improved by the adaptive IS, as clearly indicated in 
Figures \ref{fig:normal_var} and \ref{fig:normal_mse} by comparing the corresponding red and blue curves. Moreover, the variance reduction ratios (ratio between the variances of without-IS and with-IS estimators) increase quickly as $p$ approaches 1. For example, in Tables \ref{tab:normal01}-\ref{tab:normal0001} where $p$ takes $0.99,0.999$, and $0.9999$, respectively, fixing the sample size as $128000$, the variance reduction ratio for SAA is $34,~271$, and $1913$, respectively. 
(ii) SAA generally has the smallest variance of the three procedures, 
and the variance of PR-SA is generally smaller than that of RM-SA. This can be seen by comparing the red curves in Figure \ref{fig:normal_var}). More precisely, in Table \ref{tab:normal01} when $p=0.99$, fixing the sample size as 128000 for instance, the variance of SAA-IS is about {16\% (3.28E-06/2.05E-05)} of the variance of RM-SA-IS and the variance of PR-SA-IS is about {24\% (5.08E-06/2.05E-05)} of the variance of RM-SA-IS. When $p=0.9999$, the variances of the three procedures are closer. In Table \ref{tab:normal0001}, again fixing the sample size at 128000 for instance, the variance of SAA-IS is about {89\% (2.99E-06/3.35E-06)} of the variance of RM-SA-IS and the variance of PR-SA-IS is about {96\% (3.23E-06/3.35E-06)} of the variance of RM-SA-IS. 
(iii) When $p$ is very close to 1 (e.g., $p=0.999$ and $p=0.9999$), the variance reduction ratios  increase with the sample size up to some point and then plateau. 
This observation is clearer for SAA. For example, in Tables {\ref{tab:normal001} and \ref{tab:normal0001}}, each ratio column tends to increase at first as the sample size increases. When the sample size is more than $64000$, the ratios appear roughly unchanged. This could be attributed to that with a large sample size, the IS sampler stabilizes around the optimum and additional samples provide negligible variance reduction improvements.  

Furthermore, we compare the adaptive IS with the fixed optimal IS where we assume $q^*$ is known and use the IS parameter
\begin{equation*}
\alpha^* = \mathop{\arg\min}_\alpha \Var_{Z\sim P_\alpha}(\mathbf{1}\{Z>q^*\}\ell(Z,\alpha)).
\end{equation*}
This optimal IS is the best we can do among the considered IS class. Using this fixed optimal IS parameter, we run SAA, RM-SA and PR-SA using the same truncation set and stepsize as described in the beginning of this subsection. Figure \ref{fig:comp_norm} shows that all of SAA, RM-SA and PR-SA that embed our adaptive IS (represented by the red curves) give variances and MSEs very close to the counterparts with the fixed optimal IS (represented by the light blue curves). This indicates our adaptive IS achieves almost the same variance reduction as the fixed optimal IS, which coincides with the implications of Theorems \ref{thm: SAA_CLT}, \ref{prop: SA_CLT} and \ref{prop: SA_average_CLT} (see the discussions right after these theorems).

\begin{figure}[H]
\begin{minipage}[t]{0.48\linewidth}
\centering
\includegraphics[width=8cm]{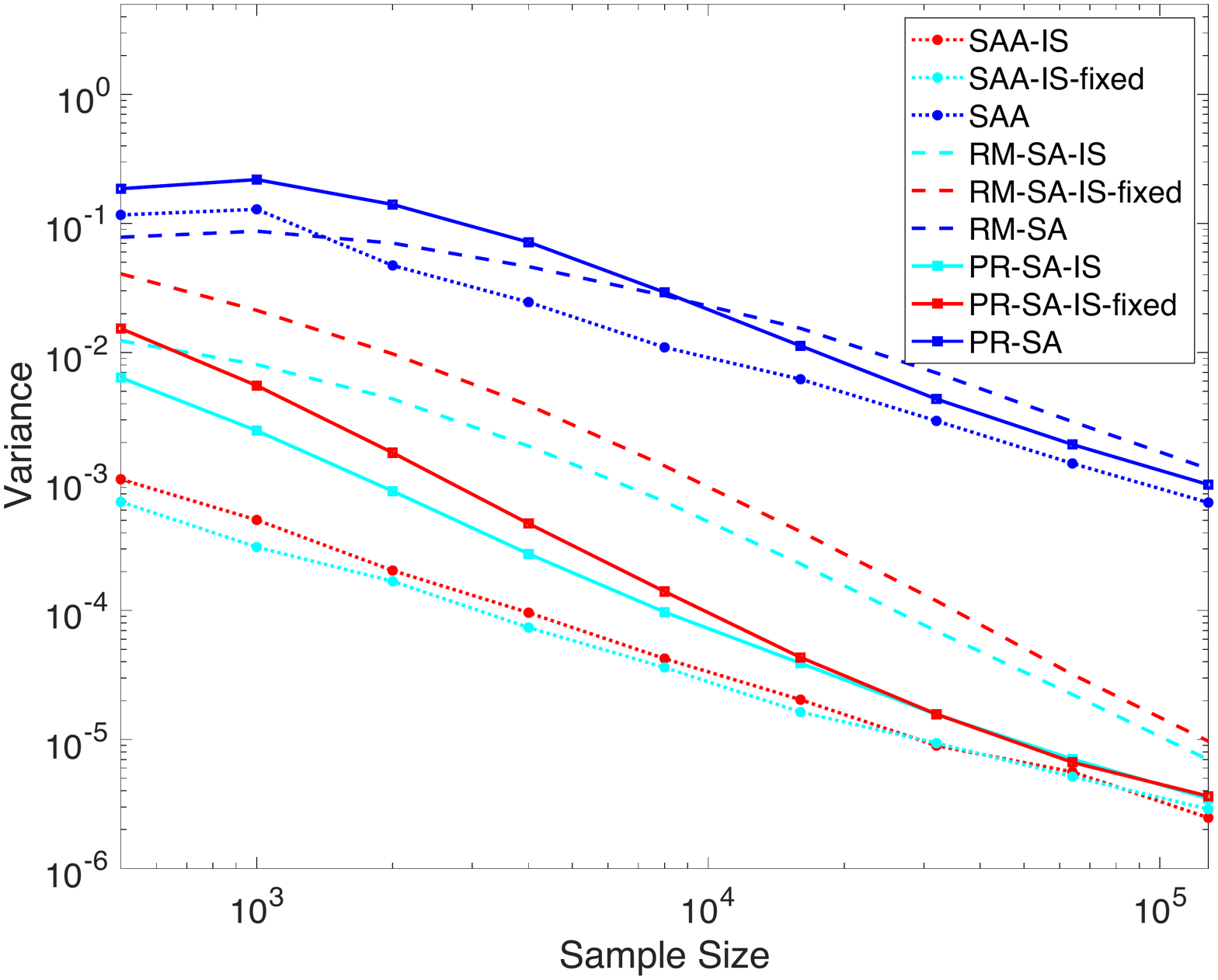}
\end{minipage}
\begin{minipage}[t]{0.48\linewidth}
\centering
\includegraphics[width=8cm]{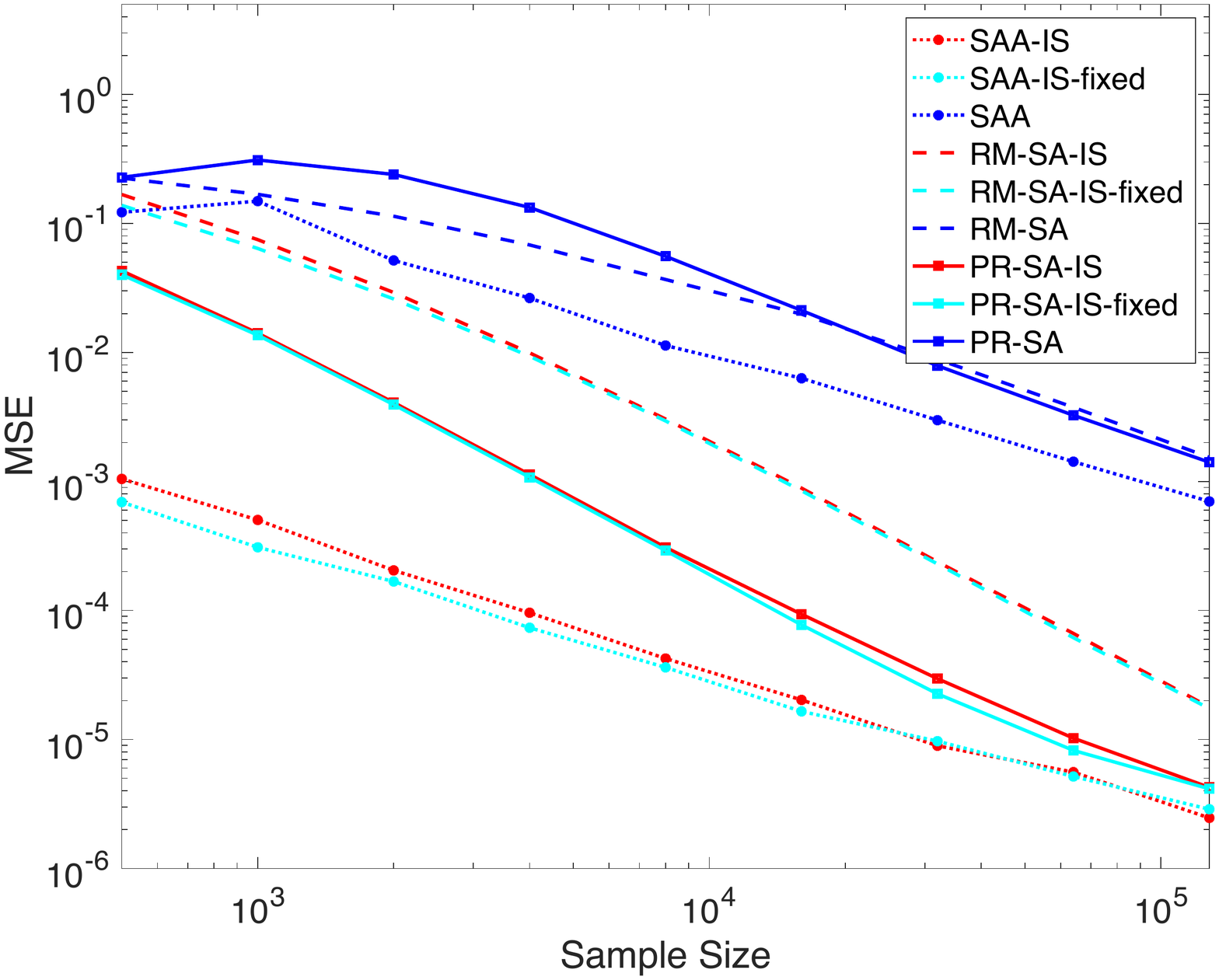}
\end{minipage}
\caption{Comparison of adaptive IS and fixed optimal IS for different approaches in Example 1 with $p=0.999$ (Variance in the left panel; MSE in the right panel)}
\label{fig:comp_norm}
\end{figure}

\subsection{VaR and CVaR for a Financial Portfolio}\label{sec:numerical_financial}
In this example, we estimate the VaR and CVaR of a financial portfolio. 
Let $\Phi(\bS(t),t)$ be the value of a portfolio at time $t$, where $\bS(t)=[S_1(t),\ldots,S_m(t)]^\top$ are $m$ risk factors, 
and $L(t) \equiv L(\bS(t),t)= \Phi(\bS(0),0) - \Phi(\bS(t),t)$ be the loss of this portfolio at time $t$. The $p$-VaR is the $p$-quantile of $L$:
\begin{equation*}
v_p\equiv\VaR_p(L(t)) = \inf\{x: P\{L(t)\leq x\}\geq p\}.
\end{equation*}

The $p$-CVaR is the expected return on the portfolio in the worst $p\%$ of cases, i.e., 
\begin{eqnarray*}
c_p&\equiv&\CVaR_p(L(t)) = \mathds{E}[L(t) | L(t)\geq \VaR_{p}(L(t))] \\
&=& \frac{1}{1-p}\int_{\VaR_p(L(t))}^{+\infty} yf_{L(t)}(y) dy =\frac{1}{1-p}\int_{p}^1 \VaR_{\beta}(L(t)) d\beta.
\end{eqnarray*} 

If $L$ has a positive density in the neighborhood of $v_p$, then $p$-CVaR  and $p$-VaR are related by the following equation \citep{RoUry2000}:
\begin{equation}\label{eq:CVaR_quantile}
c_p= v_p+\frac{1}{1-p}\mathds{E}\left[ (L(t) - v_p)^+\right].
\end{equation}

\subsubsection{Choice of IS Sampler and Other Algorithmic Configurations.}
We first specify an initial IS to estimate $P(L(t)>x)$. We use the sampler suggested by \cite{GlHeSh2000}. Let $\Delta \bS(t)= \bS(t) - \bS(0)$. To simplify notation, henceforth we drop the dependency on time $t$ in $L(t)$ and $\Delta \bS(t)$. When $t$ is small, the loss of the portfolio can be approximated by
\begin{equation*}\label{eq:Lapprox}
L\approx a_0 + \ba'\Delta\bS + \Delta\bS^\top \bA \Delta\bS \equiv a_0 + Q,
\end{equation*}
where $a_0 = -\Theta t$, $\ba = -\bdelta$, and $\bA = \bGamma$. Here $\Theta$, $\bdelta$ and $\bGamma$ are the Greeks for this portfolio, i.e., $\Theta = \partial \Phi/\partial t$  is the partial derivative of $\Phi$ on $t$, $\bdelta=[\delta_1,\ldots,\delta_m]^\top$ is the gradient of $\Phi$ on $\bS$ with $\delta_i = \partial \Phi / \partial S_i$, and $\bGamma$ is the Hessian matrix of $\Phi$ on $\bS$ with the $(i,j)$th element $[\bGamma]_{i,j} = \partial^2 \Phi/\partial S_i \partial S_j$. 

Suppose that $\Delta \bS$ has a multivariate normal distribution with mean $\bzero$ and covariance matrix $\bSigma$.
Then following \cite{GlHeSh2000}, we can write $Q$  in a diagonalized quadratic form: 
\begin{eqnarray*}
Q &=&\ba^{\top} \Delta \bS+\Delta \bS^{\top} \bA \Delta \bS\stackrel{d}{=}\ba^{\top} \bC \bZ+\bZ^{\top} \bC^{\top} \bA \bC \bZ\\ &=&\bb^{\top} \bZ+\bZ^{\top} \bLambda \bZ \equiv \sum_{i=1}^{m}\left(b_{i} Z_{i}+\lambda_{i} Z_{i}^{2}\right),
\end{eqnarray*}
where $\bZ\sim\mathcal{N}(0,\bI_m)$, $\bC$ is a matrix such that $\bC\bC^{\top}=\bSigma $, $\bC ^{\top}\bA\bC=\Lambda$ is a diagonal matrix, $\bb^{\top}=\ba^{\top} \bC$, and $\lambda_i$ is the $i$th diagonal element of $\bLambda$.

The IS sampler is obtained through exponential twisting with likelihood ratio
\begin{equation*}
\ell\equiv\ell(Q,\alpha) =\exp \left(-\alpha Q+\psi(\alpha)\right) =\exp \left(-\alpha\left(\bb^{\top} \bZ+\bZ^{\top} \bLambda \bZ^{\top}\right)+\psi(\alpha)\right),
\end{equation*}
where $\psi(\alpha)$ is the logarithm of the moment generating function of $Q$. 
Note that such a change of measure is equivalent to changing the distribution of $\bZ$ from $\mathcal{N}(0,\bI_m)$ to $N(\bmu(\alpha),\bB(\alpha))$, where $\bB(\alpha)=(I-2 \alpha \bLambda)^{-1}$ and $\bmu(\alpha)=\alpha \bB(\alpha) \bb$. The second moment of the estimator using this sampler is given by 
\begin{equation}\label{eq:m2upbound}
m_{2}(x, \alpha)=\mathds{E}_{\alpha}\left[\mathbf{1}\{L>x\} \ell^{2}\right] \leq \exp \left(2 \psi(\alpha)-2 \alpha\left(x-a_{0}\right)\right),
\end{equation}
where $\mathds{E}_{\alpha}$ denotes expectation under IS with twisting parameter $\alpha$.

A good choice of $\alpha$ for the estimation of $P(L>x)$ should make $m_2(x,\alpha)$ small. However, finding the value of $\alpha$ to minimize $m_{2}(x, \alpha)$ is computationally expensive. Instead, \cite{GlHeSh2000} suggested minimizing the upper bound in \eqref{eq:m2upbound}, i.e., we choose
$$\alpha^*_x = I(x) =\arg\min_{\alpha}\exp \left(2 \psi(\alpha)-2 \alpha\left(x-a_{0}\right)\right).$$ 
From the convexity of $\psi$, we know that $I(x)$ satisfies the first-order condition
\begin{equation*}
\psi^{\prime}\left(I(x)\right)=x-a_{0}.
\end{equation*}
Solving this first-order condition yields $I(x)$, which will serve as our black-box IS function. 

For the truncation set needed in our algorithms, we can estimate an upper bound for $v_p$ using
\begin{equation*}
\Pr\{L>x\} =\mathds{E}_{\alpha}[\mathbf{1}\{L>x\} \ell] \leq \exp \left\{\psi(\alpha)-\alpha\left(x-a_{0}\right)\right\}, 
\end{equation*}
and a lower bound can be derived similarly. With these bounds that give $v_p\in[q_{\min},q_{\max}]$, we can now use our Algorithms \ref{alg:SAA_quantile} and \ref{alg:SA_quantile} with $A_n=A=[q_{\min},q_{\max}]$. The following proposition verifies the conditions needed to provide the asymptotic guarantees of our algorithms.
\begin{proposition}\label{prop: verification_pVaR}
(i) If $A_n = [I(q_{\min}),I(q_{\max})]$, then the assumptions for 
 Theorem \ref{prop: SAA_QE_opt} hold.
(ii) When $A=[q_{\min},q_{\max}]$, the assumptions for Theorems \ref{prop:Asynorm_quantile_RM} and \ref{prop: QE_Polyak_General} hold. 
\end{proposition}

Finally, from the estimator for $v_p$, we can estimate CVaR using Equation \eqref{eq:CVaR_quantile}.

\subsubsection{Numerical Results on a Portfolio with Ten Risk Factors.}\label{sc:example4}
We use the example in \cite{GlHeSh2000}. 
Suppose that there are $250$ trading days in one year, and we investigate the loss of the portfolio over $10$ days, i.e., the risk time period $t=0.04$ (years). 
We take the risk-free interest rate as $r=5\%$ and
assume that there are ten uncorrelated underlying assets (risk factors), with all assets having the same initial value of $100$ and the same volatility of $0.3$. 
The portfolio shorts $10$ at-the-money call options and $5$ at-the-money put options on each asset, all options having a half-year maturity.
The Greeks $\Theta$, $\bdelta$, and $\bGamma$ can be calculated through the Black-Scholes formula, and then the parameters $a_0$, $\bb$, $\bLambda$, and $\bSigma$ can be calculated accordingly. 



We first estimate the $p$-VaR $v_p$. 
We set $p=0.999,0.9999$, and vary the total number of simulation samples from $500$ to $500\times 2^8$. {
We use the following configurations to run the algorithms:
\begin{itemize}
\item For SAA, we set $A_n=[-5,5]$ for both $p=0.999$ and $p=0.9999$.
\item For RM-SA, when $p=0.999$, we set the stepsize $\gamma_n = \gamma/n$ with $\gamma = 30000$ and the projection set $A=[-280,-100]$. When $p=0.9999$, we set $\gamma=400000$ and the projection set $A=[-310,-110]$.
\item For PR-SA, we set the stepsize $\gamma_n = \gamma/n^{0.9}$, and the choices of $\gamma$ and the projection sets are the same as in RM-SA. Similar to the toy example, we average the estimates beginning from the $(N_0+1)$th iteration with $N_0=100$.
\end{itemize}
}

The samples of the loss of the portfolio $\{L_i = \Phi(\bS_i(0),0)-\Phi(\bS_i(t),t)\}_{i=1}^n$ are calculated through the Black-Scholes formula based on the approximated $\bS_i(t)$, where $\bS_i(t)$ represents the $i$th sample path of these ten underlying assets.
We repeat the procedure $200$ times to estimate the variance of the estimated $p$-VaRs obtained from the three procedures. 
The results are shown in Figure \ref{fig:port1_var} and Tables \ref{tab:port1_999} and \ref{tab:port1_9999}. 

\begin{figure}[H]
\begin{minipage}[t]{0.49\linewidth}
\centering
\includegraphics[width=8.0cm]{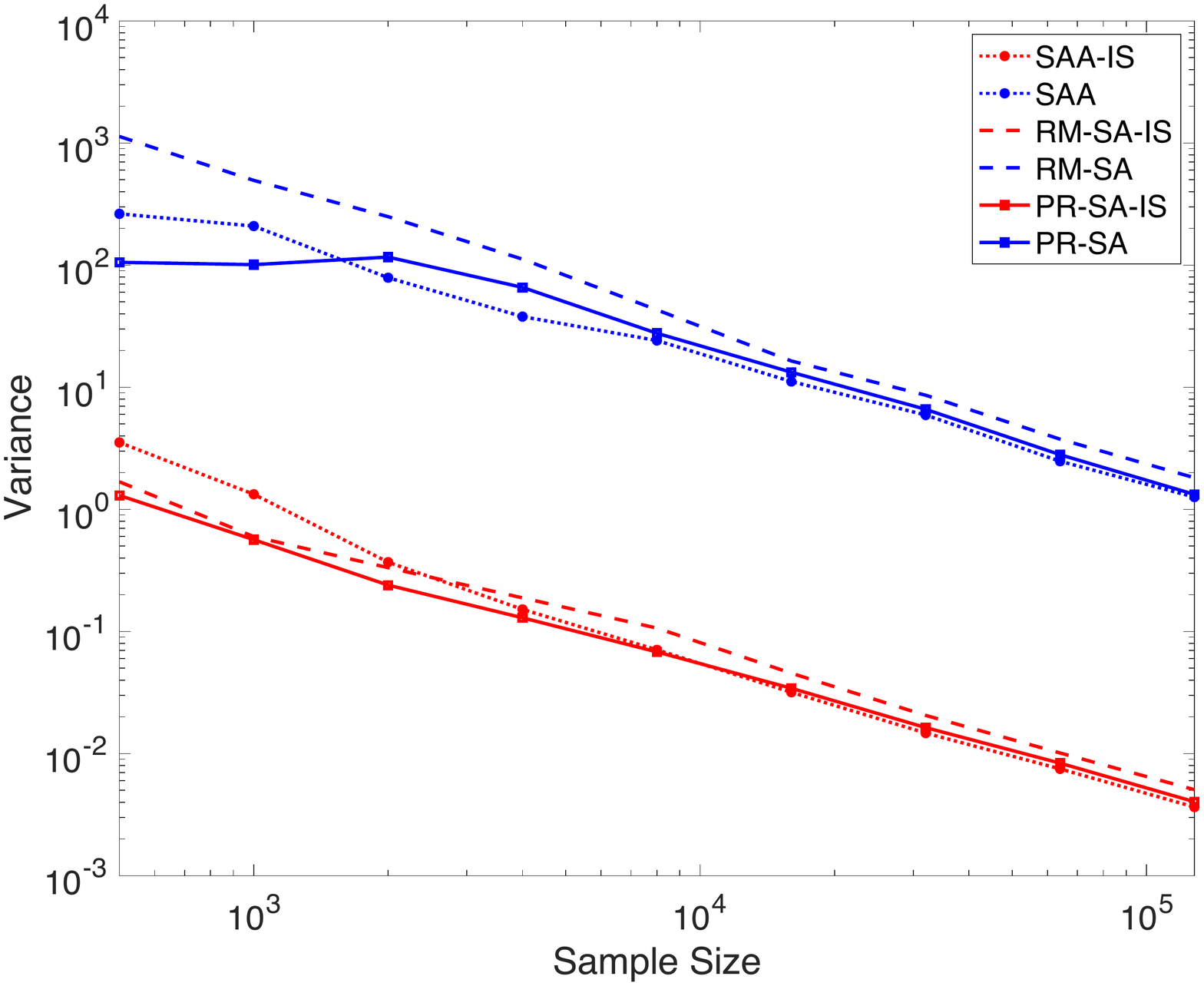}
\end{minipage}
\begin{minipage}[t]{0.49\linewidth}
\centering
\includegraphics[width=8.0cm]{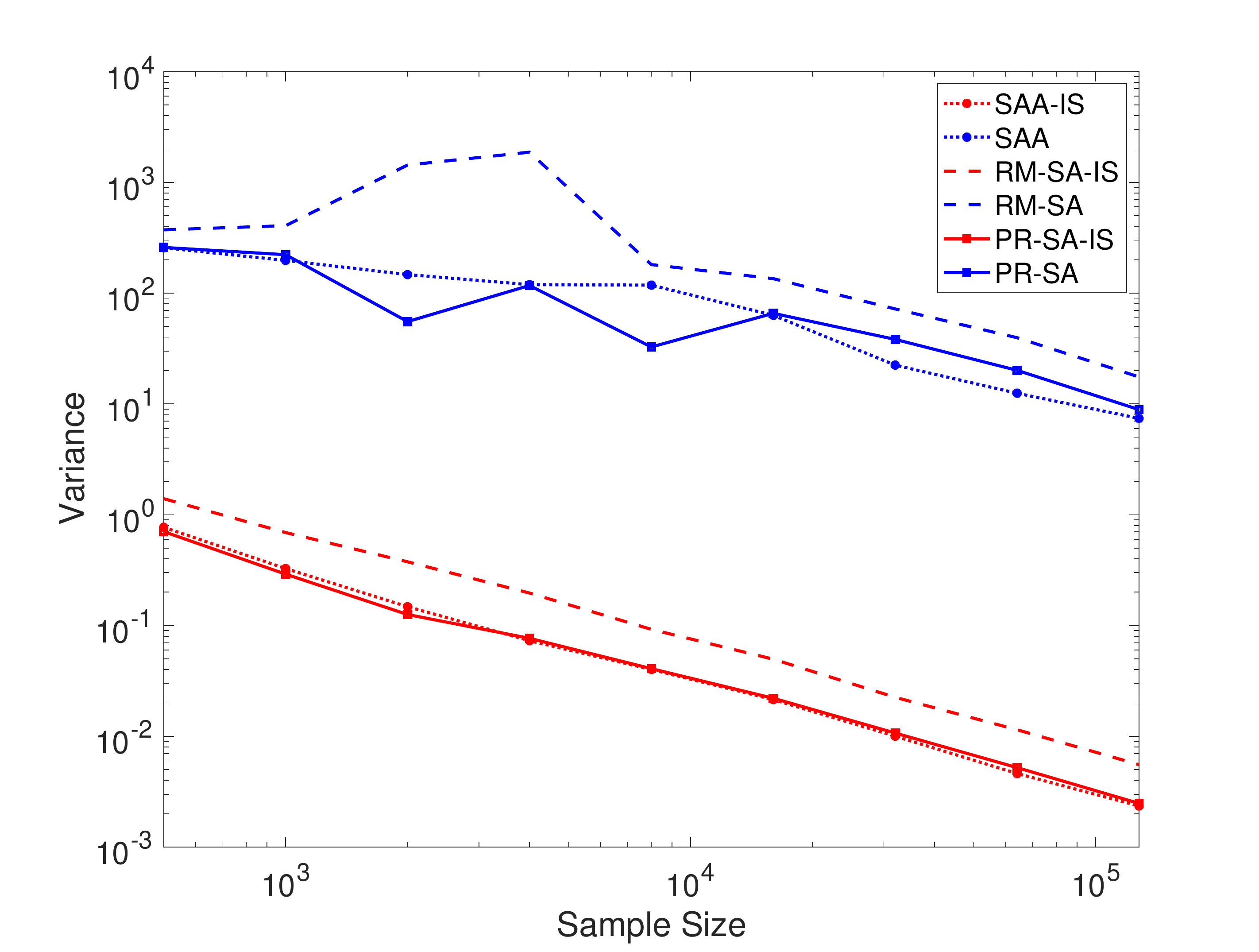}
\end{minipage}
\caption{Variance of SAA, RM-SA, PR-SA, with and without adaptive IS ($p=0.999$ for the left panel; $p=0.9999$ for the right panel)}
\label{fig:port1_var}
\end{figure}

\begin{table}[H]
\centering
\caption{Variance of SAA, RM-SA and PR-SA with and without adaptive IS ($p=0.999$)}
\label{tab:port1_999}
\small
\begin{tabular}{c|l l l|l l l|l l l l|}
\toprule
 Sample Size     & SAA-IS   & SAA      & ratio  & RM-SA-IS    & RM-SA       & ratio  & PR-SA-IS   & PR-SA      & ratio  \\  
 \midrule
 500    & 3.53E+00 & 2.63E+02 & 75  & 1.68E+00 & 1.13E+03 & 673 & 1.30E+00 & 1.05E+02 & 81  \\
1000   & 1.33E+00 & 2.09E+02 & 157 & 5.98E-01 & 4.94E+02 & 826 & 5.64E-01 & 1.01E+02 & 179 \\
2000   & 3.69E-01 & 7.89E+01 & 214 & 3.33E-01 & 2.49E+02 & 749 & 2.40E-01 & 1.16E+02 & 486 \\
4000   & 1.52E-01 & 3.79E+01 & 250 & 1.88E-01 & 1.12E+02 & 594 & 1.29E-01 & 6.57E+01 & 508 \\
8000   & 7.07E-02 & 2.41E+01 & 341 & 1.07E-01 & 4.30E+01 & 404 & 6.80E-02 & 2.77E+01 & 407 \\
16000  & 3.18E-02 & 1.11E+01 & 350 & 4.55E-02 & 1.64E+01 & 360 & 3.43E-02 & 1.33E+01 & 387 \\
32000  & 1.47E-02 & 5.92E+00 & 402 & 2.05E-02 & 8.60E+00 & 419 & 1.63E-02 & 6.60E+00 & 404 \\
64000  & 7.49E-03 & 2.48E+00 & 331 & 1.01E-02 & 3.76E+00 & 371 & 8.35E-03 & 2.80E+00 & 336 \\
128000 & 3.65E-03 & 1.26E+00 & 345 & 5.06E-03 & 1.81E+00 & 358 & 4.03E-03 & 1.33E+00 & 329\\ 
 \bottomrule
\end{tabular}
\end{table}
\vspace{-0.5cm}
\begin{table}[H]
\centering
\caption{Variance of SAA, RM-SA and PR-SA with and without adaptive IS ($p=0.9999$)}
\label{tab:port1_9999}
\small
\begin{tabular}{c|l l l|l l l|l l l l|}
\toprule
 Sample Size     & SAA-IS   & SAA      & ratio  & RM-SA-IS    & RM-SA       & ratio  & PR-SA-IS   & PR-SA      & ratio  \\  
 \midrule
 500    & 7.70E-01 & 2.57E+02 & 333  & 1.40E+00 & 3.72E+02 & 267  & 7.04E-01 & 2.58E+02 & 367  \\
1000   & 3.27E-01 & 1.97E+02 & 603  & 6.91E-01 & 4.06E+02 & 588  & 2.91E-01 & 2.22E+02 & 764  \\
2000   & 1.48E-01 & 1.47E+02 & 993  & 3.76E-01 & 1.43E+03 & 3807 & 1.26E-01 & 5.53E+01 & 440  \\
4000   & 7.30E-02 & 1.19E+02 & 1633 & 1.96E-01 & 1.87E+03 & 9517 & 7.67E-02 & 1.17E+02 & 1526 \\
8000   & 4.01E-02 & 1.18E+02 & 2940 & 9.24E-02 & 1.81E+02 & 1958 & 4.08E-02 & 3.26E+01 & 800  \\
16000  & 2.15E-02 & 6.30E+01 & 2937 & 4.97E-02 & 1.35E+02 & 2717 & 2.21E-02 & 6.57E+01 & 2978 \\
32000  & 1.00E-02 & 2.24E+01 & 2233 & 2.25E-02 & 7.20E+01 & 3199 & 1.07E-02 & 3.82E+01 & 3568 \\
64000  & 4.62E-03 & 1.25E+01 & 2702 & 1.15E-02 & 3.96E+01 & 3452 & 5.32E-03 & 2.01E+01 & 3776 \\
128000 & 2.35E-03 & 7.40E+00 & 3147 & 5.56E-03 & 1.75E+01 & 3155 & 2.48E-03 & 8.89E+00 & 3580 \\ 
 \bottomrule
\end{tabular}
\end{table}

From Figure \ref{fig:port1_var} and Tables \ref{tab:port1_999} and \ref{tab:port1_9999}, we reach similar conclusions as in Section \ref{sec:Norm_num_example}: (i) Our adaptive IS can reduce the variance of the VaR estimator significantly for SAA, RM-SA and PR-SA. For example, from Table \ref{tab:port1_9999}, we see that when $p=0.9999$ and the sample size is 128000, the variance reduction ratio is more than 3000 for all of the three procedures. (ii) SAA has the smallest variance and PR-SA has smaller variance than RM-SA. For example, in Table \ref{tab:port1_9999}, when $p=0.9999$, the variance of SAA-IS is about 42\% (2.35E-03/5.56E-03) of the variance of RM-SA-IS and the variance of PR-SA-IS about 45\% (2.48E-03/5.56E-03) of the variance of RM-SA-IS. 
(iii) For both SAA and PR-SA, the variance reduction ratios tend to increase as sample size increases and then start to plateau when the sample size is about 32000. This can be seen from the ratio column for SAA and PR-SA in Table \ref{tab:port1_999} or \ref{tab:port1_9999}. For RM-SA, the ratio is less stable. In Table \ref{tab:port1_9999}, when the sample size is 4000, the ratio for RM-SA is 9517, which appears to be an outlier, as it is quite large compared to the other ratios; however, when the sample size becomes larger (larger than 32000), the ratio tends to become stable.

Next, we estimate CVaRs using the estimated $v_p$. Since we do not have an analytical expression for computing $\mathds{E}[(L-v_p)^+]$, we use $10^7$ samples to approximate the expectation, i.e., we generate $10^7$ samples of $L$ and then plug into the estimated $v_p$ using the three different procedures (SAA, RM-SA, PR-SA) and estimate the variances. By comparing the red curves with corresponding blue curves in  Figure \ref{fig:port1_cvar_var},  we see that our adaptive IS again substantially reduces the variances of all three procedures. 
%
\begin{figure}[H]
\begin{minipage}[t]{0.49\linewidth}
\centering
\includegraphics[width=8.0cm]{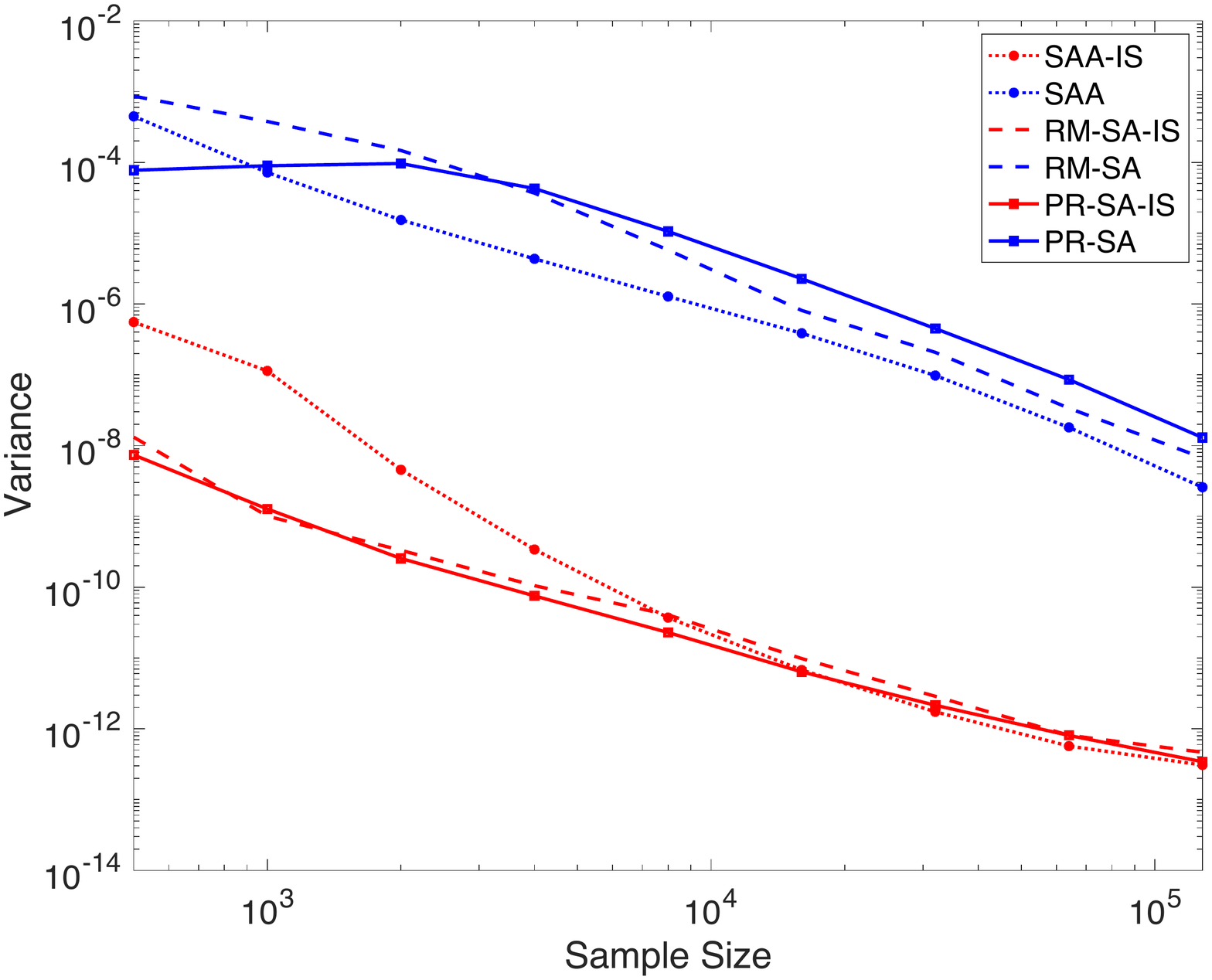}
\end{minipage}
\begin{minipage}[t]{0.49\linewidth}
\centering
\includegraphics[width=8.0cm]{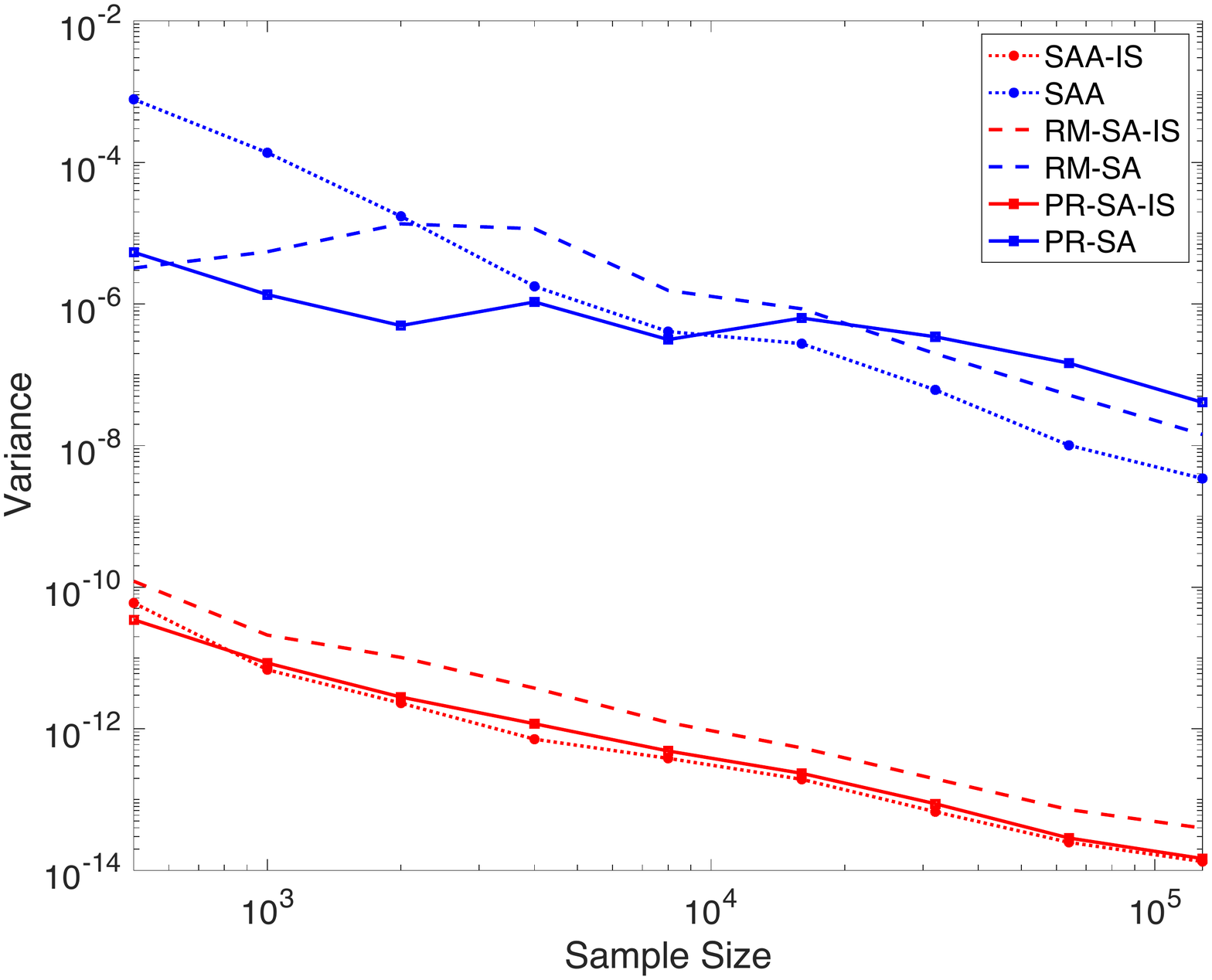}
\end{minipage}
\caption{Variance of CVaR for SAA, RM-SA, PR-SA, with and without adaptive IS ($p=0.999$ for the left panel; $p=0.9999$ for the right panel)}
\label{fig:port1_cvar_var}
\end{figure}

\section{Conclusion}\label{sec:concl}
In this paper, we propose an adaptive IS scheme to resolve the circular challenge in stochastic root-finding problems involving rare events. The circular challenge arises because, in the presence of rare events, variance reduction via IS is crucial to boost the estimation accuracy to an acceptable level, yet configuring a good IS relies on knowing the root, which is our a priori unknown target. We design an adaptive approach to simultaneously estimate the root and the IS parameters, and embed it in three commonly used root-finding procedures: SAA, RM-SA and PR-SA. 
We use a worst-case asymptotic variance comparison to show the benefit and necessity of adaptivity, and support our algorithms with theoretical analysis on strong consistency and asymptotic normality. 
We use extreme quantile estimation as a concrete example, and obtain the corresponding theoretical results under milder conditions than in the general setting. 
Finally, we demonstrate via numerical experiments the effectiveness of our adaptive IS.


\ACKNOWLEDGMENT{This material is based upon work supported by the National Science Foundation under Grants CAREER CMMI-1834710 and IIS-1849280, the U.S. Air Force Office of Scientific Research under Grant FA95502010211, and the National Natural Science Foundation of China under Grant 71801148.
A preliminary version of this work (\citealt{lam2018efficiencies}) was published in the 2018 Proceedings of the Winter Simulation Conference.}

\bibliographystyle{ormsv080}
\bibliography{SIS}

\ECSwitch
\ECHead{Appendices}
\setcounter{equation}{0}
\setcounter{lemma}{0}
\setcounter{algorithm}{0}
\setcounter{subsection}{0}
\setcounter{assumption}{0}
\setcounter{theorem}{1}
\setcounter{proposition}{0}
\renewcommand{\theequation}{A.\arabic{equation}}
\renewcommand{\thelemma}{A.\arabic{lemma}}
\renewcommand{\thesubsection}{\thesection.\arabic{subsection}}
\renewcommand{\thealgorithm}{A.\arabic{algorithm}}
\renewcommand{\theassumption}{A.\arabic{assumption}}
\renewcommand{\thetheorem}{A.\arabic{theorem}}
\renewcommand{\theproposition}{A.\arabic{proposition}}

\section{Multivariate Version of Algorithms and Theoretical Results}
In this section, we provide the multidimensional versions of algorithms for SAA and SA with adaptive IS. We also extend the assumptions under the one-dimensional setting to the multidimensional setting, and present the theoretical results including consistency and asymptotic normality under the multidimensional setting.
\subsection{Multidimensional Version of Algorithms \ref{alg:SAA_rootfinding_blackbox} and \ref{alg: SA_rootfinding_blackbox}}{\label{appx:algorithmsmd}} 

Recall that $\bF(\cdot,\cdot):\mathbb R^r\times\mathbb R^d \to\mathbb R^d$, $\bX=(X_1,X_2,\ldots,X_r)\in\mathbb R^r$, $\btheta\in \Theta \subseteq \mathbb R^{d}$, and $D/D\btheta$ is the Jacobian operator. Then we have the multidimensional version (i.e., $d\geq2$) of Algorithms \ref{alg:SAA_rootfinding_blackbox} and \ref{alg: SA_rootfinding_blackbox} presented in Algorithms \ref{alg:SAA_rootfinding_blackbox_multidim} and \ref{alg: SA_rootfinding_blackbox_multidim}.
\begin{algorithm}[htp]
\caption{SAA with adaptive importance sampling for stochastic root finding}
\label{alg:SAA_rootfinding_blackbox_multidim}
\begin{algorithmic}[1]
\Ensure Original sampling distribution $P$; initial IS parameter $\balpha_1$;  initial iteration index $n=1$; truncation sets $A_1\subset A_2 \subset \dots $; black-box IS function $I$.
\While{stopping criteria not met}
\State Generate sample $\bX_n\sim P_{\bealpha_n}$; 
\State Update root estimate $\hat \btheta_n$ by solving the equation 
\[
\frac{1}{n}\sum_{i=1}^n \bF(\bX_i,\betheta)\ell(\bX_i,\balpha_i)={\bold{c}};
\]
\State Update Jacobian estimate
\[
\hat{\bJ}_n=\frac{D}{D\btheta}\left[\frac{1}{n}\sum_{i=1}^{n}\bF(\bX_{i},{\hat{\btheta}_n})\ell(\bX_{i},\balpha_i)\right];
\]
\State Update IS parameter $\balpha_{n+1} = \Pi_{A_{n+1}}[I(\hat \btheta_n, \hat{\bJ}_n)]$ ;
\State Set $n=n+1$;
\EndWhile
\lastcon{Root estimate $\hat \btheta_n$.
}
\end{algorithmic}
\end{algorithm}


\begin{algorithm}[htb]
\caption{SA with adaptive importance sampling for stochastic root finding}
\label{alg: SA_rootfinding_blackbox_multidim}
\begin{algorithmic}[1]
\Ensure Original sampling distribution $P$; initial IS parameter $\balpha_1$; initial root $\hat \btheta_0 $; stepsize constant $\gamma$; prior information set $A$; initial iteration index $n=1$; black-box IS function $I$.
\While{stopping criteria not met}
\State Generate sample $\bX_n\sim P_{\bealpha_n}$, and calculate $\bF(\bX_n,\hat \btheta_{n-1})$ and $\ell(\bX_n,\balpha_n)$;
\State Set $\gamma_n=\gamma/n^{\alpha}$ (usually $\alpha=1$ for RM-SA; $1/2<\alpha<1$ for PR-SA);
\State Update root estimate
\begin{equation}\label{eq:SA_update_multidim}
\hat{\btheta}_{n}=\Pi_A\left[\hat{\btheta}_{n-1}-\gamma_{n}\left(\bF(\bX_n,\hat \btheta_{n-1})\ell(\bX_{n},\balpha_{n})-\bc\right)\right];
\end{equation}
\State Update Jacobian estimate
\[
\hat{\bJ}_n=\frac{D}{D\btheta}\left[\frac{1}{n}\sum_{i=1}^{n}\bF(\bX_{i},\hat{\btheta}_{n})\ell(\bX_{i},\balpha_i)\right];
\]

\State Update IS parameter $\balpha_{n+1} = {I(\hat \btheta_n, \hat{\bJ}_n)}$ ;
\State Set $n=n+1$;
\EndWhile
\lastcon{Root estimate $\hat \btheta_n$ for RM-SA, or $\bar \btheta_n=\sum_{i=1}^n \hat \btheta_i/n$ for PR-SA.
}
\end{algorithmic}
\end{algorithm}

\subsection{Assumptions and Theoretical Results for Multidimensional Case}\label{appx:multi_assu_result}
We first state the assumptions for the multidimensional case, and then provide the theoretical results of SAA and SA with adaptive IS. We will later provide the proofs of these theoretical results in Sections \ref{appx:prooflemma1}-\ref{appx:proofprop6}, to which the proofs under the one-dimensional setting (i.e., $\theta\in\mathbb{R}$) can be regarded as special cases.
Throughout the appendix, we write $\bF^{(k)}$ for the $k$-th dimension of $\bF$.

\subsubsection{SAA.}
\begin{assumption}\label{assu: truncation_alpha_Multidim}
For each $\btheta \in \Theta$ and each $k=1,2,\dots,d$, $\mathds{E}_{\bX\sim P_{\bealpha_n}}\left[\left(\bF^{(k)}(\bX,\btheta)\ell(\bX,\balpha_{n})\right)^{2}\right]=O(n^{1-\epsilon})$ holds for some $\epsilon>0$.
\end{assumption}
\begin{assumption}\label{assu: bracket_multidim}
Define $\sfF=\{f(\bX,\balpha):=\bF^{(k)}(\bX,\btheta)\ell(\bX,\balpha), k=1,2,\cdots,d,\btheta\in \Theta\}$.
For each $\epsilon>0$, there exists a finite set $K_{\epsilon}$ whose
elements are pairs of functions such that:\\
(1) For each $f\in\sfF$, there exists $(f_{L},f_{R})\in K_{\epsilon}$ such that $f_{L}\leq f\leq f_{R}$;\\
(2) For each pair of $(f_{L},f_{R})\in K_{\epsilon}$, the limits 
\begin{equation}
\lim_{n\rightarrow\infty}\frac{1}{n}\sum_{i=1}^{n}f_{L}(\bX_{i},\balpha_{i})~~\text{and}~~\lim_{n\rightarrow\infty}\frac{1}{n}\sum_{i=1}^{n}f_{R}(\bX_{i},\balpha_{i})
\end{equation}
exist, and
\begin{equation}\label{bracket_size}
\lim_{n\rightarrow\infty}\frac{1}{n}\sum_{i=1}^{n}\left[f_{R}(\bX_{i},\balpha_{i})-f_{L}(\bX_{i},\balpha_{i})\right]\leq\epsilon.
\end{equation}
\end{assumption}

\begin{assumption} \label{assu: Obj_regular_multidim}
The objective function $\boldf(\betheta)$ 
is differentiable at $\btheta=\btheta^{*}$ with continuously invertible Jacobian matrix, and $\betheta^*$ is the unique root in the sense that for any $\epsilon>0$,
\begin{equation*}
\inf_{|| \betheta-\betheta^{*}||\geq\epsilon}\left\Vert \boldf(\betheta)-\bc\right\Vert >0.
\end{equation*}
\end{assumption}
With the first three assumptions, we have strong consistency: 
\begin{theorem}[Consistency of SAA with embedded adaptive IS (multivariate case)] \label{prop:SAA_consistency_multidim}
Under Assumptions \ref{assu: truncation_alpha_Multidim} - \ref{assu: Obj_regular_multidim}, the root estimator generated by Algorithm \ref{alg:SAA_rootfinding_blackbox_multidim} is strongly consistent, i.e., $$\hat{\btheta}_{n}\rightarrow\btheta^{*}~a.s.$$
\end{theorem}

Next we show asymptotic normality, which requires an extra assumption that does not have its counterpart for the one-dimensional case. This assumption arises because for the multidimensional case, we need to estimate the Jacobian $\bJ(\betheta)=D\boldf(\betheta)/D\betheta$ at $\betheta=\betheta^*$. Let $\Theta_{\delta}=\{\betheta\in\Theta: \left\Vert\betheta-\betheta^*\right\Vert\leq\delta\}$, and $[\bA]_{ij}$ is the $(i,j)$th element of matrix $\bA$. Then we make the following assumption.
\renewcommand{\theassumption}{A.E.1}
\begin{assumption}\label{assu: GC_Jacobian}
Let $\sfF_{D,\delta}=\{f(\bX,\balpha):=\left[{D}\bF(\bX,\btheta)\ell(\bX,\balpha)/{D\btheta}\right]_{ij},\btheta\in\Theta_{\delta}, i,j=1,2,\cdots,d\}$. There exists a $\delta>0$ such that $\sfF_{D,\delta}$ satisfies
the conditions in Assumptions \ref{assu: bracket_multidim}.
\end{assumption}

Let
$f_{\betheta}^{(k)}$ be the function defined by $f_{\betheta}^{(k)}(\bX):=\bF^{(k)}(\bX,\btheta)$, and $\sfF^{(k)}_{\delta}=\{f_{\betheta}^{(k)},\btheta\in\Theta_{\delta}\}$. So from the definition of $V_{n,i}(g)$ we have that \[
V_{n,i}\left(f_{\betheta}^{(k)}\right)=\frac{\bF^{(k)}(\bX_{i},\btheta)\ell(\bX_{i},\balpha_{i})-\mathds{E}_{\bX\sim P}[\bF^{(k)}(\bX,\btheta)]}{\sqrt{n}},
\]
and recall that $\mathds{E}_{i-1}$ is the conditional expectation given $\mathcal{F}_{i-1}$. Then we have the multidimensional version of Assumptions \ref{assu:(uniform-integrable-entropy)} and \ref{assu:(Lindeberg-Feller)}.

\renewcommand{\theassumption}{A.\arabic{assumption}}
\setcounter{assumption}{3}
\begin{assumption}\label{assu:(uniform-integrable-entropy_multidim)}
For each $1\leq k\leq d$, the following holds: There exists $\Pi=\{\Pi(\epsilon)\}_{\epsilon\in(0,\Delta_{\Pi}]}$ such that each $\Pi(\epsilon)=\{\sfF(\epsilon;m):1\leq m \leq N_{\Pi}(\epsilon)\}$ is a covering of $\sfF_{\delta}^{(k)}$ (i.e., $\cup_{1\leq m \leq N_{\Pi}(\epsilon)}\sfF(\epsilon;m)=\sfF_{\delta}^{(k)}$) and $N_{\Pi}(\Delta_{\Pi})=1$. Here for each $1\leq m\leq N_{\Pi}(\epsilon)$, $\sfF(\epsilon;m)$ is an $\epsilon$-ball under $L_2$-distance $\rho(g,h) := (\mathds{E}_{\bX\sim P}[(g(\bX)-h(\bX))^2])^{1/2}.$ Moreover,
\[
\sup_{\epsilon\in(0,\Delta_{\Pi}]\cap \mathbb{Q}} \max_{1\leq m\leq N_{\Pi}(\epsilon)}\frac{\sqrt{\sum_{j=1}^{n}\mathds{E}_{j-1}\left[\left\vert V_{n,j} (\sfF(\epsilon;m))\right\vert^2\right]}}{\epsilon}=O_p(1),
\]
where for a set $\sfF^{\prime}$, $V_{n,j}(\sfF^{\prime})$ is defined as the smallest $\mathcal{F}_i$-measurable function that is greater than $\sup_{f,g\in \sfF^{\prime}}\left\vert V_{n,j}(f)-V_{n,j}(g)\right\vert$. 
Furthermore, 
\[
\int_0^{\Delta_{\Pi}}\sqrt{\log N_{\Pi}(\epsilon)} d\epsilon<\infty.
\]
\end{assumption}

\begin{assumption}(Lindeberg's condition)\label{assu:(Lindeberg-Feller_multidim)}
There exists a $\delta>0$ such that for each dimension $k$ and $\betheta\in\Theta_{\delta}$,
for every $\epsilon>0$, 
\[
\sum_{i=1}^{n}\mathds{E}_{i-1}\left[\left(V_{n,i}^{(k)}(E)\right)^{2}\mathbf{1}\left\{V_{n,i}^{(k)}(E)>\epsilon\right\}\right]\stackrel{P}{\longrightarrow}0,
\] 
where $V_{n,i}^{(k)}(E)$ is the adapted envelope for $V_{n,i}(f),f\in\sfF^{(k)}_\delta$, i.e., $V_{n,i}^{(k)}(E)$ is the smallest $\mathcal{F}_i$-measurable random variable such that  $\sup_{f\in\sfF^{(k)}_\delta}|V_{n,i}(f)|\leq V_{n,i}^{(k)}(E)~a.s$. 
\end{assumption}

Let $\rho (\betheta_1,\betheta_2):= \left(\mathds{E}_{\bX\sim P}\left\Vert\bF(\bX,\betheta_1)-\bF(\bX,\betheta_2)\right\Vert_2^2\right)^{1/2}$
and recall that $\bJ(\betheta)=D\boldf(\betheta)/D\betheta$. Then we have the following theorem:

\begin{theorem} \label{thm: SAA_CLT_multidim}
{\bf (Asymptotic normality of SAA with embedded adaptive IS (multivariate case))} Under Assumptions \ref{assu: truncation_alpha_Multidim} - \ref{assu:(Lindeberg-Feller_multidim)} and Assumption \ref{assu: GC_Jacobian}, suppose that the function $\mathds{E}_{\bX\sim P_{\bealpha}}\left[\left(\bF(\bX,\btheta)\ell(\bX,\balpha)\right)^{2}\right]$ is continuous in $\bealpha$, and  $\rho\left(\betheta,\betheta^*\right)\rightarrow0$ as $\btheta\rightarrow\btheta^{*}$.
Suppose further that the black-box function $I$ is continuous function.
Then, we have the asymptotic normality of $\hat{\btheta}_{n}$ generated from Algorithm \ref{alg:SAA_rootfinding_blackbox_multidim}, given by
\[
\sqrt{n}\left(\hat{\btheta}_{n}-\btheta^{*}\right)\Rightarrow \mathcal{N}\left(\bzero,\bV\right).
\]
where 
$\bV=\left[\bJ(\betheta^*)\right]^{-\top}\bSigma\left[\bJ(\betheta^*)\right]^{-1}$
and 
$
\bSigma = \Var_{\bX\sim P_{\bealpha^{*}}}\left(\bF(\bX,\btheta^{*})\ell(\bX,\balpha^{*})\right)$ {with $\balpha^{*} = I(\betheta^*,\bJ(\betheta^*))$}. 

Assume further that the scalar-valued performance function $g$ has all continuous partial derivatives at all dimensions at $\btheta^{*}$ and not all are zero. 
Then
\[
\sqrt{n}\left(g(\hat{\btheta}_{n})-g(\btheta^{*})\right)\Rightarrow \mathcal{N}\left(\bzero,\nabla g(\btheta^{*})^{\top}\bV\nabla g(\btheta^{*})\right).
\]
\end{theorem}

\subsubsection{SA.}
Let $\bV_n = \bF(\bX_{n+1},\hat{\btheta}_n)\ell(\bX_{n+1},\bealpha_{n+1})-\boldf(\hat{\btheta}_n)$ be the noise of {the estimated objective function}. 
\begin{assumption}\label{assu: sup_L2_multidim}
There exists a constant $C>0$ such that 
 $\mathds{E}_n\left[\left\Vert\bF(\bX_{n+1},\hat{\btheta}_{n})\ell(\bX_{n+1},\balpha_{n+1})\right\Vert_2^{2}\right]<C$.
\end{assumption}

Since typically $\bealpha_{n+1}=I(\hat{\betheta}_n,\hat{\bJ}_n)$ is designed to be a good sampler for the estimation of $\boldf({ \hat\betheta_n})$, it is natural that this assumption can hold as long as $I$ is reasonable. 
\begin{assumption}\label{assu: continuity_multidim}
The objective function $\boldf(\btheta)$ is continuous on $\btheta$, and $\boldf(\btheta)=\nabla h(\btheta)$
for some function $h$.
\end{assumption}

Suppose that we have some prior knowledge that $\btheta\in A\subset\mathbb{R}^{d}$ for some set $A$. Following the setup in Section 4.3 of \cite{kushner2003stochastic}, we suppose that $A$ can be defined in the following way. Let $q_{i}(\cdot),i=1,2,\cdots,p$ be continuously differentiable real-valued functions on $\mathbb{R}^{d}$ with gradients {$\nabla q_{i}(\theta)$}.
Let {$\nabla q_{i}(\theta)\neq 0$} when $q_{i}(\theta)=0$. $A$ is defined as
\begin{equation}\label{eq: def_H}
A:=\{\theta:q_{i}(\theta)\leq0,i=1,2,\cdots,p\}.
\end{equation}
A typical example is $A=[a_1,b_1]\times \dots\times[a_d,b_d]$, where $A$ is a rectangle. Then we introduce the following assumption regarding the shape of $A$ and the uniqueness of the solution:
\renewcommand{\theassumption}{A.E.2}
\begin{assumption}\label{assu: unique_solution_multidim}
$A$ is defined as above and is compact. $\btheta^{*}$ belongs to the interior of $A$ and is the only stationary point to the projected ODE:
\[
\frac{d}{dt}\btheta(t)=\boldf(\btheta)+\bz,\bz\in-C(\btheta),
\]
where $C(\theta)$ is defined as the cone of outer normals to the active constraint at $\btheta$.
\end{assumption}

The above assumptions are sufficient for strong consistency.
\begin{theorem}[Consistency of SA with embedded adaptive IS (multivariate case)]
\label{prop:RM_consistency_multidim}
Under Assumptions \ref{assu: sup_L2_multidim}-\ref{assu: continuity_multidim} and {\ref{assu: unique_solution_multidim}},
we have that the
$\hat{\btheta}_{n}$ {and $\bar{\btheta}_{n}$} defined in Algorithm \ref{alg: SA_rootfinding_blackbox_multidim} converges
to $\btheta^{*}$ {a.s.} 
\end{theorem}

The following assumptions are needed to guarantee asymptotic normality. 
\renewcommand{\theassumption}{A.E.3}
\begin{assumption}\label{assu:Polyak_obj1_multidim}
$\bR(\theta)$ is second-order differentiable in a neighborhood of $\btheta^*$ and $\bJ^*$ is positive definite.
\end{assumption}

We also need a uniform integrability condition for the noise:
\renewcommand{\theassumption}{A.\arabic{assumption}}
\setcounter{assumption}{7}
\begin{assumption}
\label{assu: Polyak_noise_multidim}

\[
\sup_{n}\mathds{E}_{n-1}[\left\Vert\bV_{n}\right\Vert|^{2}\mathbf{1}\{\left\Vert\bV_{n}\right\Vert>C\}]\stackrel{P}{\longrightarrow}0\text{ as \ensuremath{C\rightarrow\infty}.}
\]
\end{assumption}

Then for the RM-SA algorithm, we have the following asymptotic normality.
\begin{theorem} \label{prop: SA_CLT_multidim}
{\bf (Asymptotic normality of RM-SA with embedded adaptive IS (multivariate case))} Under Assumptions
\ref{assu: sup_L2_multidim}-\ref{assu: Polyak_noise_multidim} and \ref{assu: GC_Jacobian}-\ref{assu:Polyak_obj1_multidim}, 
suppose that the function $\mathds{E}_{\bX\sim P_{\bealpha}}\left[\left(\bF(\bX,\btheta)\ell(\bX,\balpha)\right)^{2}\right]$ is continuous in $\bealpha$ 
and the black-box function $I$ is continuous.
Let $\bP$ be an orthogonal matrix such that 
\[
\gamma\bP^{\top}\bJ(\btheta^{*})\bP=\bLambda
\]
is diagonal. 
Assume that the smallest eigenvalue $\min_{i\in\{1,2,\ldots,d\}}[\bLambda]_{ii}>{1}/{2}$. 
Then the RM-SA estimator $\hat{\btheta}_n$ in Algorithm \ref{alg: SA_rootfinding_blackbox_multidim} is asymptotically normal, viz.,
\[
\sqrt{n}\left(\hat{\btheta}_{n}-\btheta^{*}\right)\Rightarrow \mathcal{N}(\bzero,\bP\bM\bP^{\top}),
\]
where $[\bM]_{ij}=\gamma^2[\bP^{\top}\bSigma\bP]_{ij}([\bLambda]_{ii}+[\bLambda]_{jj}-1)^{-1}$ and $\bSigma = \Var_{\bX\sim P_{\bealpha^{*}}}\left(\bF(\bX,\btheta^{*})\ell(\bX,\balpha^{*})\right)$ {with $\balpha^{*} = I(\betheta^*,\bJ(\betheta^*))$}. 

Assume further that the scalar-valued performance function $g$ has all continuous partial derivatives at all dimensions at $\btheta^{*}$ and not all are zero.
Then
\[
\sqrt{n}\left(g(\hat{\btheta}_{n})-g(\btheta^{*})\right)\Rightarrow \mathcal{N}(\bzero,\nabla g(\btheta^{*})^{\top}\bP\bM\bP^{\top}\nabla g(\btheta^{*})).
\]
\end{theorem}

And for the PR-SA, we have the following asymptotic normality.
\begin{theorem}\label{prop: SA_average_CLT_multidim}
{\bf (Asymptotic normality of PR-SA with embedded adaptive IS (multivariate case))} Under the same assumptions for Theorem \ref{prop: SA_CLT_multidim}, 
the PR-SA estimator $\bar{\btheta}_n$ in Algorithm \ref{alg: SA_rootfinding_blackbox_multidim} is asymptotically normal, viz.,
\[
\sqrt{n}\left(\bar{\btheta}_{n}-\btheta^{*}\right)\Rightarrow\mathcal{N}(\bzero,\bV),
\]
where $\bar{\btheta}_{n}=\sum_{i=1}^n\hat{\btheta}_{i}/n$, $\bV=\left[\bJ(\betheta^*)\right]^{-\top}\bSigma\left[\bJ(\betheta^*)\right]^{-1}$ and $\bSigma = \Var_{\bX\sim P_{\bealpha^{*}}}\left(\bF(\bX,\btheta^{*})\ell(\bX,\balpha^{*})\right)$ {with $\balpha^{*} = I(\betheta^*,\bJ(\betheta^*))$}. 

Assume further that the scalar-valued performance function $g$ has all continuous partial derivatives at all dimensions at $\btheta^{*}$ and not all are zero. Then
\[
\sqrt{n}\left(g(\bar{\btheta}_{n})-g(\btheta^{*})\right)\Rightarrow \mathcal{N}\left(\bzero,\nabla g(\btheta^{*})^{\top}\bV\nabla g(\btheta^{*})\right).
\]

\end{theorem}

\section{Proofs on Theoretical Convergences}
This section provides all the proofs of the theoretical results in this paper. The proof of Theorem \ref{thm: SAA_CLT_multidim} is complicated, so we derive some intermediate results in Section \ref{appx:proof_of_A2} and provide the final proof in Section \ref{sc:finalproof_ThmA3}.
\subsection{Proof of Multidimensional Version of Lemma \ref{lem: average_consistent}}{\label{appx:prooflemma1}}
Here we provide a multidimensional version of Lemma \ref{lem: average_consistent} and its proof.
\begin{lemma}\label{lem: average_consistent_multidim}
Under Assumption \ref{assu: truncation_alpha_Multidim}, for each $\btheta\in\Theta$,
we have that 
\begin{equation}
\frac{\sum_{i=1}^{n}\bF(\bX_{i},\btheta)\ell(\bX_{i},\balpha_{i})}{n}\rightarrow \mathds{E}_{\bX\sim P}\left[\bF(\bX,\btheta)\right]~\text{a.s.}\
\end{equation}
\end{lemma}
\begin{proof}{Proof:}
Notice that $\ell(\bX,\balpha)$ is a scalar, so it suffices to prove that the pointwise convergence of each component of $\bF(\bX,\btheta)$. 
For ease of notation, we do not distinguish the components of $\bF(\bX,\btheta)$, and let $F(\bX,\btheta)$ be any component of $\bF(\bX,\btheta)$. Notice that, for each dimension, 
\begin{equation}
{F(\bX_{i},\btheta)\ell(\bX_{i},\bealpha_{i})-\mathds{E}_{\bX\sim P}\left[F(\bX,\btheta)\right]}
\end{equation}
is a martingale difference array. 

Assumption \ref{assu: truncation_alpha_Multidim} implies that 
\begin{equation}\label{eq: In_proof_average_consistent}
\sum_{i=1}^{\infty}\frac{\mathds{E}\left[\left(F(\bX_{i},\btheta)\ell(\bX_{i},\balpha_{i})-\mathds{E}_{\bX\sim P}[F(\bX,\btheta)]\right)^{2}\right]}{i^{2}}<\infty.
\end{equation}
Recall that $f(\btheta):= \mathds{E}_{\bX\sim P}\left[F(\bX,\btheta)\right]$, and from the martingale difference property, in the following computation of the second moment, the cross term will be zero, and we have
\begin{align*}
 & \mathds{E}\left[\left(\sum_{i=1}^{n}\frac{F(\bX_{i},\btheta)\ell(\bX_{i},\balpha_{i})-f(\btheta)}{i}\right)^2\right]\\
= & \mathds{E}\left[\mathds{E}\left[\left(\sum_{i=1}^{n}\frac{F(\bX_{i},\btheta)\ell(\bX_{i},\balpha_{i})-f(\btheta)}{i}\right)^2\right]\bigg|\mathcal{F}_{n-1}\right]\\
= & \mathds{E}\left[\left(\sum_{i=1}^{n-1}\frac{F(\bX_{i},\btheta)\ell(\bX_{i},\balpha_{i})-f(\btheta)}{i}\right)^2\right]+\mathds{E}\left[\left(\frac{F(\bX_{n},\btheta)\ell(\bX_{n},\balpha_{n})-f(\btheta)}{n}\right)^{2}\right]\\
 & +2\mathds{E}\left[\sum_{i=1}^{n-1}\left(\frac{F(\bX_{i},\btheta)\ell(\bX_{i},\balpha_{i})-f(\btheta)}{i}\right)\mathds{E}\left[\frac{F(\bX_{n},\btheta)\ell(\bX_{n},\balpha_{n})-f(\btheta)}{n}\bigg|\mathcal{F}_{n-1}\right]\right]\\
= & \mathds{E}\left[\left(\sum_{i=1}^{n-1}\frac{F(\bX_{i},\btheta)\ell(\bX_{i},\balpha_{i})-f(\btheta)}{i}\right)^2\right]+\mathds{E}\left[\left(\frac{F(\bX_{n},\btheta)\ell(\bX_{n},\balpha_{n})-f(\btheta)}{n}\right)^{2}\right]\\
= & \cdots\\
= & \sum_{i=1}^{n}\mathds{E}\left[\left(\frac{F(\bX_{i},\btheta)\ell(\bX_{i},\balpha_{i})-f(\btheta)}{i}\right)^{2}\right].
\end{align*}
So using (\ref{eq: In_proof_average_consistent}) we have that 
\begin{equation}
\sup_{n}\mathds{E}\left[\left(\sum_{i=1}^{n}\frac{F(\bX_{i},\btheta)\ell(\bX_{i},\balpha_{i})-f(\btheta)}{i}\right)^2\right]<\infty.
\end{equation}
Then for the absolute value,
\begin{equation}
\sup_{n}\mathds{E}\left[\left|\sum_{i=1}^{n}\frac{F(\bX_{i},\btheta)\ell(\bX_{i},\balpha_{i})-f(\btheta)}{i}\right|\right]<\infty.
\end{equation}
Then notice that $M_{n}=\sum_{i=1}^{n}\left[{(F(\bX_{i},\btheta)\ell(\bX_{i},\balpha_{i})-\mathds{E}_{\bX\sim P}F(\bX,\btheta))}/{i}\right]$
is a martingale. From the martingale convergence theorem (Theorem 5.2.8, \citealt{Durret2019}), we have that
there exists a.s. finite $Z$ such that 
\begin{equation}
\sum_{i=1}^{n}\frac{F(\bX_{i},\btheta)\ell(\bX_{i},\balpha_{i})-f(\btheta)}{i}\rightarrow Z.
\end{equation}
Then by Kronecker's lemma (Theorem 2.5.5, \citealt{Durret2019}), we have that 
\begin{equation}
\sum_{i=1}^{n}\frac{F(\bX_{i},\btheta)\ell(\bX_{i},\balpha_{i})-f(\btheta)}{n}\rightarrow 0~a.s.
\end{equation}
\hfill{$\Box$}\end{proof}

\subsection{Proof of Multidimensional Version of Lemma \ref{Lem: unif_converg_aver_obj}}\label{appx:prooflemma2}
The multidimensional version of Lemma \ref{Lem: unif_converg_aver_obj} is stated and followed by its proof.
\begin{lemma}\label{Lem: unif_converg_aver_obj_multidim}
Under Assumptions \ref{assu: truncation_alpha_Multidim} and \ref{assu: bracket_multidim},
we have that, as $n\rightarrow\infty$, 
\begin{equation}
\sup_{\betheta\in \Theta}\left\Vert\frac{\sum_{i=1}^{n}\bF(\bX_{i},\btheta)\ell(\bX_{i},\balpha_{i})}{n}-\mathds{E}_{\bX\sim P}\left[\bF(\bX,\btheta)\right]\right\Vert\rightarrow0~\text{a.s.}
\end{equation}
\end{lemma}
\begin{proof}{Proof:}
Similar to the proof of Lemma \ref{lem: average_consistent_multidim}, let $F(\bX,\btheta)$ be any component of $\bF(\bX,\btheta)$, and it suffices to prove the one-dimensional case.
According to Lemma \ref{lem: average_consistent_multidim}, we have that
\begin{align*}
\lim_{n\rightarrow \infty} \frac{1}{n}\sum_{i=1}^n f_{R}(\bX_{i},\balpha_{i}) - \mathds{E}_{\bX\sim P}\left[F(\bX,\btheta)\right] &= \lim_{n\rightarrow \infty} \frac{1}{n}\sum_{i=1}^n \big[ f_{R}(\bX_{i},\balpha_{i})  - F(\bX_{i},\btheta)\ell(\bX_{i},\balpha_{i}) \big] \\
& \leq \lim_{n\rightarrow \infty} \frac{1}{n}\sum_{i=1}^n \big[ f_{R}(\bX_{i},\balpha_{i})  - f_{L}(\bX_{i},\balpha_{i}) \big] \leq \epsilon.
\end{align*}

Then for each $\btheta$, we have that
\begin{align*}
&\frac{\sum_{i=1}^{n}F(\bX_{i},\btheta)\ell(\bX_{i},\balpha_{i})}{n}-\mathds{E}_{\bX\sim P}\left[F(\bX,\btheta)\right] \\
& \leq\frac{1}{n}\sum_{i=1}^{n}\big\{f_{R}(\bX_{i},\balpha_{i})-\mathds{E}_{\bX\sim P}\left[F(\bX,\btheta)\right]\big\}\\
 & =\frac{1}{n}\sum_{i=1}^{n}f_{R}(\bX_{i},\balpha_{i})-\lim_{n\rightarrow\infty}\frac{1}{n}\sum_{i=1}^{n}f_{R}(\bX_{i},\balpha_{i}) +\lim_{n\rightarrow\infty}\frac{1}{n}\sum_{i=1}^{n}f_{R}(\bX_{i},\balpha_{i})-\mathds{E}_{\bX\sim P}\left[F(\bX,\btheta)\right]\\
 & \leq\frac{1}{n}\sum_{i=1}^{n}f_{R}(\bX_{i},\balpha_{i})-\lim_{n\rightarrow\infty}\frac{1}{n}\sum_{i=1}^{n}f_{R}(\bX_{i},\balpha_{i})+\epsilon.
\end{align*}
Similarly, we have the lower bound
\begin{equation*}
\frac{\sum_{i=1}^{n}F(\bX_{i},\btheta)\ell(\bX_{i},\balpha_{i})}{n}-\mathds{E}_{\bX\sim P}[F(\bX,\btheta)]\geq\frac{1}{n}\sum_{i=1}^{n}f_{L}(\bX_{i},\balpha_{i})-\lim_{n\rightarrow\infty}\frac{1}{n}\sum_{i=1}^{n}f_{L}(\bX_{i},\balpha_{i})-\epsilon.
\end{equation*}
Based on these two inequalities, we have that for each $\btheta$,
\begin{align*}
 & \left|\frac{1}{n}\sum_{i=1}^{n}F(\bX_{i},\btheta)\ell(\bX_{i},\balpha_{i})- \mathds{E}_{\bX\sim P}\left[ F(\bX,\btheta)\right]\right|\\
\leq & \epsilon+\max_{f=f_{L}\text{or \ensuremath{f_{R}};(\ensuremath{f_{L},f_{R}})\ensuremath{\in K_{\epsilon}}}}\left\{ \left|\frac{1}{n}\sum_{i=1}^{n}f(\bX_{i},\balpha_{i})-\lim_{n\rightarrow\infty}\frac{1}{n}\sum_{i=1}^{n}f(\bX_{i},\balpha_{i})\right|\right\}. 
\end{align*}

Taking supremum over $\btheta$, and let $n\rightarrow\infty$, since
$K_{\epsilon}$ is a finite set, we have that 
\[
\lim_{n\rightarrow\infty}\sup_{\betheta\in \Theta}\left|\frac{1}{n}\sum_{i=1}^{n}F(\bX_{i},\btheta)\ell(\bX_{i},\balpha_{i})-\mathds{E}_{\bX\sim P}\left[F(\bX,\btheta)\right]\right|\leq\epsilon.
\]
Then since $\epsilon$ is arbitrary, we have that
\[
\lim_{n\rightarrow\infty}\sup_{\betheta\in \Theta}\left|\frac{1}{n}\sum_{i=1}^{n}F(\bX_{i},\btheta)\ell(\bX_{i},\bealpha_{i})-\mathds{E}_{\bX\sim P}\left[F(\bX,\btheta)\right]\right|=0~a.s.
\]
\hfill{$\Box$}\end{proof}

\subsection{Proof of Theorem \ref{prop:SAA_consistency_multidim} (Multidimensional Version of Theorem \ref{prop:SAA_consistency})} \label{appx:proof_Thm2}
\proof{Proof:}
Denote $\boldf_{n}(\btheta)=\sum_{i=1}^{n}\bF(\bX_{i},\btheta)\ell(\bX_{i},\balpha_{i})/n$
and $\boldf(\btheta)=\mathds{E}_{\bX\sim P}\left[\bF(\bX,\btheta)\right]$. 
According to Lemma \ref{Lem: unif_converg_aver_obj_multidim}, $\boldf_{n}(\btheta)$ is uniformly convergent to $\boldf(\btheta)$. 
For each $\epsilon>0$, a.s., there exists $N_1(\epsilon)$ such that when $n\geq N_1(\epsilon)$,
\[
\sup_{\betheta\in \mathbb{R}^{d}}\left\Vert\boldf_{n}(\btheta)-\boldf(\btheta)\right\Vert \leq \frac{1}{2}\inf_{\Vert\betheta-\betheta^{*}\Vert\geq\epsilon}\left\Vert \boldf(\btheta)\right\Vert. 
\]
Since $\boldf_{n}(\hat{\btheta}_{n})\rightarrow \bzero$, from the last inequality
we have that there exists $N_{2}(\epsilon)$ such that 
\[
\Vert \boldf(\hat{\btheta}_{n}) \Vert \leq\frac{3}{4}\inf_{\Vert\betheta - \betheta^{*}\Vert\geq\epsilon}\left\Vert \boldf(\btheta)\right\Vert, 
\]
when $n\geq N_{2}(\epsilon)$. This implies $\Vert \hat{\btheta}_{n}-\btheta^{*}\Vert <\epsilon$
when $n\geq\max\{N_1(\epsilon),N_{2}(\epsilon)\}$. Since $\epsilon>0$
is arbitrary, we conclude that $\hat{\btheta}_{n}\rightarrow\btheta^{*}~a.s.$ \hfill{$\Box$}
\endproof

\subsection{Proof of Theorem \ref{thm: SAA_CLT_multidim} (Multidimensional Version of Theorem \ref{thm: SAA_CLT})}\label{appx:proof_of_A2}
To prove Theorem \ref{thm: SAA_CLT_multidim}, we need several lemmas. We first state these lemmas and their proofs (or references), and then prove the main theorem. 

\subsubsection{Consistency of $\bealpha_i$.}
We begin by showing the consistency of $\bealpha_i$. 
\begin{lemma}\label{lem: consistent_alpha}
Suppose that $I(\betheta,\bJ)$ is a continuous function and Assumptions \ref{assu: truncation_alpha_Multidim}-\ref{assu: Obj_regular_multidim} and Assumption \ref{assu: GC_Jacobian} hold. Then $\bealpha_i\rightarrow\bealpha^*$ a.s.
\end{lemma}
\proof{Proof of Lemma \ref{lem: consistent_alpha}:}
From Theorem \ref{prop:SAA_consistency_multidim}, we know that $\hat{\betheta}_n\rightarrow\betheta^*~a.s.$ From Assumption \ref{assu: GC_Jacobian} and with the same proof in \ref{appx:prooflemma2}, we have that 
\[
\hat \bJ_{n}(\btheta):=\frac{1}{n}\sum_{i=1}^{n}\frac{D}{D\btheta}\bF(\bX_{i},\btheta)\ell(\bX_{i},\balpha_{i})
\]
converges to $\bJ(\theta)$ uniformly for $\btheta\in\Theta_{\delta}$. Thus from the consistency of $\hat{\btheta}_{n}$, we have that $\bJ_{n}(\hat{\btheta}_{n})\rightarrow \bJ(\btheta^{*})$. So from the continuity of I, we have shown the result.
\hfill{$\Box$}
\endproof

\subsubsection{Asymptotic Equicontinuity.}
We next show a lemma about asymptotic equicontinuity. To describe the result, we introduce a notion called outer probability. Let $P^*$ denote the {\em outer probability} of a subset $B$ of $\Omega$, i.e., 
\[
P^*\{B\} = \inf\{P(A): B\subset A,A\in\mathcal{F}\}.
\]
Notice that, when $B$ is measurable, $P^*(B)$ is just the probability $P(B)$, and recall that for each measurable function $g$ 
\[
S_{n}(g)=\sum_{i=1}^n V_{n,i}(g)\text{~~with~~}V_{n,i}(g)=\frac{g(\bX_{i})\ell(\bX_{i},\balpha_{i})-\mathds{E}_{\bX\sim P}[g(\bX)]}{\sqrt{n}}.
\]
Then we can state the asymptotic equicontinuity in the following lemma.

\begin{lemma} \label{lem: Asymptotic_equicontinuity}
Suppose Assumptions \ref{assu:(uniform-integrable-entropy_multidim)} and \ref{assu:(Lindeberg-Feller_multidim)} hold. Then
for each dimension $k$, given $\epsilon>0$ and $\gamma>0$, there
exists $\eta>0$ for which 
\[
\limsup_{n\rightarrow\infty}P^{*}\left\{\sup_{\rho(f,f_{\betheta^*}^k)\leq\eta,f\in\sfF^{(k)} _\delta}\left|S_{n}(f)-S_{n}(f_{\betheta^*}^{(k)})\right|>\gamma\right\}\leq\epsilon.
\]
\end{lemma}

To prove this lemma, we first introduce a weak convergence result shown in \cite{nishiyama2000}.

To this end, we introduce some notations from Definition 2.2 of \cite{nishiyama2000}. 
Let $\sfF$ be any arbitrary set. $\Pi=\{\Pi(\epsilon)\}_{\epsilon\in(0,\Delta_{\Pi}]} $, where $\Delta_{\Pi}\in(0,\infty)\cap\mathbb{Q}$, is called a decreasing series of finite partitions (abb. DFP) [resp., nested series of finite partitions (abb. NFP)] of $\sfF$ if it satisfies the following (i), (ii) and (iii) [resp., (i), (ii), and (iii)']:
(i) Each $\Pi(\epsilon)=\{\sfF(\epsilon;k):1\leq k\leq N_{\Pi}(\epsilon)\}$ is a finite partition of $\sfF$.
(ii) $N_{\Pi}(\Delta_{\Pi})=1$ and $\lim_{\epsilon\rightarrow 0^+}N_{\Pi}(\epsilon)=\infty$.
(iii) $N_{\Pi}(\epsilon)\geq N_{\Pi}(\epsilon^{\prime})$ whenever $\epsilon\leq\epsilon^{\prime}$.
(iii') $\Pi(\epsilon)\supset\Pi(\epsilon^{\prime})$ whenever $\epsilon\leq\epsilon^{\prime}$.
From \cite{nishiyama2000} we can conclude the following result.
\begin{lemma}\label{lem: in_Nishiyama} (Discrete-time version of Lemma 3.3 of \citealt{nishiyama2000}) 
For any $\delta>0$, if Assumption \ref{assu:(Lindeberg-Feller_multidim)} holds, and there exists a DFP $\Pi$ of $\sfF^{(k)}_{\delta}$ such that:
\begin{flalign}
(a)~~~~~~~~~~~~~~~~~~~~~~\sup_{\epsilon\in(0,\Delta_{\Pi}]\cap \mathbb{Q}} \max_{1\leq m\leq N_{\Pi}(\epsilon)}\frac{\sqrt{\sum_{j=1}^{n}\mathds{E}_{j-1}\left\vert V_{n,j} (\sfF(\epsilon;m))\right\vert^2}}{\epsilon}=&O_p(1),&\nonumber
\end{flalign}
where for a set $\sfF^{\prime}$, $V_{n,j}(\sfF^{\prime})$ is defined as the smallest $\mathcal{F}_i$-measurable function that is greater than $\sup_{f,g\in \sfF^{\prime}}\left\vert V_{n,j}(f)-V_{n,j}(g)\right\vert$, 
and
\begin{flalign}
(b)~~~~~~~~~~~~~~~~~~~~~~~~~~~~~~~~~~~~~~~~~~~~\int_0^{\Delta_{\Pi}}\sqrt{\log N_{\Pi}(\epsilon)} d\epsilon<&\infty.&\nonumber
\end{flalign}
Then, for any $\epsilon,\gamma$, there exists a finite partition $\{\sfF_{j}:1\leq j\leq N\}$ of $\sfF^{(k)}_{\delta}$ such that 
\[
\limsup_{n\rightarrow \infty} P^* \left(\sup_{1\leq j\leq N, f,g \in \sfF_{j}}|S_n(f)-S_n(g)|>\gamma\right)\leq \epsilon.
\]
\end{lemma}

\proof{Proof:}
We observe that Assumption \ref{assu:(Lindeberg-Feller)} is equivalent to [L2'] in \cite{nishiyama2000} and the other conditions in this lemma is equivalent to [PE'] in \cite{nishiyama2000}. In \cite{nishiyama2000}, the discussion for the continuous-time model in Section 3 and the discussion for the discrete-time model in Section 4 are parallel. So this discrete-time version of Lemma 3.3 of \cite{nishiyama2000} holds. \hfill{$\Box$}
\endproof

\textbf{Proof of Lemma \ref{lem: Asymptotic_equicontinuity}}
Now we can prove Lemma \ref{lem: Asymptotic_equicontinuity} based on Lemma \ref{lem: in_Nishiyama}.
\proof{Proof of Lemma \ref{lem: Asymptotic_equicontinuity}.}
 For each $1\leq k\leq d$, let $\Pi_0=\{\Pi_0(\epsilon)\}_{\epsilon\in (0,\Delta_{\Pi})}$ be the set of covers for $\sfF_{\delta}^{(k)}$ in Assumption \ref{assu:(uniform-integrable-entropy_multidim)}. And write the elements in $\Pi_0(\epsilon)$ as $\Pi_0(\epsilon)=\{\sfF_0(\epsilon;m),1\leq m\leq N_{\Pi}(\epsilon)\}$. As we can see, the conditions required in Lemma \ref{lem: in_Nishiyama} have similar forms as Assumption \ref{assu:(uniform-integrable-entropy_multidim)}. Indeed, it is not hard to construct a DFP from $\Pi_0$ which satisfy the conditions in Lemma \ref{lem: in_Nishiyama}. To satisfy (i) in the definition of DFP, we just need to delete the overlapping areas in some sets in $\Pi_0(\epsilon)$ to make it a partition. Since this operation will not make any set larger and it does not change the number of sets, the conditions required for Lemma \ref{lem: in_Nishiyama} are inherited from Assumption \ref{assu:(uniform-integrable-entropy_multidim)}.
So with Lemma \ref{lem: in_Nishiyama}, we have that there exists a finite partition $\{\sfF_j^{(k)}:1\leq j\leq J\}$ of $\sfF_{\delta}^{(k)}$ such that
\begin{equation}\label{eq: partition_finite}
\limsup_{n\rightarrow \infty} P^* \left(\sup_{1\leq j\leq J, f,g \in \sfF^{(k)}_j}|S_n(f)-S_n(g)|>\gamma\right)\leq \epsilon.
\end{equation}

Based on this, to derive the result of Lemma \ref{lem: Asymptotic_equicontinuity} from the above inequality, it suffices to explain that one of the sets in $\sfF^{(k)}_j, 1\leq j\leq J$  contains $f_{\betheta^*}^k$ as an interior point under $L_2$-distance $\rho$. This claim can be seen from the proof in \cite{nishiyama2000}. Indeed, by Assumption \ref{assu:(uniform-integrable-entropy_multidim)}, the sets in $\Pi_0$ are $\epsilon$-balls, which are open sets under $L_2$-distance $\rho$.  
In addition, from the proof of Lemma 3.3 of \cite{nishiyama2000}, we can see that the sets $\sfF^{(k)}_j$ in \eqref{eq: partition_finite} comes from a NFP (denote by $\{\Pi^{\prime}(\epsilon)\}$) which is constructed (using the procedure depicted in Lemma 2.4 of \citealt{nishiyama2000}) from the DFP $\{\Pi(\epsilon)\}$. And by checking the construction in the proof of Lemma 2.4 of \cite{nishiyama2000} we can check that, if each $\Pi(\epsilon)$ contains a set that contains $f_{\betheta^*}^k$ as an interior point, then for each $\epsilon$, $\Pi^{\prime}(\epsilon)$  also has a set that contains $f_{\betheta^*}^k$ as interior point. And this implies $f_{\betheta^*}$ is an interior point of one of $\sfF^{(k)}_j$ in \eqref{eq: partition_finite}. \hfill{$\Box$}
\endproof

\subsubsection{Asymptotic Normality Result for the Sample Average Estimator.}

Lemma \ref{lem: Asymptotic_equicontinuity} provides the asymptotic equicontinuity for the sum of martingale difference array at  $f_{\betheta^*}$. 
Then, given the choices of $\balpha_i,i=1,2,\ldots$, we can use the result in \cite{vanWell1996}, given by Lemma \ref{lem: Thm3.3.1}, to obtain the asymptotic normality of $\hat{\btheta}_{n}$. 
\begin{lemma}(Theorem 3.3.1, \citealt{vanWell1996}) \label{lem: Thm3.3.1}
Let $\Psi_{n}$ and $\Psi$ be random maps and a fixed map, respectively, from $\bTheta$ into a Banach space such that
\begin{equation}\label{eq:3.3.2}
 \sqrt{n}\left(\Psi_{n}-\Psi\right)\left(\hat{\betheta}_{n}\right)-\sqrt{n}\left(\Psi_{n}-\Psi\right)\left(\betheta^*\right)=o_{P}\left(1+\sqrt{n}\left\|\hat{\betheta}_{n}-\betheta^*\right\|\right),
 \end{equation}
and such that the sequence $\sqrt{n}\left(\Psi_{n}-\Psi\right)\left(\betheta^*\right)$ converges in distribution to a tight random element $Z .$ Let $\betheta \mapsto \Psi(\betheta)$ be Fr\'{e}chet-differentiable at $\betheta^*$ with a continuously invertible derivative $\dot\Psi_{\betheta^*}$.
If $\Psi\left(\betheta^*\right)=0$ and $\hat{\betheta}_{n}$ satisfies $\Psi_{n}\left(\hat{\betheta}_{n}\right)=o_{P}\left(n^{-1 / 2}\right)$ and converges in outer probability to $\betheta^*,$ then
\[
\sqrt{n} \dot{\Psi}_{\betheta^*}\left(\hat{\betheta}_{n}-\betheta^*\right)=-\sqrt{n}\left(\Psi_{n}-\Psi\right)\left(\betheta^*\right)+o_{P}(1).
\]
Consequently, $\sqrt{n}\left(\hat{\betheta}_{n}-\betheta^*\right) \leadsto-\dot{\Psi}_{\betheta^*}^{-1} Z .$ If it is known that the sequence $\sqrt{n}\left\|\hat{\betheta}_{n}-\betheta^*\right\|$ is asymptotically tight, then the first conclusion is valid without the assumption of continuous invertibility of $\dot{\Psi}_{\betheta^*} .$ If it is known that $\sqrt{n}\left(\hat{\betheta}_{n}-\betheta^*\right)$ is asymptotically tight, then it suffices that $\Psi$ is Hadamard-differentiable.
\end{lemma}

\subsubsection{Final Proof of Theorem \ref{thm: SAA_CLT_multidim}.}\label{sc:finalproof_ThmA3}
Finally, we will provide the proof of Theorem \ref{thm: SAA_CLT_multidim}.
\proof{Proof of Theorem \ref{thm: SAA_CLT_multidim}:}
To prove this theorem, we use Lemma \ref{lem: Thm3.3.1}, and we will establish the conditions for Lemma \ref{lem: Thm3.3.1}. 
In our case, $\Psi_n$ in Lemma \ref{lem: Thm3.3.1} is the map
\[
\Psi_n(\betheta)=\sum_{i=1}^{n}\frac{\bF(\bX_{i},\btheta)\ell(\bX_{i},\balpha_{i})}{n}-\bc
\]
and $\Psi$ in Lemma \ref{lem: Thm3.3.1} is the map
\[
\Psi(\betheta)= \mathds{E}_{\bX\sim P}[\bF(\bX,\btheta)]-\bc.
\]
And it can be seen that 
\[
\Psi_n(\betheta)-\Psi(\betheta)=\sum_{i=1}^{n}\frac{\bF(\bX_{i},\btheta)\ell(\bX_{i},\balpha_{i})-\mathds{E}_{\bX\sim P}[\bF(\bX,\btheta)]}{\sqrt{n}}.
\]
Since $\Psi_n$ and $\Psi$ are functions that map from $\mathds{R}^d$ to $\mathds{R}^d$, we do not need to worry about the tightness issue as stated in Lemma \ref{lem: Thm3.3.1}. Also, the differentiability in the functional sense reduces to the usual differentiability which is implied by our smoothness assumption on the objective (Assumption \ref{assu: Obj_regular_multidim}). So to establish the conditions for Lemma \ref{lem: Thm3.3.1}, it suffices to verify condition \eqref{eq:3.3.2} and show the convergence 
\[
\sum_{i=1}^{n}\frac{\bF(\bX_{i},\btheta^{*})\ell(\bX_{i},\balpha_{i})-\mathds{E}_{\bX\sim P}[\bF(\bX,\btheta^{*})]}{\sqrt{n}}\Rightarrow \mathcal{N}(\bzero,\bSigma).
\]

For condition \eqref{eq:3.3.2}, we will show that
\begin{align*}
  \sum_{i=1}^{n}\frac{\bF(\bX_{i},\hat{\btheta}_{n})\ell(\bX_{i},\balpha_{i})-\mathds{E}_{\bX\sim P}[\bF(\bX,\hat{\btheta}_{n})]}{\sqrt{n}}-\sum_{i=1}^{n}\frac{\bF(\bX_{i},\btheta^*)\ell(\bX_{i},\balpha_{i})-\mathds{E}_{\bX\sim P}[\bF(\bX,\btheta^*)]}{\sqrt{n}}
=  o_{P}\left(1 \right).
\end{align*}
It suffices to show that, for each dimension of $\bF$, the above equality is correct.

Given strong consistency $\hat{\btheta}_{n}\rightarrow\btheta^{*}$, this could
be implied by Lemma \ref{lem: Asymptotic_equicontinuity}. To be more precise, from our definition of $S_n$, we can check that for
each dimension $k\in\{1,2,\ldots,d\}$, 
\begin{align*}
 & \sum_{i=1}^{n}\frac{\bF^{(k)}(\bX_{i},\btheta)\ell(\bX_{i},\balpha_{i})-\mathds{E}_{\bX\sim P}\left[\bF^{(k)}(\bX,\btheta)\right]}{\sqrt{n}}-\sum_{i=1}^{n}\frac{\bF^{(k)}(\bX_{i},\btheta)\ell(\bX_{i},\balpha_{i})-\mathds{E}_{\bX\sim P}\left[\bF^{(k)}(\bX,\btheta)\right]}{\sqrt{n}}\\
= & S_{n}\left(f_{\hat{\betheta}_{n}}^{(k)}\right)-S_{n}\left(f_{\betheta^{*}}^{(k)}\right).
\end{align*}
For any $\epsilon>0$ and $\gamma>0$, from the result of Lemma \ref{lem: Asymptotic_equicontinuity},
we can find $\eta$ and $N_1$ such that, when $n\geq N_1$,
\[
P^{*}\left\{\sup_{\rho\left(f_{\betheta}^{(k)},f_{\betheta^{*}}^{(k)}\right)\leq\eta}\left|S_{n}(f^{(k)}_{\betheta})-S_{n}(f^{(k)}_{\betheta^{*}})\right|>\gamma\right\}\leq \epsilon.
\] 
Notice that 
\begin{eqnarray*}
\rho(\betheta,\betheta^*)&=&\left(\mathds{E}_{\bX\sim P}\left[\left\Vert\bF(\bX,\betheta_1)-\bF(\bX,\betheta_2)\right\Vert_2^2\right]\right)^{1/2}\\
&\geq& \left(\mathds{E}_{\bX\sim P}\left[\left(\bF^{(k)}(\bX,\betheta_1)-\bF^{(k)}(\bX,\betheta_2)\right)^2\right]\right)^{1/2}\\
&=&\rho\left(f_{\betheta}^{(k)},f_{\betheta^{*}}^{(k)}\right),
\end{eqnarray*}
and we conclude that when $n\geq N_1$, 
\[
P^{*}\left\{\sup_{\rho(\betheta,\betheta^*)\leq\eta}\left|S_{n}(f^{(k)}_{\betheta})-S_{n}(f^{(k)}_{\betheta^{*}})\right|>\gamma\right\}\leq \epsilon.
\] 

From the strong consistency $\hat{\btheta}_{n}\rightarrow \btheta^*$, we have that $\rho\left(\hat\betheta_n,\betheta^*\right)\rightarrow0$ as $\hat \btheta_n\rightarrow\btheta^{*}$,
so there exist constant $N_{2}$
such that when $n\geq N_{2}$, 
\[
P\left\{\rho\left(\hat\betheta_n,\betheta^*\right)>\eta\right\}\leq\epsilon.
\]
Combing the above two results, we have that when $n\geq\max\{N_1,N_{2}\}$,
\[
P^{*}\left\{\left|S_{n}^{(k)}\left(f_{\hat{\betheta}_{n}}\right)-S_{n}^{(k)}\left(f_{\betheta^{*}}\right)\right|>\gamma\right\}\leq2\epsilon.
\]
Hence since $\epsilon>0$ and $\gamma>0$ is arbitrary, we have that
\[
S_{n}^{(k)}(f_{\hat{\betheta}_{n}})-S_{n}^{(k)}(f_{\betheta^{*}})\stackrel{P}{\longrightarrow} 0.
\]
This verifies the condition \eqref{eq:3.3.2}.

Next, we show the convergence 
\[
\sum_{i=1}^{n}\frac{\bF(\bX_{i},\btheta^{*})\ell(\bX_{i},\balpha_{i})-\mathds{E}_{\bX\sim P}[\bF(\bX,\btheta^{*})]}{\sqrt{n}}\Rightarrow \mathcal{N}(\bzero,\bSigma).
\]
To show this, from the Cram\'er-Wold device (Theorem 3.10.6, \citealt{Durret2019}), it suffices to show that for each vector $\bv\in\mathbb{R}^{d}$,
\[
\bv^{\top}\frac{\sum_{i=1}^{n}\bF(\bX_{i},\btheta^{*})\ell(\bX_{i},\balpha_{i})-\mathds{E}_{\bX\sim P}[\bF(\bX,\btheta^{*})]}{\sqrt{n}}\Rightarrow \mathcal{N}(\bzero,\bv^{\top}\bSigma \bv).
\]

We will use the martingale central limit theorem (Theorem 8.2.8, \citealt{Durret2019}) to show this, and need to verify the following two conditions: 
\begin{eqnarray*}
&\text{(C1)}&~~~\frac{1}{n}\sum_{i=1}^{n}\Var_{i-1}\left(\bv^{\top}\bF(\bX_{i},\btheta^{*})\ell(\bX_{i},\balpha_{i})\right)\stackrel{P}{\longrightarrow}\bv^{\top}\bSigma \bv, \label{eq: LF_1}\\
&\text{(C2)}&~~~\frac{1}{n}\sum_{i=1}^{n}\mathds{E}_{i-1}\left[\left(\bv^{\top}\bF(\bX_{i},\btheta^{*})\ell(\bX_{i},\balpha_{i})-\bv^{\top}\bc\right)^{2}\mathbf{1}{\left\{\left|\bv^{\top}\bF(\bX_{i},\btheta^{*})\ell(\bX_{i},\balpha_{i})-\bv^{\top}\bc\right|\geq\epsilon\sqrt{n}\right\}}\right]\stackrel{P}{\longrightarrow}0.\label{eq: LF_2}
\end{eqnarray*}

According to Lemma \ref{lem: consistent_alpha} and the assumption that $\mathds{E}_{\bX\sim P_{\bealpha}}\left[\left(\bF(\bX,\btheta)\ell(\bX,\balpha)\right)^{2}\right]$ is continuous in $\bealpha$, we have that $\Var_{i-1}\left(\bF(\bX_{i},\btheta^{*})\ell(\bX_{i},\balpha_{i})\right)\rightarrow \Var_{\bX\sim P_{\bealpha^*}}\left(\bF(\bX,\btheta)\ell(\bX,\balpha^*)\right)$, so (C1) is verified and we have that $\bSigma=\Var_{\bX\sim P_{\bealpha^*}}\left[\bF(\bX,\btheta)\ell(\bX,\balpha^*)\right]$. 
For (C2), roughly speaking, it will follow from Assumption \ref{assu:(Lindeberg-Feller_multidim)}. This is not hard to see when $d=1$. 
But when $d\geq2$, we will need some algebra. Let $v^{(k)}$ denote the $k$th component of $\bv$. 
From the definition of the envelope function, and the fact that $\mathds{E}_{\bX\sim P}[\bF(\bX,\btheta^{*})]=\bc$, we have that
\[
\left|\frac{\bF^{(k)}(\bX_{i},\btheta^{*})\ell(\bX_{i},\balpha_{i})-\bc^{(k)}}{\sqrt{n}}\right| \le V_{n,i}^{(k)}(E). 
\]
Then
\begin{align}
 & \frac{1}{n}\sum_{i=1}^{n}\mathds{E}_{i-1}\left[\left(\bv^{\top}\bF(\bX_{i},\btheta^{*})\ell(\bX_{i},\balpha_{i})-\bv^{\top}\bc\right)^{2}\mathbf{1}\{\left|\bv^{\top}\bF(\bX_{i},\btheta^{*})\ell(\bX_{i},\balpha_{i})-\bv^{\top}\bc\right|\geq\epsilon\sqrt{n}\}\right]\nonumber \\
= & \sum_{i=1}^{n}\mathds{E}_{i-1}\left[\left(\bv^{\top}\frac{\bF(\bX_{i},\btheta^{*})\ell(\bX_{i},\balpha_{i})-\bc}{\sqrt{n}}\right)^{2}\mathbf{1}\left\{\left|\bv^{\top}\frac{\bF(\bX_{i},\btheta^{*})\ell(\bX_{i},\balpha_{i})-\bc}{\sqrt{n}}\right|\geq\epsilon\right\}\right]\nonumber \\
\leq & \sum_{i=1}^{n}\mathds{E}_{i-1}\left[\left(\sum_{k=1}^{d}|v^{(k)}|V_{n,i}^{(k)}(E)\right)^{2}\mathbf{1}\left\{\sum_{k=1}^{d}|v^{(k)}|V_{n,i}^{(k)}(E)>\epsilon\right\}\right].\label{eq: pf_linderberge_envolope}
\end{align}

Then, we will use Assumption \ref{assu:(Lindeberg-Feller_multidim)} to show that the RHS above goes to zero as $n\rightarrow\infty$. 
From Assumption \ref{assu:(Lindeberg-Feller_multidim)}, we have that for each $k=1,2,\cdots,d$, and $\forall\epsilon>0$, 
\begin{align}
&\sum_{i=1}^{n}\mathds{E}_{i-1}\left[\left(|v^{(k)}|V_{n,i}^{(k)}(E)\right)^{2}\mathbf{1}\left\{|v^{(k)}|V_{n,i}^{(k)}(E)>\epsilon\right\}\right]\nonumber \\ 
&=\left|v^{(k)}\right|^2\sum_{i=1}^{n}\mathds{E}_{i-1}\left[\left(V_{n,i}^{(k)}(E)\right)^{2}\mathbf{1}\left\{V_{n,i}^{(k)}(E)>\frac{\epsilon}{|v^{(k)}|}\right\}\right]\stackrel{P}{\longrightarrow}0.\label{eq: proof_linderberg}
\end{align}
Based on this result, we could bound the RHS of (\ref{eq: pf_linderberge_envolope})
as follows
\begin{align*}
 & \sum_{i=1}^{n}\mathds{E}_{i-1}\left[\left(\sum_{k=1}^{d}|v^{(k)}|V_{n,i}^{(k)}(E)\right)^{2}\mathbf{1}\left\{\sum_{k=1}^{d}|v^{(k)}|V_{n,i}^{(k)}(E)>\epsilon\right\}\right]\\
\leq & \sum_{i=1}^{n}\mathds{E}_{i-1}\left[\left(\sum_{k=1}^{d}|v^{(k)}|V_{n,i}^{(k)}(E)\right)^{2}\left(\sum_{k=1}^{d}\mathbf{1}\left\{|v^{(k)}|V_{n,i}^{(k)}(E)>\frac{\epsilon}{d}\right\}\right)\right]\\
\leq & d\sum_{i=1}^{n}\mathds{E}_{i-1}\left[\sum_{k=1}^{d}\left(|v^{(k)}|V_{n,i}^{(k)}(E)\right)^{2}\left(\sum_{k=1}^{d}\mathbf{1}\left\{|v^{(k)}|V_{n,i}^{(k)}(E)>\frac{\epsilon}{d}\right\}\right)\right]\\
\leq & d^{2}\sum_{k=1}^{d}\sum_{i=1}^{n}\mathds{E}_{i-1}\left[\left(|v^{(k)}|V_{n,i}^{(k)}(E)\right)^{2}\mathbf{1}\left\{|v^{(k)}|V_{n,i}^{(k)}(E)>\frac{\epsilon}{d}\right\}\right] \stackrel{P}{\longrightarrow}0~~\text{as}~~n\rightarrow\infty.
\end{align*}
In the last inequality, we used rearrangement inequality and it goes to zero according to \eqref{eq: proof_linderberg}. Then applying (\ref{eq: pf_linderberge_envolope}), (C2) is verified. \hfill{$\Box$}
\endproof

\subsection{Proof of Theorem \ref{prop:RM_consistency_multidim} (Multidimensional Version of Theorem \ref{prop:RM_consistency})}\label{sec: proof_consistency_SA}
\proof{Proof:}
 We verify the conditions for Theorem 2.3 in Section 5.2 of \cite{kushner2003stochastic}. The constraint set condition we impose in Assumption \ref{assu: unique_solution_multidim} satisfies (A4.3.2) in \cite{kushner2003stochastic}. Then we check (A2.1)-(A2.6) in Section 5.2 of \cite{kushner2003stochastic}. Corresponding to our algorithm, $Y_n$ in \cite{kushner2003stochastic} is given by $\bF (\bX_{n+1},\hat{\betheta}_n)\ell (\bX_{n+1},\bealpha_{n+1})-\bc$. So (A2.1) is checked by Assumption \ref{assu: sup_L2_multidim}. (A2.2) holds with $\beta_n=0$. (A2.3) is checked by the continuity of $\boldf(\betheta)$. And it is also easy to check the choice of stepsize in our algorithm satisfies
(A2.4). (A2.5) holds since $\beta_n=0$. (A2.6) follows from our Assumption \ref{assu: continuity_multidim}. Now we have verified the assumptions. And our Assumption \ref{assu: unique_solution_multidim} implies that there is only one stationary point to the projected ODE. So, by Theorem 2.3 in Section 5.2 of \cite{kushner2003stochastic}, we have that $\hat{\betheta}_n\rightarrow\betheta^*$ a.s., as desired.

\hfill{$\Box$}\endproof

\subsection{Proof of Theorem \ref{prop: SA_CLT_multidim} (Multidimensional Version of Theorem \ref{prop: SA_CLT})}\label{appx:proofTheorem3}
Before proving Theorem \ref{prop: SA_CLT_multidim}, we need the following lemma.
\begin{lemma}[Theorem 2.2, \citealt{fabian1968}]\label{lem:fabian}
Suppose $k$ is a positive integer, $\mathcal{F}_{n}$ a non-decreasing sequence of $\sigma$-fields, $\mathcal{F}_{n} \subset \mathcal{F}$; suppose $\bU_{n}, \bV_{n}, \bT_{n} \in \mathbb{R}^{k}, \bT \in \mathbb{R}^{k}, \bGamma_{n}, \bPhi_{n} \in \mathbb{R}^{k \times k}, \bSigma, \bGamma, \bPhi, \bP \in \mathbb{R}^{k\times k}$,
$\bGamma$ is positive definite, $\bP$ is orthogonal and $\bP^{\top} \bGamma \bP=\bLambda$ diagonal. 
Suppose $\bGamma_{n}, \bPhi_{n-1},$ $\bV_{n-1}$ are $\mathcal{F}_{n}$-measurable, $C, \alpha, \beta \in \mathbb{R}$ and
\begin{equation}
\bGamma_{n} \rightarrow \bGamma, ~~\bPhi_{n} \rightarrow \bPhi, ~~\bT_{n} \rightarrow \bT ~~ \text{or}~~ \mathds{E}\left[\left\|T_{n}-T\right\|\right] \rightarrow 0,
\end{equation}
\begin{equation}\label{eq: limit_conditionalvar_var}
\mathds{E}_{n} [\bV_{n}]=\bzero, ~~ C>\left\| \mathds{E}_{n}\left[\bV_{n} \bV_{n}^{\top}\right]-\bSigma\right\| \rightarrow 0,
\end{equation}
and, with 
\begin{equation}
\sigma_{j, r}^{2}=\mathds{E} \left[\left\|\bV_{j}\right\|^{2} \mathbf{1}\left\{\left\|\bV_{j}\right\|^{2} \geq r j^{\alpha}\right\}\right],
\end{equation}
let either
\begin{equation}
\lim _{j \rightarrow \infty} \sigma_{j, r}^{2}=0 ~~ \text { for every }r>0,
\end{equation}
or
\begin{equation}
\alpha=1, ~~ \lim _{n \rightarrow \infty} n^{-1} \sum_{j=1}^{n} \sigma_{j, r}^{2}=0 ~~ \text { for every }r>0.
\end{equation}
Suppose that, with $\lambda=\min _{i} [\Lambda]_{ii}$, $\beta_{+}=\beta$ if $\alpha=1$, $\beta_{+}=0$ if $\alpha \neq 1$,
\begin{equation}
0<\alpha \leqq 1, \quad 0 \leqq \beta, \quad \beta_{+}<2 \lambda
\end{equation}
$\operatorname{and}$
\begin{equation}
\bU_{n+1}=\left(\bI-n^{-\alpha} \bGamma_{n}\right) \bU_{n}+n^{-(\alpha+\beta) / 2} \bPhi_{n} \bV_{n}+n^{-\alpha-\beta / 2} \bT_{n}.
\end{equation}
Then the asymptotic distribution of $n^{\beta / 2} \bU_{n}$ is normal with mean $\left(\bGamma-\left(\beta_{+} / 2\right) I\right)^{-1} \bT$
and covariance matrix $\bP \bM \bP^{\top}$, where
\begin{equation}
[\bM]_{i j}=\left[\bP^{\top} \bPhi \bSigma \bPhi^{\top} \bP\right]_{i j}\left([\bLambda]_{i i}+[\bLambda]_{j j}-\beta_{+}\right)^{-1}.
\end{equation}
\end{lemma}

Next, we will prove Theorem \ref{prop: SA_CLT_multidim}.
\proof{Proof of Theorem \ref{prop: SA_CLT_multidim}:}
From Lemma \ref{lem:fabian}, with the same argument as in the analysis of SAA (see condition (C1) in the proof of Theorem \ref{thm: SAA_CLT_multidim}), we have that $\mathds{E}_n [\bV_n\bV_n^{\top}]\rightarrow\bSigma$ where $\bSigma=\Var_{\bX\sim P_{\bealpha^*}}\left(\bF(\bX,\betheta)\ell(\bX,\bealpha^*)\right)$. This implies condition \eqref{eq: limit_conditionalvar_var} of Lemma \ref{lem:fabian}. By the mean value theorem,
there exists $\bxi_{n}$ on the line segment between $\hat{\btheta}_{n}$
and $\btheta^{*}$ such that 
\[
\boldf(\hat{\btheta}_{n})=\frac{D}{D\betheta}\boldf(\bxi_{n})(\hat{\btheta}_{n}-\btheta^{*}).
\]
Then let $\bY_{n}=\bF(\bX_{n+1},\hat{\btheta}_{n})\ell(\bX_{n+1},\balpha_{n+1})-\bc,\bV_{n}=\bY_{n}+\bc-\boldf(\hat{\btheta}_{n}),\bU_{n}=\hat{\btheta}_{n}-\btheta^{*},\alpha=1,\beta=1,\bGamma_{n}=\gamma\bJ(\bxi_{n}),\bPhi_{n}=-\gamma \bI_d.$ Let 
\[
\bT_n = \Pi_A\left[\hat{\btheta}_{n}-\gamma_{n}\bF(\bX_n,\hat \btheta_{n-1})\ell(\bX_{n},\balpha_{n})\right] - \left[\hat{\btheta}_{n}-\gamma_{n}\bF(\bX_n,\hat \btheta_{n-1})\ell(\bX_{n},\balpha_{n})\right]
\]
be the projection term. Then one can check that we could write our iteration (\ref{eq:SA_update_multidim})
as follows:
\[
\bU_{n+1}=(\bI-n^{-\alpha}\bGamma_{n})\bU_{n}+n^{-(\alpha+\beta)/2}\bPhi_{n}\bV_{n}+n^{-\alpha-\beta/2}\bT_{n}.
\]

Notice that from the fact that $\balpha_{n+1}\in\mathcal{F}_{n}$, $\bV_{n}$ is a martingale difference array, i.e., $\mathds{E}_{n}[\bV_{n}]=\mathds{E}_{n}[\bF(\bX_{n+1},\hat{\btheta}_{n})\ell(\bX_{n+1},\balpha_{n+1})]-\boldf(\hat{\btheta}_{n})=\bzero$. From the strong consistency $\hat{\betheta}_n\rightarrow\betheta^* $ a.s., 
and the assumption that $\btheta^{*}\in A^{o}$, we have that with probability one,
there exists $N$ such that $\bT_{n}=\bzero$ after $n\geq N$. Hence $\bT_n\rightarrow \bT = \bzero $, so the condition for $\bT_n$ in the theorem is satisfied with $\bT=\bzero$. The condition $\mathds{E}_{n}\left[\bV_{n} \bV_{n}^{\top}\right]\rightarrow\bSigma$ follows from the consistency of $\bealpha$ and the continuity of the variance with respect to $\bealpha$ (the same as the proof of (C1) in \ref{sc:finalproof_ThmA3}). 
The other conditions of the theorem just follow from definition and our assumptions. And the result of the theorem tells us that 
\[
\sqrt{n}(\hat{\btheta}_{n}-\btheta^{*})\Rightarrow N(\bzero,\bP\bM\bP^{\top}).
\]
Then applying the Delta method, we will get the claim as in the theorem.
\hfill{$\Box$}\endproof

\subsection{Proof of Theorem \ref{prop: SA_average_CLT_multidim} (multidimensional version of Theorem \ref{prop: SA_average_CLT})}\label{appx:proofprop6}

\begin{proof}{Proof:}
Our proof is based on Theorem 2 of \cite{polyak1992}. First we verify Assumptions 3.1-3.4 of \cite{polyak1992}. For Assumption 3.1, we can pick $V(\theta) = h(\betheta)$ where $h(\betheta)$ is the function such that $\nabla h =\boldf$ in our Assumption \ref{assu: unique_solution_multidim}. Then we can verify Assumptions 3.1-3.2 of \cite{polyak1992} based on the positive definiteness of the Jabobian of $\boldf$ at $\betheta^*$ and the second-order differentiability. Assumption 3.3 follows from our Assumptions \ref{assu: sup_L2_multidim} and \ref{assu: Polyak_noise_multidim}. And it is easy to check our stepsize satisfies Assumption 3.4 of \cite{polyak1992}. Thus we have verified the conditions for \cite{polyak1992}. Applying that result we will get our claimed result. 

In the above argument, we have assumed the validity of Theorem 2 of \cite{polyak1992} under projection, which is not present in their original theorem. We argue that adding a projection does not affect the asymptotic result.
Actually, when projection is performed, we can use $\bxi_{n}+\bz_{n}$ to replace $\bxi_{n}$ in the proof of \cite{polyak1992}, where $\bz_{n}$ accounts for the change induced by projection in each update of $\btheta$.
Since $\bz_{n}$ will become zero when $n$ is large, we find that performing this replacement will not affect the proof. 
So the same result holds.

\hfill{$\Box$}\end{proof}


\subsection{Proof of Theorem \ref{prop: SAA_QE_opt}}\label{appx:proof_prop_SAA_QE_opt}

\begin{proof}{Proof:}
First we check the assumptions for Theorem \ref{prop:SAA_consistency} and get the strong consistency $\hat{q}_n\rightarrow q^*$ a.s. 

Assumption \ref{assu: truncation_alpha} corresponds to Assumption \ref{assu: SAA_QE_consistency}.
A difference is that the result holds for $q^{*}+\delta$ (not all $q$). 
Actually, from the proof of Lemma \ref{lem: average_consistent_multidim},
we have that under Assumption \ref{assu: SAA_QE_consistency}, 
\[
\frac{1}{n}\sum_{i=1}^{n}\mathbf{1}\{h(\bX_{i})\leq q^{*}+\delta\}\ell(\bX_{i},\balpha_{i})-F_h(q^{*}+\delta)\rightarrow0
\]
for $q=q^{*}+\delta$. Hence when $n$ is large enough, 
\[
\frac{1}{n}\sum_{i=1}^{n}\mathbf{1}\{h(\bX_{i})\leq q^{*}+\delta\}\ell(\bX_{i},\balpha_{i})\geq F_h(q^{*})=p,
\]
which implies
\[
\hat{q}_{n}=\inf\left\{q:\frac{1}{n}\text{\ensuremath{\sum_{i=1}^{n}\mathbf{1}\{h(\bX_{i})\geq q\}\ell(\bX_{i},\balpha_{i})}}\geq p\right\}\leq q^{*}+\delta.
\]


For Assumption \ref{assu: bracket}, for any $\epsilon>0$ choose $\{p_i\}$ to be a set of $[0,1]$-valued real numbers such that $(p_i-\epsilon/2, p_i+\epsilon/2)$ forms a $\epsilon$-net of $[0,1]$. Then let $Q_\epsilon =\{q_i: P\{h(\bX)\leq q_i\}=p_i-\epsilon/2~\text{or}~P\{h(\bX)\leq q_i\}=p_i+\epsilon/2\}$. Then define
\begin{equation*}
K_{\epsilon}=\{\left(f_{L}(\bX,\balpha),f_{R}(\bX,\balpha)\right):=\left(\mathbf{1}\{h(\bX)\leq q_{L}\}\ell(\bX,\balpha),\mathbf{1}\{h(\bX)\leq q_{R}\}\ell(\bX,\balpha)\right)|q_{L},q_{R}\in Q_{\epsilon}\}.
\end{equation*}
We only need at most $\left\lceil {1}/{\epsilon}\right\rceil $
elements in $Q_{\epsilon}$ and this implies that $K_{\epsilon}$ is a finite set.

For Assumption \ref{assu: Obj_regular}, it follows from Assumption \ref{assu: SAA_QE_obj}
and the monotonicity of the objective function. 

Then we apply Theorem \ref{prop:SAA_consistency}. One difference in the quantile estimation case is that \[
\hat{q}_{n}=\inf\left\{q:\text{\ensuremath{\frac{1}{n}\sum_{i=1}^{n}\mathbf{1}\{h(\bX_{i})\leq q\}\ell(\bX_{i},\balpha_{i})}}\geq p\right\}
\]
is not the exact solution to $\sum_{i=1}^{n}\mathbf{1}\{h(\bX_{i})\leq q\}\ell(\bX_{i},\balpha_{i})/n=p$ as in Algorithm \ref{alg:SAA_rootfinding_blackbox}. But
by Lemma \ref{Lem: unif_converg_aver_obj}, we have uniform convergence of $\sum_{i=1}^{n}\mathbf{1}\{h(\bX_{i})\leq q\}\ell(\bX_{i},\balpha_{i})/n$
to $F_h(q)$, which implies that $\hat{q}_n$
is an approximate solution in the sense that 
\[
\frac{1}{n}\sum_{i=1}^{n}\mathbf{1}\{h(\bX_{i})\leq\hat{q}_{n})\ell(\bX_{i},\balpha_{i})\rightarrow p~~\text{as}~~n\to\infty.
\]
With this condition, by checking the proof of Theorem \ref{prop:SAA_consistency}, we get $\hat{q}_n\rightarrow q^*$ a.s.

Next, we use a somehow different argument to show the asymptotic normality result. The proof is similar in spirit to the proof of Theorem \ref{thm: SAA_CLT_multidim}, but we can exploit the special structure of quantile estimation. By doing this, the required condition would be slightly milder. More precisely, we will not need the counterpart of Assumption \ref{assu:(uniform-integrable-entropy)}.  

For any $t\in\mathbb{R}$, define the martingale difference triangular
array:
\[
Z_{n,i}=\frac{\mathbf{1}\{h(\mathbf{X}_{i})\leq(q+t\sqrt{n^{-1}})\}\ell(\mathbf{X}_{i},\boldsymbol{\alpha}_{i})-F_h(q+t\sqrt{n^{-1}})}{\sqrt{n \Var_{\bX\sim P_{\boldsymbol{\alpha}^{*}}}\left(\mathbf{1}\left\{ h(\mathbf{X})\leq(q+t\sqrt{n^{-1}})\right\} \ell(\mathbf{X},\boldsymbol{\alpha}^{*})\right)}}.
\]
We will use Theorem 8.2.4 of \cite{Durret2019}. We observe that for any $u\in[0,1]$,
\begin{align*}
 & \lim_{n\rightarrow\infty}\sum_{m=1}^{\left\lfloor nu \right\rfloor}\mathds{E}[Z_{n,m}^{2}|\mathcal{F}_{m-1}]\\
= & \lim_{n\rightarrow\infty}\frac{1}{n}\sum_{m=1}^{\left\lfloor nu \right\rfloor}\frac{\Var_{\bX\sim P_{\boldsymbol{\alpha}_{m}}}\left(\mathbf{1}\left\{ h(\mathbf{X})\leq(q^{*}+t\sqrt{n^{-1}})\right\} \ell(\mathbf{X},\boldsymbol{\alpha}_{m})\right)}{\Var_{\bX\sim P_{\boldsymbol{\alpha}^{*}}}\left(\mathbf{1}\left\{ h(\mathbf{X})\leq(q^{*}+t\sqrt{n^{-1}})\right\} \ell(\mathbf{X},\boldsymbol{\alpha}^{*})\right)},
\end{align*}
where $\left\lfloor nu \right\rfloor$ is the greatest integer less than or equal to $nu$.

By our Assumption \ref{assu: continuous_variance}, for any $\epsilon>0$ we can find $\delta>0$
such that whenever $\left\Vert \boldsymbol{\alpha}-\boldsymbol{\alpha}^{*}\right\Vert \leq\delta$, $\left|t\sqrt{n^{-1}}\right|\leq\delta$,
we have 
\[
\left|\Var_{\bX\sim P_{\boldsymbol{\alpha}}}\left(\mathbf{1}\left\{ h(\mathbf{X})\leq(q+t\sqrt{n^{-1}})\right\} \ell(\mathbf{X},\boldsymbol{\alpha})\right)-\Var_{\bX\sim P_{\boldsymbol{\alpha}^{*}}}\left(\mathbf{1}\left\{ h(\mathbf{X})\leq q^{*}\right\} \ell(\mathbf{X},\boldsymbol{\alpha}^{*})\right)\right|\leq\epsilon.
\]

From the continuity of $I$ and the consistency of $\hat{q}_{i}$,
we have that $\boldsymbol{\alpha}_{i}\rightarrow\boldsymbol{\alpha}^*$ a.s. So there must exist $N=N(\omega)$ such that 
\[
\left|\Var_{\bX\sim P_{\boldsymbol{\alpha}_{m}}}\left(\mathbf{1}\left\{ h(\mathbf{X})\leq(q+t\sqrt{n^{-1}})\right\} \ell(\mathbf{X},\boldsymbol{\alpha}_{m})\right)-\Var_{\bX\sim P_{\boldsymbol{\alpha}^{*}}}\left(\mathbf{1}\left\{ h(\mathbf{X})\leq q^{*}\right\} \ell(\mathbf{X},\boldsymbol{\alpha}^{*})\right)\right|\leq\epsilon
\]
whenever $m>N$. Thus with probability one, 
\[
\lim_{n\rightarrow\infty}\sum_{m=1}^{\left\lfloor nu \right\rfloor}\mathds{E}[Z_{n,m}^{2}|\mathcal{F}_{m-1}]=u.
\]
Also notice that by Proposition \ref{prop: LF_sufficient}, our Assumption \ref{assu: QE_Feller} implies the second condition of Theorem 8.2.4 of \cite{Durret2019}, so using that theorem we get  
\[
\sum_{i=1}^{n}Z_{n,i}\Rightarrow \mathcal{N}(0,1).
\]
With this result, consider 
\[
G_{n}(t)=P(\sqrt{n}(\hat{q}_{n}-q^{*})\leq t),
\]
which could be written as 
\begin{align*}
 & G_{n}(t) =  P(\hat{q}_{n}\leq q^{*}+t\sqrt{n^{-1}})\\
= & P\left(p\leq\frac{1}{n}\sum_{m=1}^{n}\left(\mathbf{1}\left\{ h(\mathbf{X})\leq(q^{*}+t\sqrt{n^{-1}})\right\} \ell(\mathbf{X},\boldsymbol{\alpha}_{m})\right)\right)\\
= & P\left(\sum_{i=1}^{n}Z_{n,i}\geq\frac{-\sqrt{n}F_h\left(\left(q^{*}+t\sqrt{n^{-1}}\right)-p\right)}{\sqrt{\Var_{\bX\sim P_{\boldsymbol{\alpha}^{*}}}\left(\mathbf{1}\left\{ h(\mathbf{X})\leq q^{*}\right\} \ell(\mathbf{X},\boldsymbol{\alpha}^{*})\right)}}\right).
\end{align*}
From the differentiability of $F$ we know that $\sqrt{n}F_h\left(\left(q^{*}+t\sqrt{n^{-1}}\right)-p\right)\rightarrow tf_h(q^{*})$.
Thus by Slutsky's theorem, the above probability converges to 
\begin{equation*}
\Phi\left(\frac{tf_h(q^{*})}{\sqrt{\Var_{\bX\sim P_{\boldsymbol{\alpha}^{*}}}\left(\mathbf{1}\left\{ h(\mathbf{X})\leq q^{*}\right\} \ell(\mathbf{X},\boldsymbol{\alpha}^{*})\right)}}\right),
\end{equation*}
which is equal to 
\begin{equation*}
P\left(\mathcal{N}\left(0,\frac{\Var_{\bX\sim P_{\boldsymbol{\alpha}^{*}}}\left(\mathbf{1}\left\{ h(\mathbf{X})\leq q^{*}\right\} \ell(\mathbf{X},\boldsymbol{\alpha}^{*})\right)}{\left(f_h(q^{*})\right)^{2}}\right)\leq t\right).
\end{equation*}
Thus, as claimed,
\begin{equation*}
\sqrt{n}(\hat{q}_{n}-q^{*})\Rightarrow \mathcal{N}\left(0,\frac{\Var_{\bX\sim P_{\boldsymbol{\alpha}^{*}}}\left(\mathbf{1}\left\{ h(\mathbf{X})\leq q^{*}\right\} \ell(\mathbf{X},\boldsymbol{\alpha}^{*})\right)}{\left(f_h(q^{*})\right)^{2}}\right).
\end{equation*}

\hfill{$\Box$}
\end{proof}


\section{Proofs on Assumption Verification}
This section provides proofs of Propositions \ref{prop: verify_entropy}-\ref{prop: PN_sufficient}, which verify some critical technical assumptions used in proving theoretical convergences of the algorithms. This section also provides proofs of Propositions \ref{prop: verification_normal} and \ref{prop: verification_pVaR}, which verify the assumptions in the numerical examples.
\subsection{Proof of Proposition \ref{prop: verify_entropy}}
\proof{Proof:}
For each $\epsilon$, we let $\Pi_1(\epsilon)=\{A_1(\epsilon),A_2(\epsilon),\dots,A_{m(\epsilon)}(\epsilon)\}$ be an $\epsilon$-covering of $\Theta_\delta$ under $\rho$ where $m(\epsilon)=N(\epsilon,\Theta_{\delta},\rho)$.  Then $\Pi(\epsilon)=\{\{f_{\theta}:\theta\in A_1(\epsilon)\},\dots,\{f_{\theta}:\theta\in A_{m(\epsilon)}(\epsilon)\}\}$ is an $\epsilon$-covering of $\sfF_{\delta}$ under $\rho$. We argue that this choice of $\Pi$ satisfies the conditions in Assumption \ref{assu:(uniform-integrable-entropy)}. The condition in the second display of Assumption \ref{assu:(uniform-integrable-entropy)} is verified by condition (i) of this proposition. It remains to verify the condition in the first display of Assumption \ref{assu:(uniform-integrable-entropy)}. 

Recall that $V_{n,i}(f_{\theta}) = (F(\bX_i,\theta)\ell(\bX_i,\bealpha_i)-f(\theta))/{\sqrt{n}}$. So for any $\theta_1,\theta_2\in\Theta_{\delta}$, we have that 
\begin{eqnarray*}
n(V_{n,i}(f_{\theta_1})-V_{n,i}(f_{\theta_2}))^2 &\leq& 2(F(\bX_i,\theta_1)-F(\bX_i,\theta_2))^2{(\ell(\bX_i,\bealpha_i))^2} + 2(f(\theta_1)-f(\theta_2))^2 \\
&\leq& 2(F(\bX_i,\theta_1)-F(\bX_i,\theta_2))^2{(\ell(\bX_i,\bealpha_i))^2} + 2(\rho(\theta_1,\theta_2))^2\\
&\leq& 2(L(\bX_i,\bealpha_i)\ell(\bX_i,\bealpha_i)+1)(\rho(\theta_1,\theta_2))^2.
\end{eqnarray*}
Here the last inequality follows from condition (ii) of this proposition. For any $f_{\theta_1},f_{\theta_2}$ in the same $\epsilon$-ball $\sfF(\epsilon;k)$, we have that $\rho(\theta_1,\theta_2)\leq 2\epsilon$, so
\[
n(V_{n,i}(f_{\theta_1})-V_{n,i}(f_{\theta_2}))^2\leq 8(L(\bX_i,\bealpha_i)\ell(\bX_i,\bealpha_i)+1)\epsilon^2.
\]
This implies $(V_{n,j}(\sfF(\epsilon;k)))^2\leq 8(L(\bX_i,\bealpha_i)\ell(\bX_i,\bealpha_i)+1)\epsilon^2$, so 
\begin{eqnarray*}
\frac{\sqrt{\sum_{j=1}^{n}\mathds{E}_{j-1}[\left\vert V_{n,j} (\sfF(\epsilon;k))\right\vert^2]}}{\epsilon}&\leq& \sqrt{8+\frac{8}{n}\sum_{i=1}^n\mathds{E}_{\bX_i\sim P_{\bealpha_i}}[L(\bX_i,\bealpha_i)\ell(\bX_i,\alpha_i)]}\\
&=&\sqrt{8+\frac{8}{n}\sum_{i=1}^n\mathds{E}_{\bX\sim P}[L(\bX,\bealpha_i)]}.
\end{eqnarray*}
The RHS is independent of $k$ and $\epsilon$. Thus taking supremum w.r.t. $k$ and $\epsilon$, we get that
\[
\sup_{\epsilon\in(0,\Delta_{\Pi}]\cap \mathbb{Q}} \max_{1\leq k\leq N_{\Pi}(\epsilon)}\frac{\sqrt{\sum_{j=1}^{n}\mathds{E}_{j-1}[\left\vert V_{n,j} (\sfF(\epsilon;k))\right\vert^2]}}{\epsilon}\leq \sqrt{8+\frac{8}{n}\sum_{i=1}^n\mathds{E}_{\bX\sim P}[L(\bX,\bealpha_i)]}.
\]
Since $\bealpha_i\rightarrow\bealpha$ a.s., and $\sup_{\left\Vert\bealpha-\bealpha^*\right\Vert\leq \delta_1}\mathds{E}[L(\bX,\bealpha)]<\infty$ as assumed in this proposition, we have that the RHS is $O_p(1)$. Thus we have shown the desired result. 
\hfill{$\Box$}
\endproof

\subsection{Proof of Proposition \ref{prop: LF_sufficient} and Proposition \ref{prop: PN_sufficient}}{\label{appx:proofprop13}}
\proof{Proof:}
We know that 
\[
\sum_{i=1}^{n}\mathds{E}_{i-1}\left[V_{n,i}(E)^{2}\mathbf{1}\{V_{n,i}(E)>\epsilon\}\right]=\frac{1}{n}\sum_{i=1}^{n}\mathds{E}_{i-1}\left[V_{i}(E)^{2}\mathbf{1}\{V_{i}(E)>\epsilon\sqrt{n}\}\right],
\]
where $V_{i}(E)=\sqrt{n}V_{n,i}(E)$, and from the definition of $V_{n,i}(E)$,
we have that $V_{i}(E)$ is envelope function for 
\[
\left\{ \left|F(\bX,\theta)\ell(\bX,\bealpha)-f(\theta)\right|:\theta\in\Theta_{\delta}\right\}.
\]
Notice that for any $R$, $\epsilon\sqrt{n}>R$ when $n$ is large
enough. So in order to show that (Proposition \ref{prop: LF_sufficient})
\[
\frac{1}{n}\sum_{i=1}^{n}\mathds{E}_{i-1}\left[V_{i}(E)^{2}\mathbf{1}\{V_{i}(E)>\epsilon\sqrt{n}\}\right]\rightarrow 0,
\]
it suffices to show that, with probability 1, there exists $N(\omega)<\infty$
such that 
\begin{equation}
\lim_{R\rightarrow\infty}\sup_{n\geq N(\omega)}\mathds{E}_{i-1}\left[V_{i}(E)^{2}\mathbf{1}\{V_{i}(E)>R\}\right]=0.\label{eq: LF_proof}
\end{equation}
We can see that this also implies the condition in Assumption \ref{assu: Polyak_noise}. So we will show \eqref{eq: LF_proof} for the rest of the proof. 

Let $C=\sup_{\theta\in\Theta_{\delta}}f(\theta)$. Notice that for
all $\theta\in\Theta_{\delta},\left\Vert \balpha-\balpha^{*}\right\Vert \leq\delta_{1}$, by the assumption of Proposition \ref{prop: LF_sufficient}, 
\[
V(\bX)\ell(\bX,\bealpha)\geq\left|F(\bX,\theta)\ell(\bX,\bealpha)\right|^{2}\geq\left(\left|F(\bX,\theta)\ell(\bX,\bealpha)-f(\theta)\right|-C\right)^{2}.
\]
When $\sup_{\theta\in\Theta_{\delta}}\left|F(\bX,\theta)\ell(\bX,\balpha)-f(\theta)\right|\geq2C$,
taking supremum over $\theta$ in the above inequality, we get
\[
V(\bX)\ell(\bX,\bealpha)\geq\left(\sup_{\theta\in\Theta_{\delta}}\left|F(\bX,\theta)\ell(\bX,\bealpha)-f(\theta)\right|-C\right)^{2}\geq\frac{1}{4}\left(\sup_{\theta\in\Theta_{\delta}}\left|F(X,\theta)\ell(\bX,\bealpha)-f(\theta)\right|\right)^{2}.
\]
So by arguing the case when $\sup_{\theta\in\Theta_{\delta}}\left|F(\bX,\theta)\ell(\bX,\balpha)-f(\theta)\right|\geq2C$ and the case when $\sup_{\theta\in\Theta_{\delta}}\left|F(\bX,\theta)\ell(\bX,\balpha)-f(\theta)\right|\leq 2C$ separately, we always have 
\[
\max\{V(\bX)\ell(\bX,\alpha),C^{2}\}\geq\frac{1}{4}\left(\sup_{\theta\in\Theta_{\delta}}\left|F(\bX,\theta)\ell(\bX,\bealpha)-f(\theta)\right|\right)^{2}.
\]
Thus from the definition of $V_{i}(E)$ we have that when $\left\Vert \bealpha_{i}-\bealpha^{*}\right\Vert \leq\delta_{1}$,
\[
\max\{V(\bX_{i})\ell(\bX_{i},\bealpha_{i}),C^2\}\geq\frac{1}{4}(V_{i}(E))^{2}.
\]
Now we take $R>C^2$. The above relation tells us that when $V_{i}(E)>R$,
we must have $V(\bX_{i})\ell(\bX_{i},\bealpha_{i})\geq (V_{i}(E))^{2}/4$.
Thus 
\[
V_{i}(E)^{2}\mathbf{1}\{V_{i}(E)>R\}\leq4V(\bX_{i})\ell(\bX_{i},\balpha_{i})\mathbf{1}\left\{V(\bX_{i})\ell(\bX_{i},\balpha_{i})\geq\frac{1}{4}R^{2}\right\}.
\]
So given $\left\Vert \balpha_{i}-\balpha^{*}\right\Vert \leq\delta_{1}$,
we have that 
\begin{align*}
 & \mathds{E}_{i-1}\left[V_{i}(E)^{2}\mathbf{1}\left\{V_{i}(E)\geq R\right\}\right]\\
\leq & 4\mathds{E}_{i-1}\left[V(\bX_{i})\ell(\bX_{i},\balpha_{i})\mathbf{1}\left\{V(\bX_{i})\ell(\bX_{i},\bealpha_{i})\geq\frac{1}{4}R^{2}\right\}\right]\\
= & 4\mathds{E}_{\bX\sim P}\left[V(\bX)\mathbf{1}\left\{V(\bX)\ell(\bX,\bealpha_{i})\geq\frac{1}{4}R^{2}\right\}\right]\\
\leq & 4\mathds{E}_{\bX\sim P}\left[V(\bX)\mathbf{1}\left\{\sup_{\left\Vert \bealpha-\bealpha^{*}\right\Vert \leq\delta_{1}}V(\bX)\ell(\bX,\bealpha)\geq\frac{1}{4}R^{2}\right\}\right].
\end{align*}
As $R\rightarrow\infty$, $P\left(\sup_{\left\Vert \bealpha-\bealpha^{*}\right\Vert \leq\delta_{1}}V(\bX)\ell(\bX,\bealpha)\geq R^{2}/4\right)\rightarrow0$,
thus given $\left\Vert \balpha_{i}-\balpha^{*}\right\Vert \leq\delta_{1}$,
we have that $\lim_{R\rightarrow\infty}\mathds{E}_{i-1}\left[V_{i}(E)^{2}\mathbf{1}\{V_{i}(E)\geq R\}\right]=0$.
From the consistency of $\balpha_i$, there always exists $N(\omega)$
such that $\left\Vert \balpha_{i}-\balpha^{*}\right\Vert \leq\delta_{1}$
when $i>N(\omega)$. So \eqref{eq: LF_proof} is proved.
\hfill{$\Box$}
\endproof

\subsection{Assumption Verification in Numerical Examples}
In this subsection, we verify the assumptions for algorithms introduced in Section \ref{sec:example}. Since most of the algorithms in Section \ref{sec:example} considers extreme large quantile, we verify the assumptions using the estimator of form $\mathbf{1}\{Z\geq x\}\ell (Z,\alpha)$ (see the discussion at the end of Section \ref{sec:quantile}). 

For these quantile estimation algorithms, we will check the assumption introduced in Section \ref{sec:quantile}. The assumptions regarding the smoothness and regularity of the objective functions are all easy to check because in the examples, all of the distribution functions are smooth with positive gradients around the target quantile. 
Other assumptions that may require a little bit of algebra to show are Assumptions \ref{assu: SAA_QE_consistency}, \ref{assu: QE_Feller} and \ref{assu: QE_Fabian_noise}. They are summarized as follows. 
\begin{itemize}
\item \textbf{(SAA1)} For the SAA method, there exists a $\delta>0$ such that 
\[
\mathds{E}_{n-1}\left[\left(\mathbf{1}\{Z_n\geq q^{*}-\delta\}\ell(Z_n,\alpha_{n})\right)^{2}\right]=O\left(n^{\frac{1}{2}-\epsilon}\right)
\]
for some $\epsilon>0$.
\item \textbf{(SAA2)} For the SAA method, there exists $\delta,\delta_1>0$ and a function $V(Z)\geq \sup_{\left\Vert\alpha-\alpha^*\right\Vert\leq\delta_1}\mathbf{1}\{Z\leq q^*+\delta\}\ell(Z,\alpha)$ and $\mathds{E}_{Z\sim P}[V(Z)]<\infty$.
\item \textbf{(SA1)} For the SA method, $\sup_{m}\mathds{E}_{m}[|V_{m}|^{2}]<C$ for some constant C, where $V_m = \mathbf{1}\{Z_{n+1}\geq \hat{q}_n\}\ell (Z_{n+1},\alpha_{n+1})-(1-p)$ and $Z_{n+1}\sim P_{\alpha_{n+1}}$.
\end{itemize}


Essentially, all of these assumptions can be regarded as some proper boundedness of the estimator $\mathbf{1}\{Z\geq q\}\ell$. 
Actually, if we can show that $\mathbf{1}\{Z\geq q\}\ell$ has a constant upper bound, then all of the above conditions hold automatically.

\subsubsection{Proof of Proposition \ref{prop: verification_normal}.}
\indent{\bf For (SAA1)}, we have $\ell(x,\alpha_{n})=e^{-\alpha_{n}x+{\alpha_{n}^{2}}/{2}}$.
Then for $\bar{q}=q^{*}-\delta$,
\begin{align*}
\mathds{E}_{n-1}\left[\left(\mathbf{1}\{Z_{n}\geq\bar{q}\}\ell(Z_{n},\alpha_{n})\right)\right]^{2} & =\int_{\bar{q}}^{\infty}e^{-\frac{(x-\alpha_{n})^{2}}{2}}e^{-2\alpha_{n}x+\alpha_{n}^{2}}dx\\
 & =e^{\alpha_{n}^{2}}\int_{\bar{q}}^{\infty}e^{-\frac{(x+\alpha_{n})^{2}}{2}}dx\\
 & \leq e^{\alpha_{n}^{2}}.
\end{align*}
It is clear that our choice of $A_n$ can guarantee that  $e^{\alpha_{n}^{2}}=O(n^{1-\epsilon})$, so we have 
\[
\mathds{E}\left[\mathbf{1}\{Z_{n}\leq\bar{q}\}\ell(Z_{n},\alpha_{n})^{2}\right]=O(n^{1-\epsilon}).
\]

{\bf For (SAA2)}, for any $\alpha\in(\alpha^*-\delta_1,\alpha^*+\delta_1) $, we have that  
$$\ell(x,\alpha) = e^{-\alpha x+\alpha^2/2}\leq e^{-\alpha^* x+\delta_1 \left|x\right|+ \alpha^2/2}\leq e^{-\alpha^* x+(\alpha^*+\delta_1)^2/2}(e^{x\delta_1}+e^{-x\delta_1}).$$
And it is clear that when $x\sim N(0,1)$, the RHS has finite expectation. 

{\bf For (SA1)}, for any $q\in A\subset \mathbb{R}_+$, we know that  $\ell(x,q)$ is bounded by $1$ when $x\geq q$.
Hence $\mathbf{1}\{Z_n\geq\hat{q}_n)\ell(Z_n,\hat{q}_n\}\leq1$ is always true for all $n$. This implies the condition.

\subsubsection{Proof of Proposition \ref{prop: verification_pVaR}.}\label{subsec: verification_pVaR}

\indent{\bf For (SAA1)}, follows from equation \eqref{eq:m2upbound} and $\alpha$ belongs to a bounded set. 

{\bf For (SAA2)}, we notice that  $[I(q_{\min}),I(q_{\max})]$ is a compact subset of $S:=\{\alpha:\psi(\alpha)\  \text{exists and is finite}\} = \{\alpha: 1-2\alpha\lambda_i>0,~i=1,2,\dots,m\}$. Thus we may find $\delta>0$ such that $(\alpha^*-\delta,\alpha+\delta)\subset S$. Then by convexity of $-Q\alpha+\psi(\alpha)$ (as a function of $\alpha$), we have that for any $\alpha\in (\alpha^*-\delta,\alpha^*+\delta)$, 
\begin{align*}
\ell(Q,\alpha) &\leq \exp (\max \{-Q(\alpha^*-\delta)+\psi(\alpha^*-\delta),-Q(\alpha^*+\delta)+\psi(\alpha^*+\delta)\})\\
&\leq \exp (-Q(\alpha^*-\delta)+\psi(\alpha^*-\delta)) +\exp (-Q(\alpha^*+\delta)+\psi(\alpha^*+\delta)).
\end{align*}
The RHS has finite expectation since $(\alpha^*-\delta,\alpha^*+\delta)\subset S$.

{\bf For (SA1)}, notice that the choice of $\alpha_{n+1}$ would minimize the second order moment $m(x,\alpha)$ in \eqref{eq:m2upbound} when $x=\hat{q}_n$. And when $\alpha=0$, $m(x,0)=1$. Thus we always have that $m(\hat{q}_n,\hat{\alpha}_{n+1})\leq 1$. So $\mathds{E}_n [V_n^2]\leq 1$.

\section{Additional Examples}\label{appx:add_example}
This section provides theoretical analysis and numerical results for exponential distribution and Pareto-tailed distribution.
\subsection{Exponential Distribution}


Suppose that the original distribution is exponential with parameter $\lambda$. Given $\hat{q}_{n-1}$, we will find a new exponential
distribution with parameter $\alpha_n=I(\hat{q}_{n-1})$ such that the variance of $\mathbf{1}\{Z_{n}\geq\hat{q}_{n-1}\}\ell(Z_{n},\alpha(\hat{q}_{n-1}))$ will
be minimized. Since the conditional expectation of $\mathbf{1}\{Z_{n}\geq\hat{q}_{n-1}\}\ell(Z_{n},\alpha)$ is always $1-F(\hat{q}_{n-1})$, to minimize its variance, it suffices to minimize its second moment. The second moment
is given by 
\[
\mathds{E}_{Z\sim P_{\alpha}}\left[\mathbf{1}\{Z\geq\hat{q}_{n-1}\}\ell(Z,\alpha)^{2}\right]=\int_{\hat{q}_{n-1}}^{\infty}\alpha e^{-\alpha x}\left(\frac{\lambda e^{-\lambda x}}{\alpha e^{-\alpha x}}\right)^{2}dx=\lambda^{2}\frac{e^{-2\lambda\hat{q}_{n-1}+\alpha\hat{q}_{n-1}}}{\alpha(2\lambda-\alpha)}.
\]
Minimizing over $\alpha\in(0,2\lambda)$, we find that the optimal
choice of $\alpha$ is given by 
\begin{equation}\label{eq:exp_I_function}
I(\hat{q}_{n-1})=\frac{\lambda\hat{q}_{n-1}+1-\sqrt{1+\lambda^{2}\hat{q}_{n-1}^{2}}}{\hat{q}_{n-1}}.
\end{equation}

\subsubsection{Verification of Assumptions.}

\indent{\bf For (SAA1)}, if we choose $\alpha_n=I(\hat{q}_{n-1})$, we have that for any $\bar{q}\in[q-\delta,q+\delta]$,
\begin{align*}
\mathds{E}\left[\left(\mathbf{1}\{Z_{n}\geq\bar{q}\}\ell(X_{n},\alpha_{n})\right)^{2}\right] & =\lambda^{2}\mathds{E}\left[\frac{e^{-2\lambda\bar{q}+\alpha_{n}\bar{q}}}{\alpha_{n}(2\lambda-\alpha_{n})}\right].
\end{align*}
We need this to be $O(n^{{1}/{2}-\varepsilon})$. Since we always have
that $I(\hat{q}_{n-1})\in(0,2\lambda)$, the exponential term
in the numerator is always bounded in interval $[e^{-2\lambda\bar{q}},1]$,
which will not affect the asymptotic rate. Now we consider the denominator.
From the expression of $I(q)$, we know that it is non-increasing
and $I(0+)=\lambda,I(\infty)=0$. Hence $(2\lambda-\alpha_{n})$
is also bounded away from 0. To bound the reciprocal ${1}/{\alpha_{n}}$, we need to introduce truncation set $A_n$ which guarantees that ${1}/{\alpha}=O(n^{1-\epsilon})$ uniformly for $\alpha\in A_n$. And our final choice of $\alpha_n$ is $\Pi_{A_n}I[\hat{q}_{n-1}]$.

{\bf For (SAA2)},
\[
\ell(x,\alpha)=\frac{\lambda e^{-\lambda x}}{\alpha e^{-\alpha x}}.
\]
The discussion for (SAA1) implies that $\alpha^*\in (0,\lambda)$. We also notice that, as long as $\lambda>\alpha >\alpha^*/2>0$, we have
{
\[
\frac{\lambda e^{-\lambda x}}{\alpha e^{-\alpha x}}\leq \frac{2\lambda e^{-\lambda x}}{\alpha^* e^{-\lambda x}} \cdot \frac{2\lambda}{\alpha^*}.
\]}
The RHS is a bounded by constant so of course it has finite expectation. 

{\bf For (SA1)}, we have that
\[
\mathds{E}_{Z\sim P_{\alpha_n}}[\mathbf{1}\{Z\geq\hat{q}_{n-1}\}(\ell(Z_{n},\alpha_{n}))^{2}]=\lambda^{2}\hat{q}_{n-1}^{2}\frac{e^{-\lambda\hat{q}_{n-1}+1-\sqrt{1+\lambda^{2}\hat{q}_{n-1}^{2}}}}{2(\sqrt{1+\lambda^{2}\hat{q}_{n-1}^{2}}-1)}.
\]
From the minimizing property, this is always smaller than the variance
with the original measure, which is $F(\hat{q}_{n-1})(1-F(\hat{q}_{n-1}))\leq {1}/{4}$.

\subsubsection{Asymptotic Results.}

For both SAA and PR-SA, we will reach an asymptotic variance
of 
\[
\frac{\mathds{E}\left[\mathbf{1}\{Z\geq q^{*}\}(\ell(Z,\alpha^*))^{2}\right]-p^{2}}{(f(q^{*}))^2},
\]
which is (notice that $p=1-F(q^{*})=e^{-\lambda q^{*}}$)
\[
\frac{1}{\lambda^{2}e^{-2\lambda q^{*}}}\left(\lambda^{2}(q^{*})^2\frac{e^{-\lambda q^{*}+1-\sqrt{1+\lambda^{2}(q^{*})^2}}}{2(\sqrt{1+\lambda^{2}(q^{*})^2}-1)}-p^{2}\right)=(q^{*})^2\frac{e^{\lambda q^{*}+1-\sqrt{1+\lambda^{2}(q^{*})^2}}}{2(\sqrt{1+\lambda^{2}(q^{*})^2}-1)}-\frac{1}{\lambda^{2}},
\]
hence we have the CLT (for PR-SA, replace $\hat{q}_{n}$
with $\bar{q}_{n}$ here)
\[
\sqrt{n}(\hat{q}_{n}-q^{*})\Rightarrow \mathcal{N}\left(0,(q^{*})^2\frac{e^{\lambda q^{*}+1-\sqrt{1+\lambda^{2}(q^{*})^2}}}{2(\sqrt{1+\lambda^{2}(q^{*})^2}-1)}-\frac{1}{\lambda^{2}}\right).
\]
As $q^{*}$ goes to infinity, the variance will still go to infinity.
But the speed is much slower than ${1}/{(1-F(q^{*}))}=e^{\lambda q^{*}}$.

For RM-SA with stepsize $\gamma_{n}={\gamma}/{n}$,
the asymptotic result would be 
\[
\sqrt{n}(\hat{q}_{n}-q^{*})\Rightarrow \mathcal{N}\left(0,\frac{\gamma^{2}}{2\gamma-\lambda e^{-\lambda q^{*}}}\left(\lambda^{2}(q^{*})^2\frac{e^{-\lambda q^{*}+1-\sqrt{1+\lambda^{2}(q^{*})^2}}}{2(\sqrt{1+\lambda^{2}(q^{*})^2}-1)}-p^{2}\right)\right).
\]

\subsubsection{Numerical Experiments.}

We consider estimating the quantile of an exponential distribution with parameter $\lambda=2$. 
The algorithmic configurations for SAA, RM-SA and PR-SA with and without adaptive IS are the same as in the normal distribution example, and the IS parameter is given by \eqref{eq:exp_I_function}.
Similarly we set $p=0.99,0.999,0.9999$, and vary the total number of simulation samples from $500$ to $500\times 2^8$ to estimate the quantiles. We repeat the procedure $200$ times to calculate the variance and MSE of the estimated quantiles. Figures \ref{fig:exp_var}-\ref{fig:exp_mse} show the results, and Tables \ref{tab:exp01}-\ref{tab:exp0001} show their numerical details.
\begin{figure}[H]
\begin{minipage}[t]{0.33\linewidth}
\centering
\includegraphics[width=5.5cm]{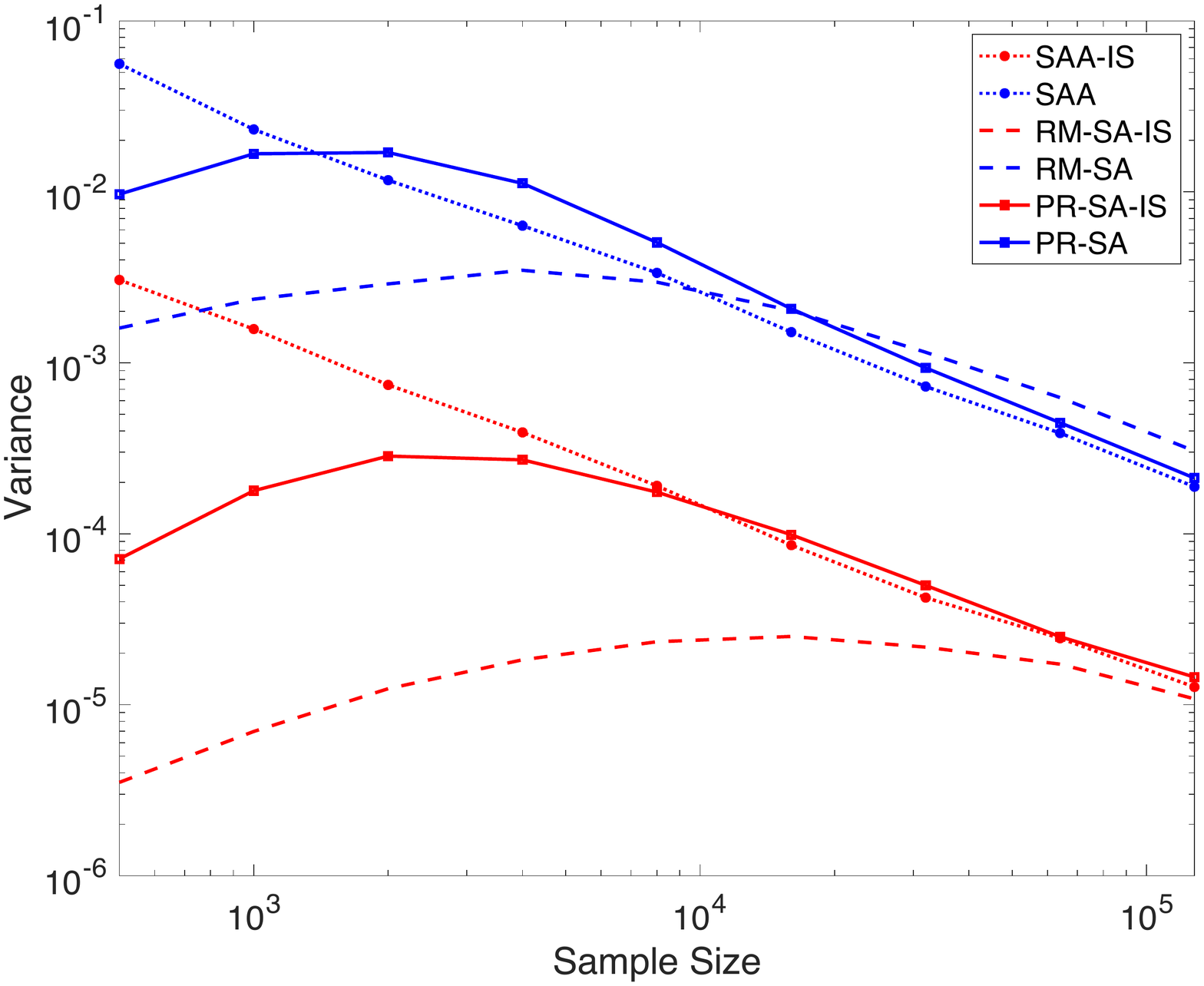}
\end{minipage}
\begin{minipage}[t]{0.33\linewidth}
\centering
\includegraphics[width=5.5cm]{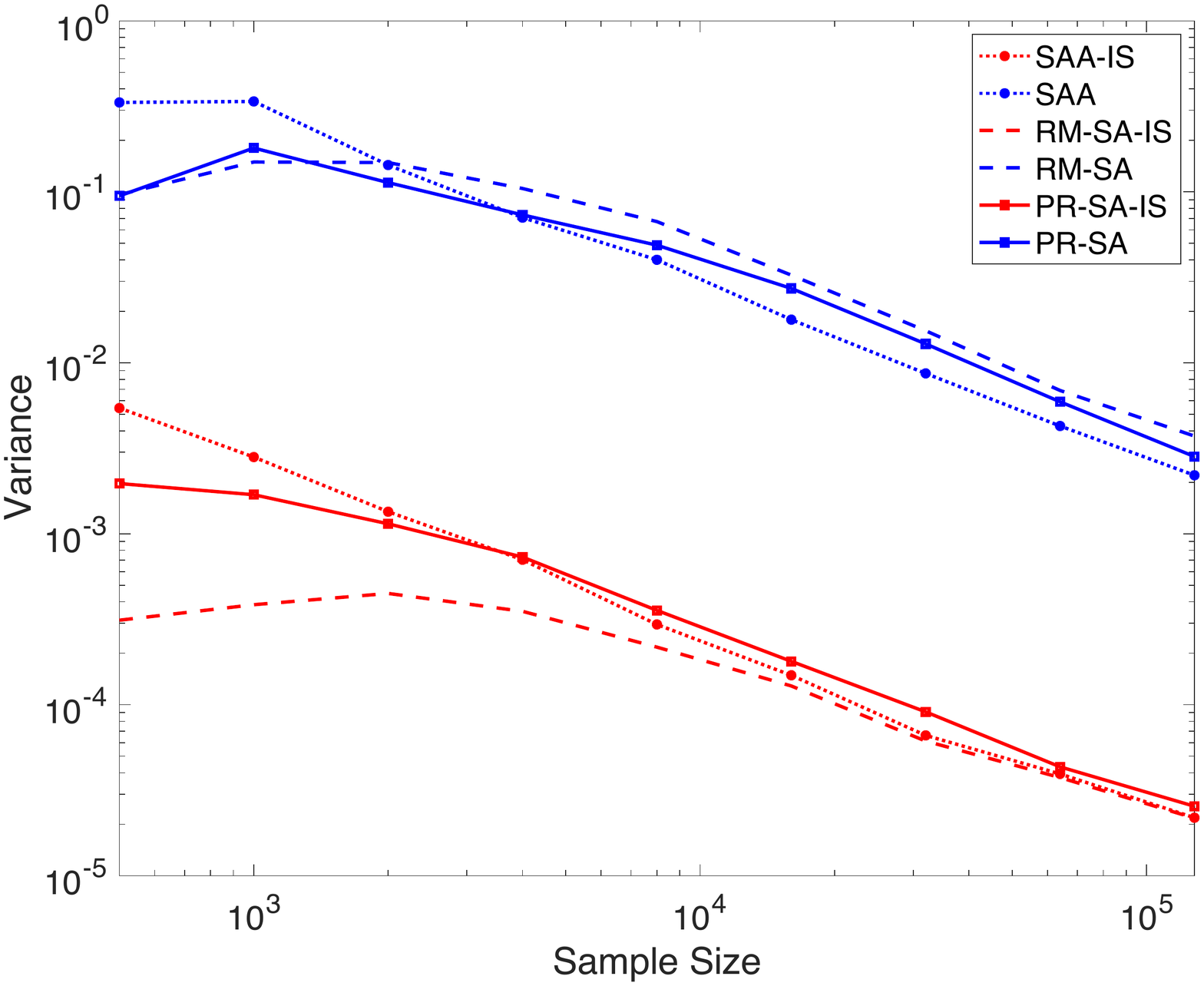}
\end{minipage}
\begin{minipage}[t]{0.32\linewidth}
\centering
\includegraphics[width=5.5cm]{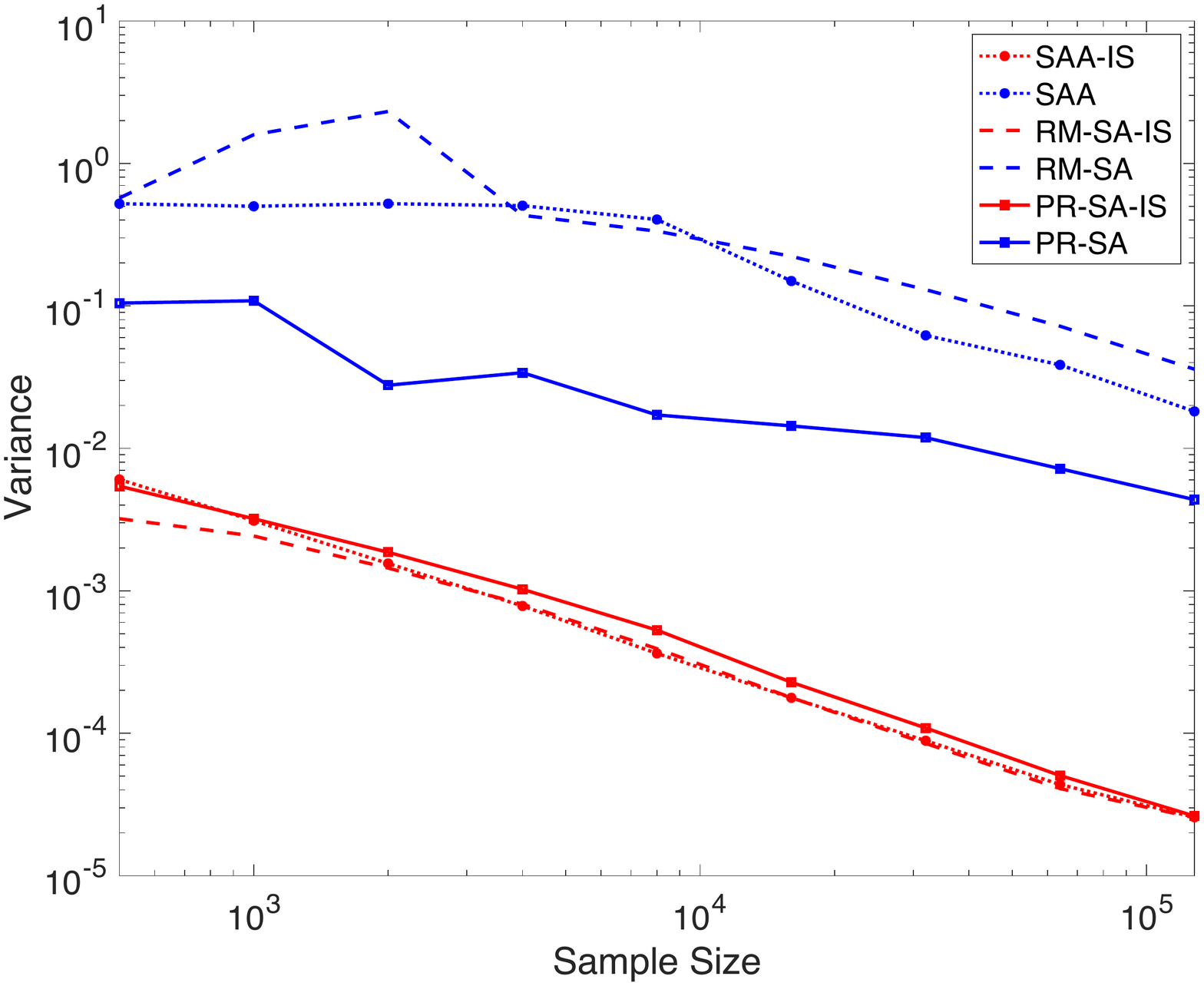}
\end{minipage}
\caption{Variance of SAA, RM-SA, PR-SA, with and without adaptive IS for exponential distribution ($p=0.99$ for the left panel; $p=0.999$ for the middle panel; $p=0.9999$ for the right panel)}
\label{fig:exp_var}
\end{figure}


\begin{figure}[H]
\begin{minipage}[t]{0.33\linewidth}
\centering
\includegraphics[width=5.5cm]{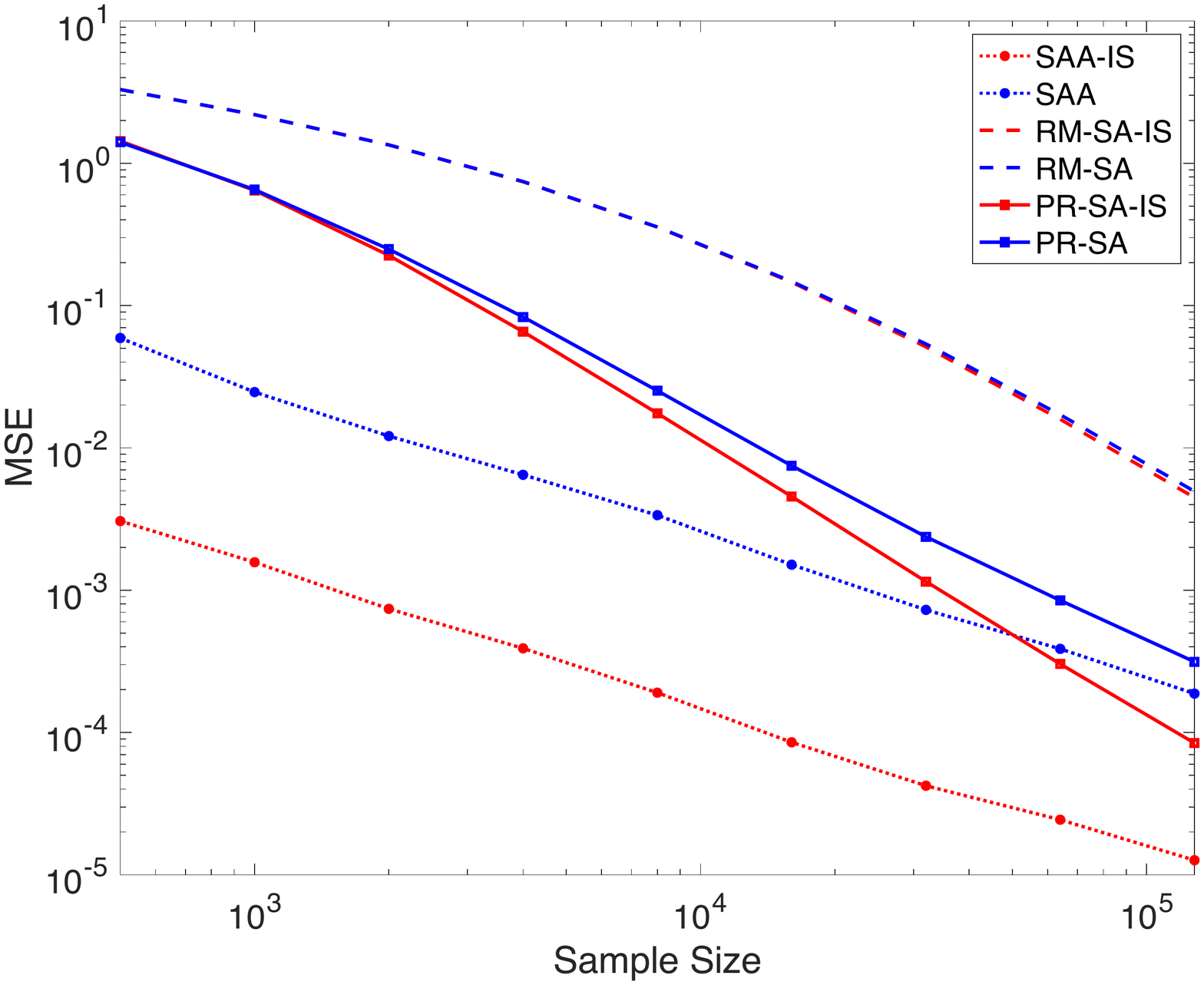}
\end{minipage}
\begin{minipage}[t]{0.33\linewidth}
\centering
\includegraphics[width=5.5cm]{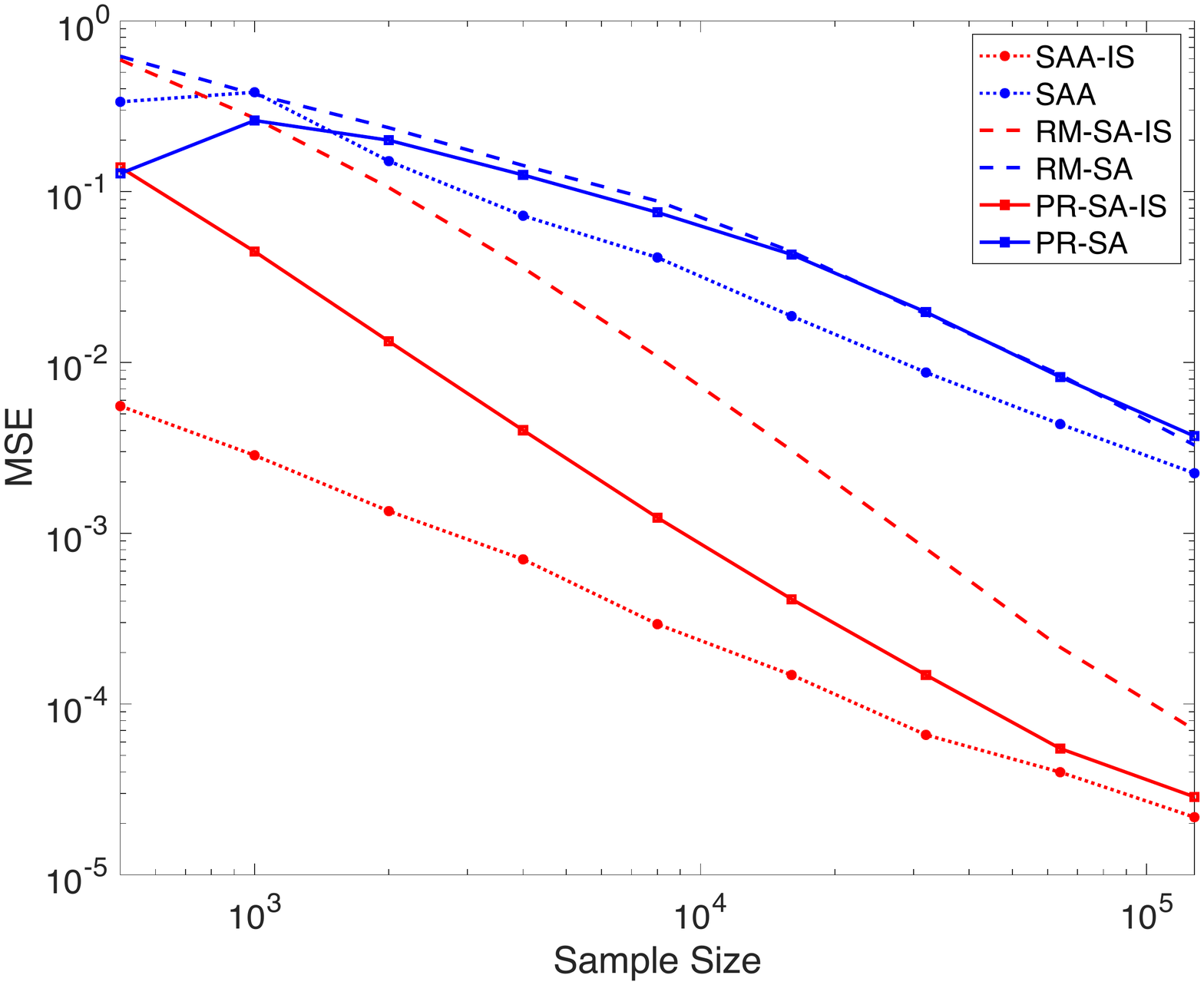}
\end{minipage}
\begin{minipage}[t]{0.32\linewidth}
\centering
\includegraphics[width=5.5cm]{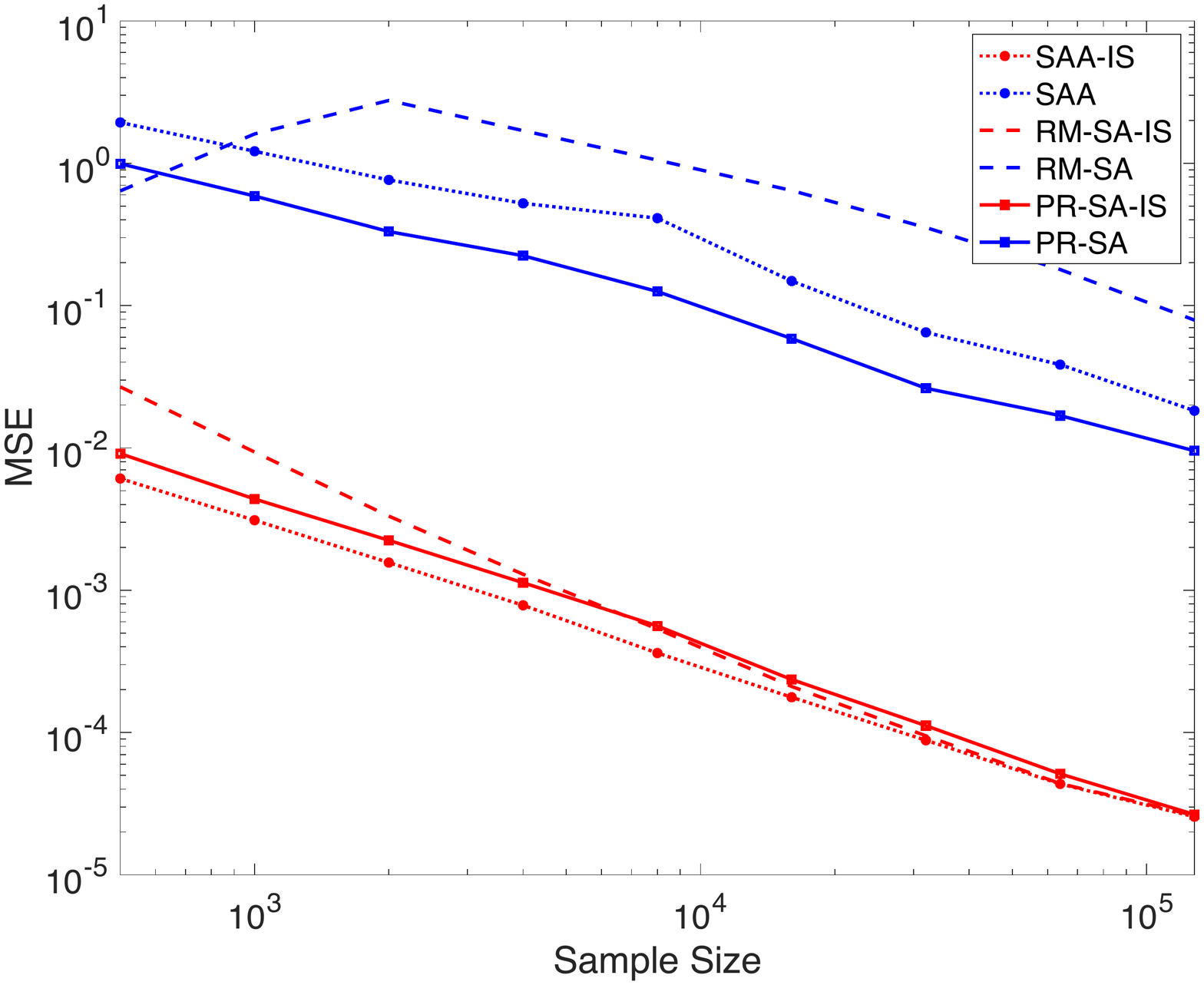}
\end{minipage}
\caption{MSE of SAA, RM-SA, PR-SA, with and without adaptive IS for exponential distribution ($p=0.99$ for the left panel; $p=0.999$ for the middle panel; $p=0.9999$ for the right panel)}
\label{fig:exp_mse}
\end{figure}

\begin{table}[H]
\centering
\caption{Variance of SAA, RM-SA and PR-SA with and without adaptive IS for exponential distribution ($p=0.99$)}
\label{tab:exp01}
\small
\begin{tabular}{c|l l l|l l l|l l l l|}
\toprule
 Sample Size    & SAA-IS   & SAA      & ratio  & RM-SA-IS    & RM-SA       & ratio  & PR-SA-IS   & PR-SA      & ratio \\
 \midrule
500    & 3.05E-03 & 5.60E-02 & 18 & 3.51E-06 & 1.60E-03 & 454 & 7.13E-05 & 9.69E-03 & 136 \\
1000   & 1.58E-03 & 2.32E-02 & 15 & 6.98E-06 & 2.36E-03 & 338 & 1.79E-04 & 1.67E-02 & 94  \\
2000   & 7.43E-04 & 1.17E-02 & 16 & 1.24E-05 & 2.90E-03 & 233 & 2.85E-04 & 1.70E-02 & 60  \\
4000   & 3.92E-04 & 6.34E-03 & 16 & 1.83E-05 & 3.48E-03 & 190 & 2.71E-04 & 1.12E-02 & 41  \\
8000   & 1.91E-04 & 3.36E-03 & 18 & 2.34E-05 & 2.96E-03 & 127 & 1.75E-04 & 5.06E-03 & 29  \\
16000  & 8.58E-05 & 1.51E-03 & 18 & 2.51E-05 & 2.02E-03 & 81  & 9.87E-05 & 2.07E-03 & 21  \\
32000  & 4.23E-05 & 7.27E-04 & 17 & 2.17E-05 & 1.15E-03 & 53  & 4.99E-05 & 9.34E-04 & 18  \\
64000  & 2.45E-05 & 3.88E-04 & 16 & 1.73E-05 & 6.26E-04 & 36  & 2.50E-05 & 4.46E-04 & 18  \\
128000 & 1.27E-05 & 1.89E-04 & 15 & 1.08E-05 & 3.05E-04 & 28  & 1.45E-05 & 2.12E-04 & 15   \\ 
\bottomrule
\end{tabular}
\end{table}

\begin{table}[H]
\centering
\caption{Variance of SAA, RM-SA and PR-SA with and without adaptive IS for exponential distribution ($p=0.999$)}
\label{tab:expl001}
\small
\begin{tabular}{c|l l l|l l l|l l l l|}
\toprule
 Sample Size     & SAA-IS   & SAA      & ratio  & RM-SA-IS    & RM-SA       & ratio  & PR-SA-IS   & PR-SA      & ratio  \\  
\midrule
500    & 5.44E-03 & 3.33E-01 & 61  & 3.12E-04 & 9.63E-02 & 308 & 1.97E-03 & 9.47E-02 & 48  \\
1000   & 2.81E-03 & 3.38E-01 & 120 & 3.85E-04 & 1.49E-01 & 389 & 1.69E-03 & 1.81E-01 & 107 \\
2000   & 1.35E-03 & 1.43E-01 & 107 & 4.48E-04 & 1.48E-01 & 331 & 1.14E-03 & 1.13E-01 & 99  \\
4000   & 7.05E-04 & 7.07E-02 & 100 & 3.52E-04 & 1.05E-01 & 297 & 7.31E-04 & 7.33E-02 & 100 \\
8000   & 2.95E-04 & 4.01E-02 & 136 & 2.17E-04 & 6.70E-02 & 308 & 3.56E-04 & 4.88E-02 & 137 \\
16000  & 1.49E-04 & 1.79E-02 & 120 & 1.29E-04 & 3.26E-02 & 253 & 1.79E-04 & 2.72E-02 & 152 \\
32000  & 6.62E-05 & 8.67E-03 & 131 & 6.13E-05 & 1.54E-02 & 251 & 9.08E-05 & 1.29E-02 & 142 \\
64000  & 3.94E-05 & 4.27E-03 & 108 & 3.75E-05 & 6.89E-03 & 184 & 4.33E-05 & 5.91E-03 & 137 \\
128000 & 2.18E-05 & 2.20E-03 & 101 & 2.17E-05 & 3.73E-03 & 172 & 2.55E-05 & 2.83E-03 & 111  \\ 
 \bottomrule
\end{tabular}
\end{table}

\begin{table}[H]
\centering
\caption{Variance of SAA, RM-SA and PR-SA with and without adaptive IS for exponential distribution ($p=0.9999$)}
\label{tab:exp0001}
\small
\begin{tabular}{c|l l l|l l l|l l l l|}
\toprule
 Sample Size     & SAA-IS   & SAA      & ratio  & RM-SA-IS    & RM-SA       & ratio  & PR-SA-IS   & PR-SA      & ratio  \\  
 \midrule
 500    & 6.04E-03 & 5.21E-01 & 86   & 3.21E-03 & 5.71E-01 & 178  & 5.42E-03 & 1.05E-01 & 19  \\
1000   & 3.10E-03 & 5.00E-01 & 161  & 2.42E-03 & 1.59E+00 & 656 & 3.20E-03 & 1.09E-01 & 34  \\
2000   & 1.55E-03 & 5.21E-01 & 335  & 1.44E-03 & 2.31E+00 & 1603 & 1.86E-03 & 2.77E-02 & 15  \\
4000   & 7.82E-04 & 5.04E-01 & 645  & 8.05E-04 & 4.32E-01 & 537  & 1.02E-03 & 3.39E-02 & 33  \\
8000   & 3.63E-04 & 4.04E-01 & 1111 & 3.92E-04 & 3.34E-01 & 852  & 5.29E-04 & 1.71E-02 & 32  \\
16000  & 1.77E-04 & 1.49E-01 & 843  & 1.78E-04 & 2.23E-01 & 1253 & 2.28E-04 & 1.44E-02 & 63  \\
32000  & 8.86E-05 & 6.20E-02 & 700  & 8.48E-05 & 1.30E-01 & 1532 & 1.09E-04 & 1.19E-02 & 109 \\
64000  & 4.37E-05 & 3.85E-02 & 881  & 4.09E-05 & 7.19E-02 & 1760 & 5.05E-05 & 7.29E-03 & 145 \\
128000 & 2.57E-05 & 1.81E-02 & 706  & 2.55E-05 & 3.59E-02 & 1406 & 2.62E-05 & 4.36E-03 & 166 \\ 
 \bottomrule
\end{tabular}
\end{table}

The variance reduction effect can be seen clearly by observing that the red curves in Figure \ref{fig:exp_var} are significantly below the blue curves. Also, the variance reduction ratio grows quickly as $p$ goes to 1. For example, 
in Tables \ref{tab:exp01}-\ref{tab:exp0001} where $p$ takes 0.99, 0.999, and 0.9999, respectively, fixing the sample size 128000, the variance reduction ratio for SAA-IS takes 15, 101, and 706, respectively.  

\subsection{Pareto-tailed Distribution}

We consider a Pareto-tailed distribution. By using IS, we change $P(Z\geq x)=x^{-\lambda},x\geq1$ to $P(Z\geq x)=x^{-\alpha},x\geq1$. 
As in the previous examples, we will use the upper-side estimator (i.e., $\mathbf{1}\{Z\geq q\}\ell(Z,\alpha)$). Suppose that the original distribution has parameter $\lambda$. Given $\hat{q}_{n-1}$, we will find a new measure $\alpha_n=I(\hat{q}_{n-1})$ such that
the variance of $\mathbf{1}\{Z\geq\hat{q}_{n-1}\}\ell(Z,\alpha)$ is minimized. As in the last subsection, it suffices to minimize the second moment.
The second moment is given by 
\[
\mathds{E}_{Z\sim P_{\alpha}}[\mathbf{1}\{Z\geq \hat q_{n-1}\}(\ell(Z,\alpha))^2]=\int_{\hat{q}_{n-1}}^{\infty}\alpha x^{-\alpha-1}\frac{\lambda^{2}x^{-2\lambda-2}}{\alpha^{2}x^{-2\alpha-2}}dx=\frac{\lambda^{2}}{\alpha}\frac{\hat{q}_{n-1}^{-2\lambda+\alpha}}{2\lambda-\alpha}.
\]
Then we need to find $\alpha\in(0,2\lambda)$ to minimize ${\hat{q}_{n-1}^{\alpha}}/{(\alpha(2\lambda-\alpha))}$.
With some computations, we can show that 
\begin{equation}\label{eq:paro_I_function}
I(\hat{q}_{n-1})=\frac{1+\lambda\log\hat{q}_{n-1}-\sqrt{1+\lambda^{2}\log^{2}\hat{q}_{n-1}}}{\log\hat{q}_{n-1}}.
\end{equation}

\subsubsection{Verification of Assumptions.}
\indent{\bf For (SAA1)}, if we let $\alpha_n = I(\hat{q}_{n-1})$, we have that for any $\bar{q}\in[q-\delta,q+\delta]$,
\begin{align*}
\mathds{E}\left[\mathbf{1}\{Z_{n}>\bar{q}\}\left(\ell(Z_{n},\alpha_{n})\right)^{2} \right]& =\mathds{E}\left[\frac{\lambda^{2}}{\alpha_{n}}\frac{\bar{q}^{-2\lambda+\alpha_{n}}}{2\lambda-\alpha_{n}}\right].
\end{align*}
We need this to be $O(n^{{1}/{2}-\varepsilon})$. Since we always have
that $I(\hat{q}_{n-1})\in(0,\lambda)$, the exponential term in
the numerator is always bounded in interval $[\bar{q}^{-2\lambda},\bar{q}^{-\lambda}]$,
which will not affect the asymptotic rate. Now we only need to guarantee that the denominator cannot be too close to 0.

From the expression of $I(q)$, we know that it is nonincreasing
and $I(0+)=\lambda,I(\infty)=0$. Hence $(2\lambda-\alpha_{n})$
is also bounded away from 0. To bound the reciprocal of $\alpha_{n}$, we 
will need that the truncation set $A_n$ satisfies ${1}/{\alpha}=O(n^{1-\epsilon})$ uniformly for $\alpha\in A_n$. 

{\bf For (SAA2) },
\[
\ell(x,\alpha)=\frac{\lambda x^{-\lambda-1}}{\alpha x^{-\alpha-1}}.
\]
Since $\alpha^*\in(0,\lambda)$, we have that ${x^{-\lambda-1}}/{x^{-\alpha-1}}$
is always bounded for $\alpha$ in a neighborhood of $\alpha^*$. Then following similar discussion as in the previous
subsection, we can verify these assumptions.

{\bf For (SA1)}, as in the previous example, this follows from the minimizing property.

\subsubsection{Asymptotic Results.}

We have the second moment of the estimator
\[
\mathds{E}_{n-1}[\mathbf{1}\{Z_n\geq\hat{q}_{n-1}\}(\ell(Z_n,\alpha_n))^{2}]=\frac{\lambda^{2}}{\alpha(\hat{q}_{n-1})}\frac{{{\hat q}_{n-1}}^{-2\lambda+\alpha(\hat{q}_{n-1})}}{2\lambda-\alpha(\hat{q}_{n-1})}.
\]

From this, for both SAA and PR-SA, we will reach an asymptotic
variance of (similar to the computation in the previous example)
\[
\frac{1}{\lambda^{2}\left(q^{*}\right)^{-2\lambda-2}}\left(\frac{\lambda^{2}}{\alpha(q^{*})}\frac{{{(q^*)}}^{-2\lambda+\alpha(q^{*})}}{2\lambda-\alpha(q^{*})}-\left(q^{*}\right)^{-2\lambda}\right)=\left(q^{*}\right)^{2}\left(\frac{1}{\alpha(q^{*})}\frac{{{(q^*)}}^{\alpha(q^{*})}}{2\lambda-\alpha(q^{*})}-\frac{1}{\lambda^{2}}\right).
\]
Here, $I(q)={(1+\lambda\log q-\sqrt{1+\lambda^{2}(\log q)^2})}/{\log q}$. This means we have the CLT (for PR-SA, replace $\hat{q}_{n}$
with $\bar{q}_{n}$ here)
\[
\sqrt{n}(\hat{q}_{n}-q)\Rightarrow \mathcal{N}\left(0,\left(q^{*}\right)^{2}\left(\frac{1}{I(q^{*})}\frac{(q^*)^{I(q^{*})}}{2\lambda-I(q^{*})}-\frac{1}{\lambda^{2}}\right)\right).
\]
As we can see, as $q^*$ goes to infinity, the variance goes to infinity.
But the speed is much slower than ${{(q^*)}}^{\lambda}$.

For RM-SA stepsize $\gamma_{n}={\gamma}/{n}$,
the asymptotic result would be 
\[
\sqrt{n}(\hat{q}_{n}-q)\Rightarrow \mathcal{N}\left(0,\frac{\gamma^{2}}{2\gamma-\lambda\left(q^{*}\right)^{-\lambda-1}}\left(\frac{\lambda^{2}}{I(q^{*})}\frac{{{(q^*)}}^{-2\lambda+I(q^{*})}}{2\lambda-I(q^{*})}-\left(q^{*}\right)^{-2\lambda}\right)\right).
\]

\subsubsection{Numerical Experiments.}
We consider a Pareto-tailed distribution
\begin{equation}
\Pr\{X\geq x\} = x^{-\lambda}
\end{equation} 
with $\lambda=2$.
We change the parameter $\lambda$ to $\alpha$ in IS. Specifically, the optimal IS parameter is given by \eqref{eq:paro_I_function}.
Similar to the previous examples, we set $p=0.99,0.999,0.9999$, and vary the total number of simulation samples from $500$ to $500\times 2^8$ to estimate the quantiles. We repeat the procedure $200$ times to calculate the variance and MSE of the estimated quantiles.
Figures \ref{fig:tail_var}-\ref{fig:tail_mse} show the results, and Tables \ref{tab:tail01}-\ref{tab:tail0001} further show their numerical details.
\begin{figure}[H]
\begin{minipage}[t]{0.33\linewidth}
\centering
\includegraphics[width=5.5cm]{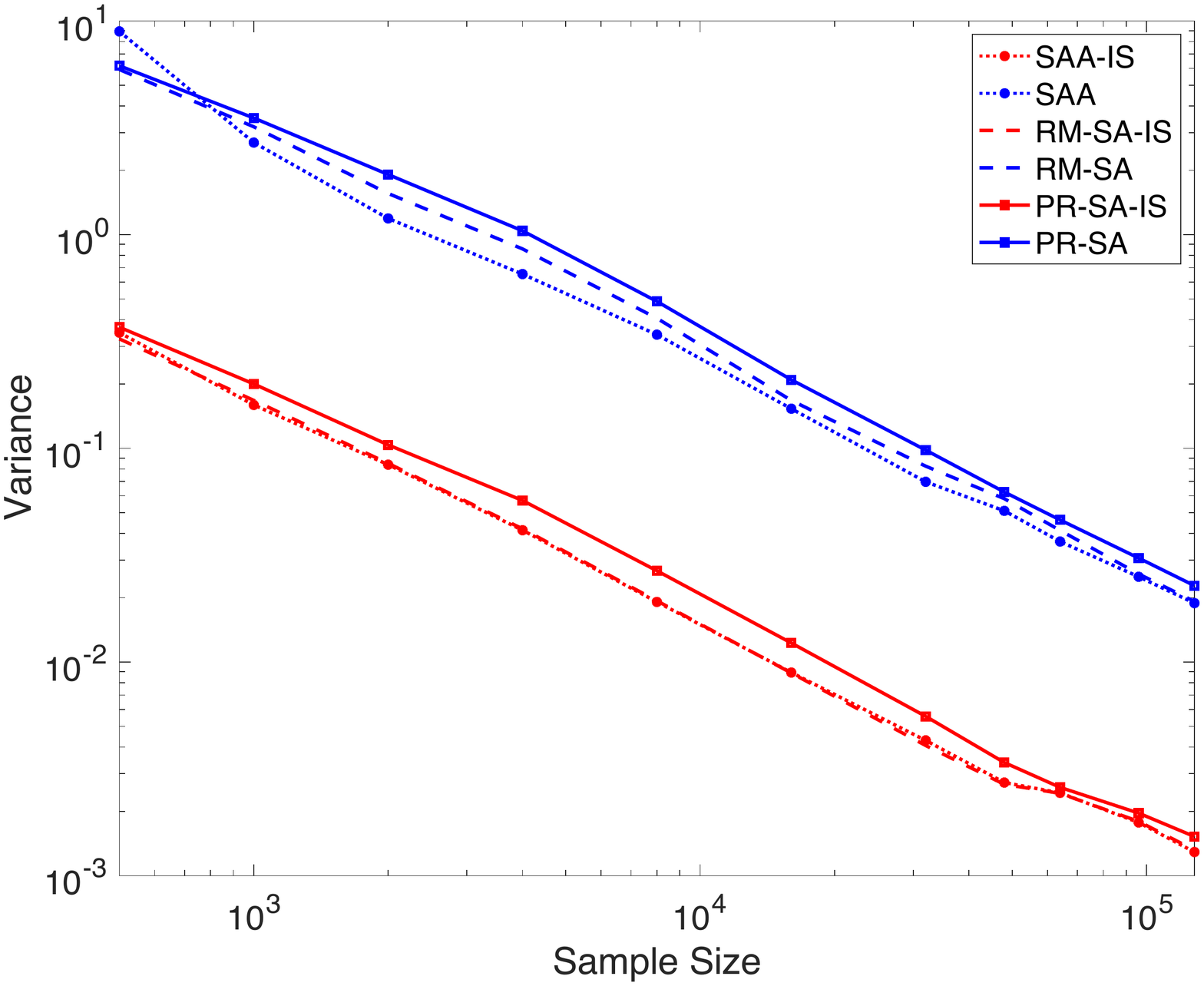}
\end{minipage}
\begin{minipage}[t]{0.33\linewidth}
\centering
\includegraphics[width=5.5cm]{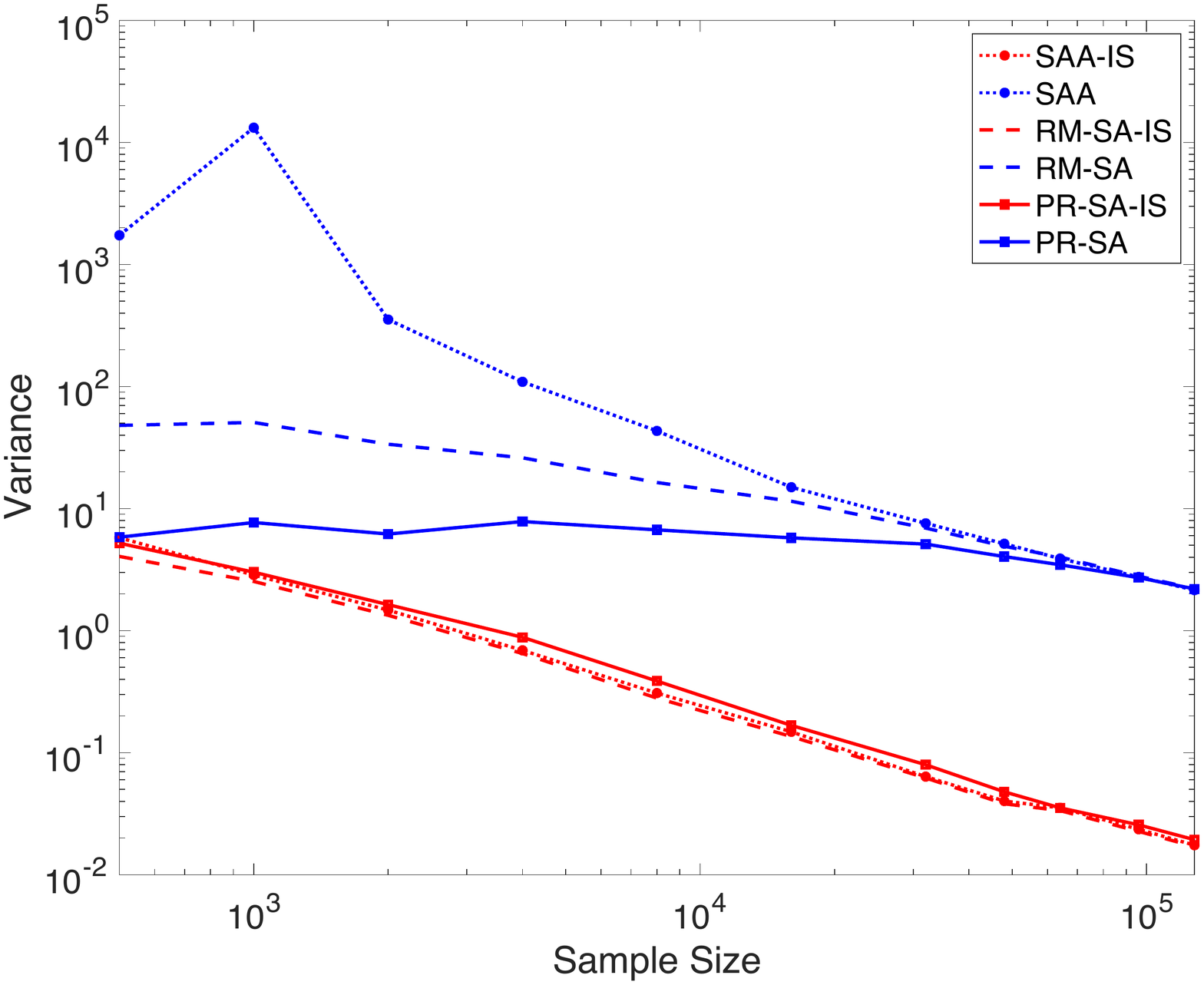}
\end{minipage}
\begin{minipage}[t]{0.32\linewidth}
\centering
\includegraphics[width=5.5cm]{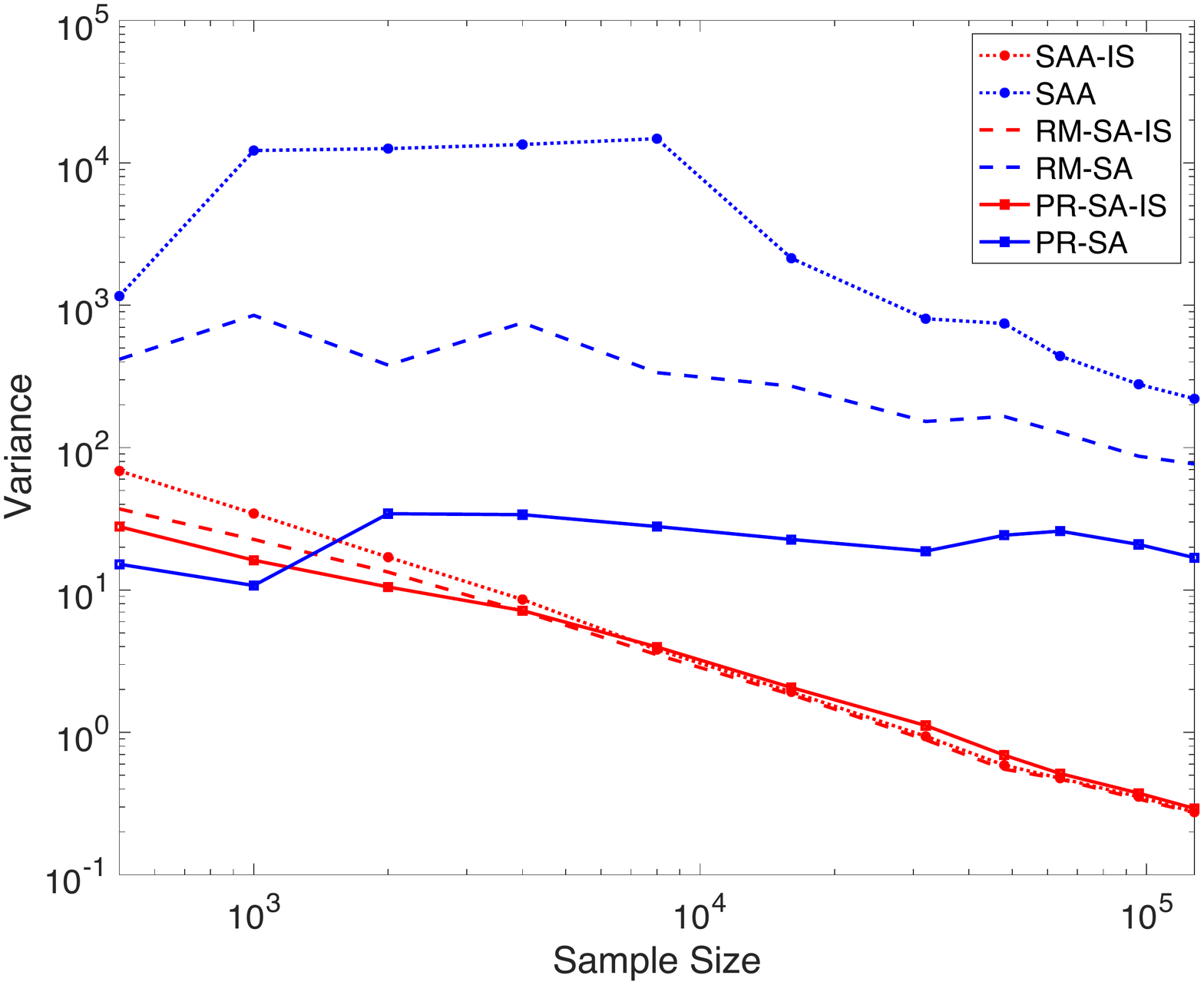}
\end{minipage}
\caption{Variance of SAA, RM-SA, PR-SA, with and without adaptive IS for Pareto-tailed distribution ($p=0.99$ for the left panel; $p=0.999$ for the middle panel; $p=0.9999$ for the right panel)}
\label{fig:tail_var}
\end{figure}


\begin{figure}[H]
\begin{minipage}[t]{0.33\linewidth}
\centering
\includegraphics[width=5.5cm]{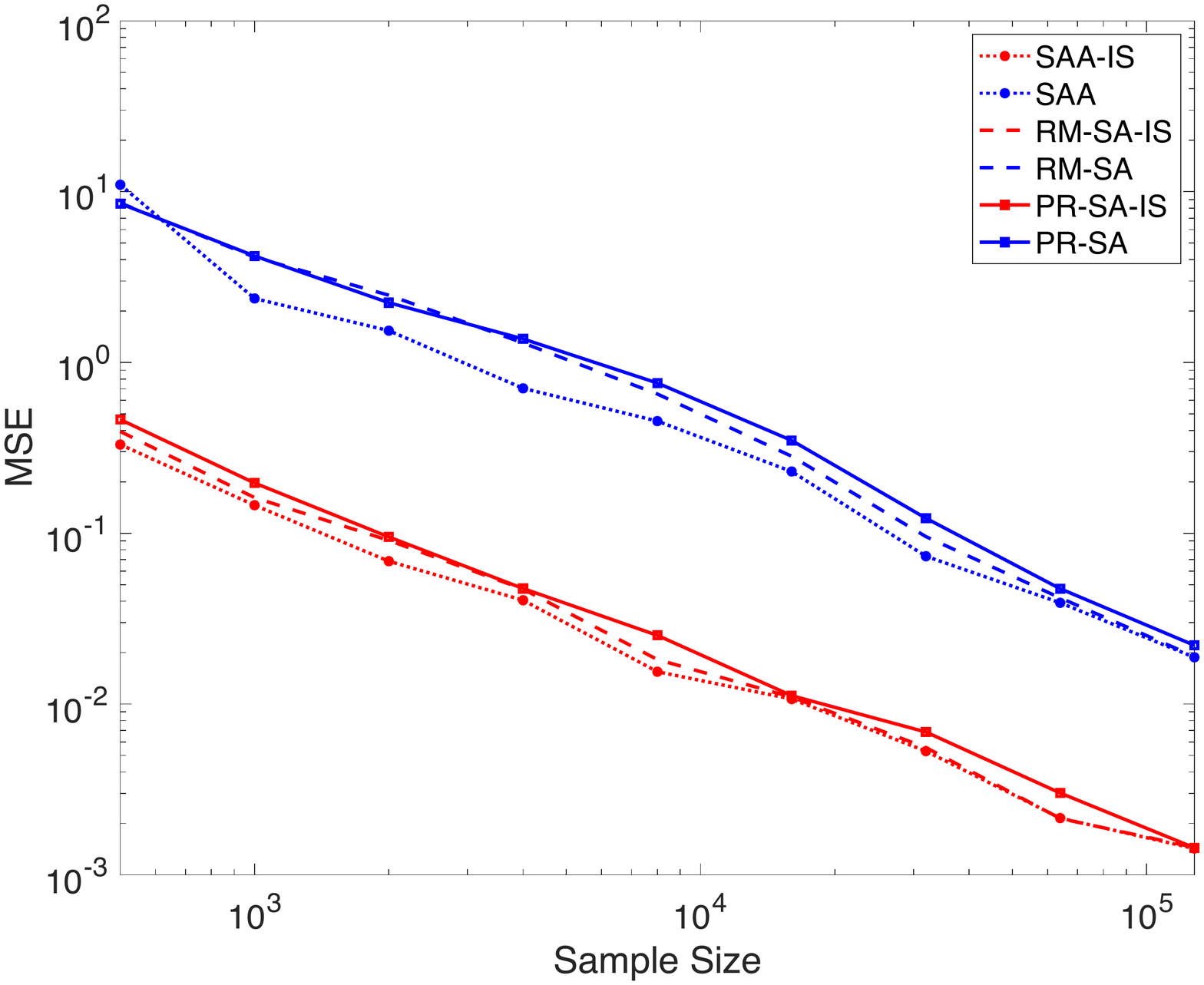}
\end{minipage}
\begin{minipage}[t]{0.33\linewidth}
\centering
\includegraphics[width=5.5cm]{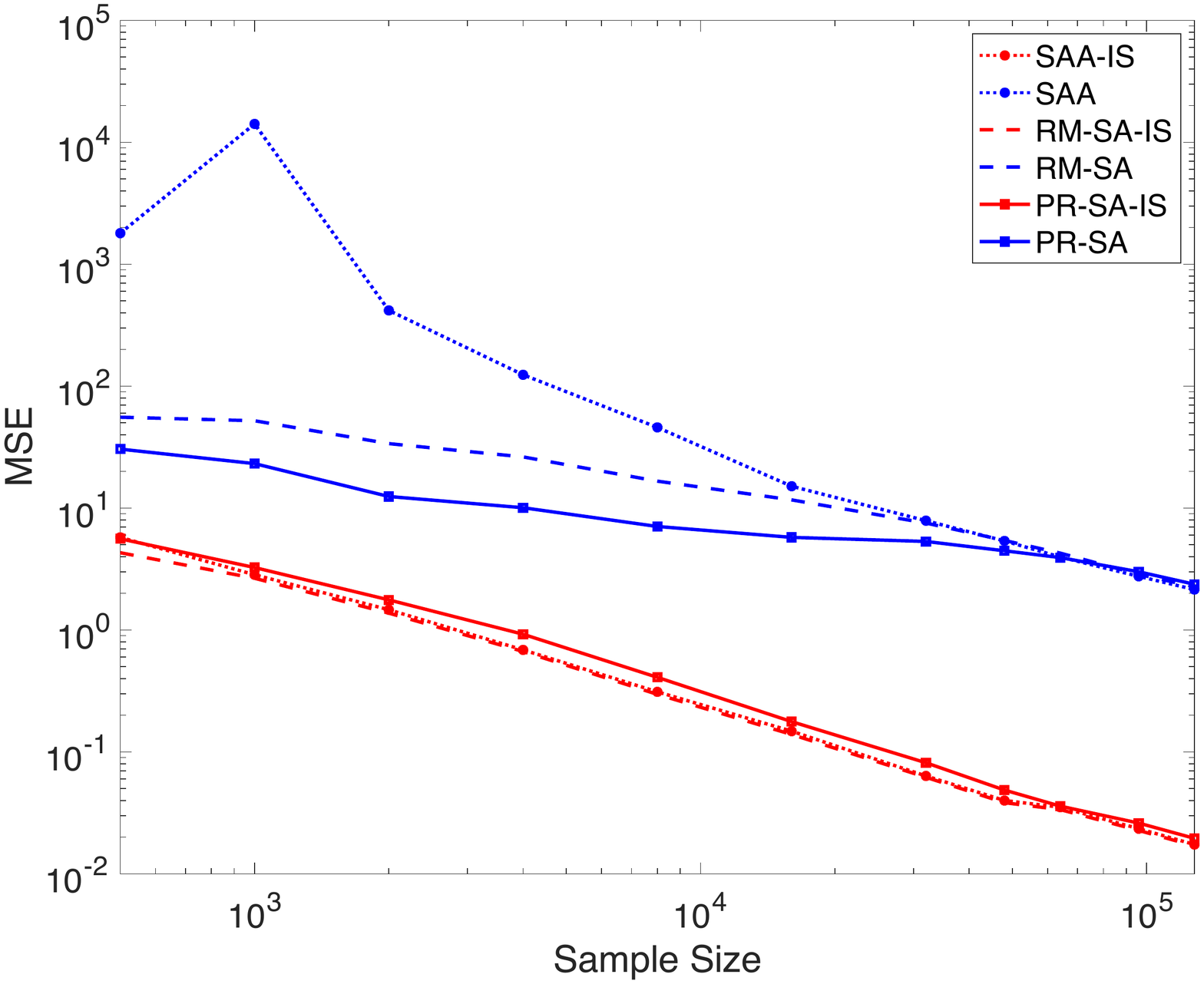}
\end{minipage}
\begin{minipage}[t]{0.32\linewidth}
\centering
\includegraphics[width=5.5cm]{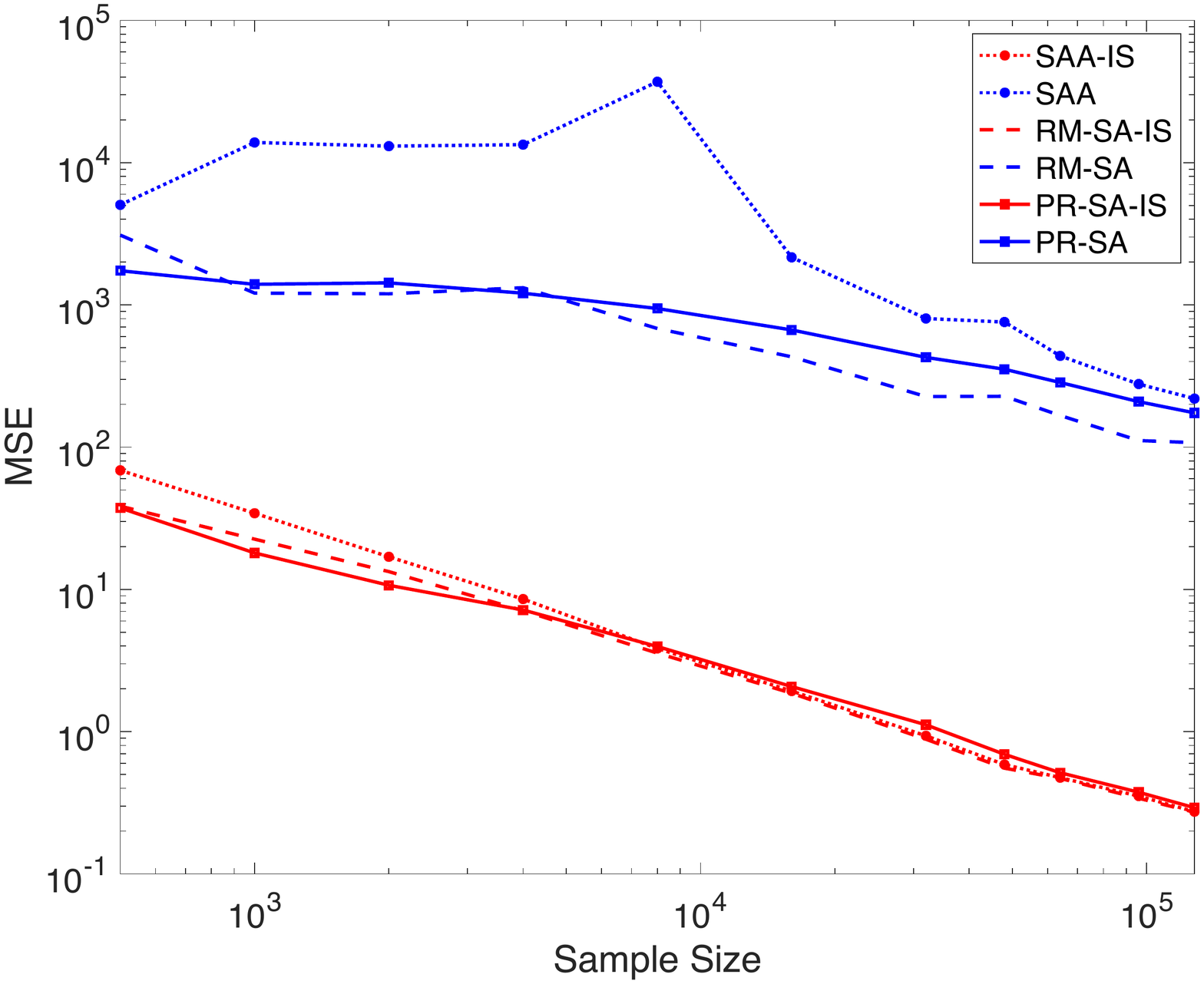}
\end{minipage}
\caption{MSE of SAA, RM-SA, PR-SA, with and without adaptive IS for Pareto-tailed distribution ($p=0.99$ for the left panel; $p=0.999$ for the middle panel; $p=0.9999$ for the right panel)}
\label{fig:tail_mse}
\end{figure}

\begin{table}[H]
\centering
\caption{Variance of SAA, RM-SA and PR-SA with and without adaptive IS for Pareto-tailed distribution ($p=0.99$)}
\label{tab:tail01}
\small
\begin{tabular}{c|l l l|l l l|l l l l|}
\toprule
 Sample Size     & SAA-IS   & SAA      & ratio  & RM-SA-IS    & RM-SA       & ratio  & PR-SA-IS   & PR-SA      & ratio  \\  
 \midrule
 500    & 3.48E-01 & 8.93E+00 & 26 & 3.25E-01 & 5.92E+00 & 18 & 3.69E-01 & 6.16E+00 & 17 \\
1000   & 1.60E-01 & 2.70E+00 & 17 & 1.68E-01 & 3.19E+00 & 19 & 2.00E-01 & 3.51E+00 & 18 \\
2000   & 8.39E-02 & 1.19E+00 & 14 & 8.50E-02 & 1.56E+00 & 18 & 1.04E-01 & 1.91E+00 & 18 \\
4000   & 4.14E-02 & 6.53E-01 & 16 & 4.20E-02 & 8.57E-01 & 20 & 5.70E-02 & 1.04E+00 & 18 \\
8000   & 1.91E-02 & 3.40E-01 & 18 & 1.93E-02 & 4.06E-01 & 21 & 2.67E-02 & 4.87E-01 & 18 \\
16000  & 8.93E-03 & 1.53E-01 & 17 & 8.88E-03 & 1.68E-01 & 19 & 1.23E-02 & 2.09E-01 & 17 \\
32000  & 4.30E-03 & 6.97E-02 & 16 & 4.07E-03 & 8.25E-02 & 20 & 5.57E-03 & 9.81E-02 & 18 \\
64000  & 2.44E-03 & 3.66E-02 & 15 & 2.43E-03 & 4.13E-02 & 17 & 2.59E-03 & 4.62E-02 & 18 \\
128000 & 1.29E-03 & 1.89E-02 & 15 & 1.32E-03 & 1.92E-02 & 15 & 1.52E-03 & 2.28E-02 & 15\\
 \bottomrule
\end{tabular}
\end{table}

\begin{table}[H]
\centering
\caption{Variance of SAA, RM-SA and PR-SA with and without adaptive IS for Pareto-tailed distribution ($p=0.999$)}
\label{tab:tail001}
\small
\begin{tabular}{c|l l l|l l l|l l l l|}
\toprule
 Sample Size     & SAA-IS   & SAA      & ratio  & RM-SA-IS    & RM-SA       & ratio  & PR-SA-IS   & PR-SA      & ratio  \\  
 \midrule
 500    & 5.77E+00 & 1.74E+03 & 301  & 4.06E+00 & 4.80E+01 & 12  & 5.20E+00 & 5.83E+00 & 1.1   \\
1000   & 2.85E+00 & 1.32E+04 & 4623 & 2.53E+00 & 5.08E+01 & 20  & 3.01E+00 & 7.69E+00 & 2.6   \\
2000   & 1.47E+00 & 3.54E+02 & 240  & 1.34E+00 & 3.37E+01 & 25  & 1.63E+00 & 6.19E+00 & 3.8   \\
4000   & 6.89E-01 & 1.09E+02 & 159  & 6.49E-01 & 2.61E+01 & 40  & 8.82E-01 & 7.83E+00 & 8.9   \\
8000   & 3.08E-01 & 4.33E+01 & 141  & 2.80E-01 & 1.64E+01 & 58  & 3.86E-01 & 6.70E+00 & 17  \\
16000  & 1.48E-01 & 1.50E+01 & 101  & 1.36E-01 & 1.15E+01 & 85  & 1.67E-01 & 5.75E+00 & 34  \\
32000  & 6.37E-02 & 7.58E+00 & 119  & 6.21E-02 & 6.92E+00 & 112 & 7.97E-02 & 5.12E+00 & 64  \\
64000  & 3.54E-02 & 3.89E+00 & 110  & 3.35E-02 & 4.01E+00 & 120 & 3.53E-02 & 3.47E+00 & 98  \\
128000 & 1.75E-02 & 2.15E+00 & 123  & 1.70E-02 & 2.21E+00 & 130 & 1.94E-02 & 2.20E+00 & 114\\
 \bottomrule
\end{tabular}
\end{table}

\begin{table}[H]
\centering
\caption{Variance of SAA, RM-SA and PR-SA with and without adaptive IS for Pareto-tailed distribution ($p=0.9999$)}
\label{tab:tail0001}
\small
\begin{tabular}{c|l l l|l l l|l l l l|}
\toprule
 Sample Size     & SAA-IS   & SAA      & ratio  & RM-SA-IS    & RM-SA       & ratio  & PR-SA-IS   & PR-SA      & ratio  \\  
 \midrule
500    & 6.86E+01 & 1.16E+03 & 17   & 3.71E+01 & 4.17E+02 & 11  & 2.78E+01 & 1.52E+01 & 0.5  \\
1000   & 3.44E+01 & 1.22E+04 & 354  & 2.26E+01 & 8.49E+02 & 38  & 1.62E+01 & 1.07E+01 & 0.7  \\
2000   & 1.70E+01 & 1.26E+04 & 740  & 1.34E+01 & 3.79E+02 & 28  & 1.05E+01 & 3.43E+01 & 3.3  \\
4000   & 8.57E+00 & 1.34E+04 & 1568 & 7.08E+00 & 7.52E+02 & 106 & 7.16E+00 & 3.38E+01 & 4.7  \\
8000   & 3.81E+00 & 1.48E+04 & 3875 & 3.51E+00 & 3.36E+02 & 96  & 3.98E+00 & 2.79E+01 & 7.0  \\
16000  & 1.92E+00 & 2.14E+03 & 1110 & 1.84E+00 & 2.70E+02 & 147 & 2.06E+00 & 2.26E+01 & 11 \\
32000  & 9.39E-01 & 8.03E+02 & 855  & 8.86E-01 & 1.53E+02 & 172 & 1.12E+00 & 1.87E+01 & 17 \\
64000  & 4.76E-01 & 4.40E+02 & 923  & 4.70E-01 & 1.28E+02 & 272 & 5.14E-01 & 2.59E+01 & 50 \\
128000 & 2.75E-01 & 2.20E+02 & 800  & 2.72E-01 & 7.67E+01 & 281 & 2.92E-01 & 1.69E+01 & 58\\
 \bottomrule
\end{tabular}
\end{table}

As in the previous examples, the variance reduction effect is clearly seen by comparing between the red curves and blue curves in Figure \ref{fig:tail_var}, and there is more variance reduction when $p$ is closer to 1. For example, in Tables \ref{tab:tail01}-\ref{tab:tail0001} where $p$ takes $0.99,0.999$, and $0.9999$, respectively, fixing the sample size as 128000, the variance reduction ratio for SAA-IS takes 15, 123, and 800, respectively.

\end{document}